\documentclass[12pt,a4paper]{article}

\usepackage{amsmath,amscd,amssymb,amsthm,mathrsfs,amsfonts,layout,indentfirst,graphicx,caption,mathabx,  tikz, stmaryrd,appendix,calc,imakeidx,upgreek}
\usepackage{palatino}  
\usepackage{slashed} 
\usepackage{mathrsfs} 
\usepackage{extarrows} 
\makeindex

\usepackage{lipsum}
\let\OLDthebibliography\thebibliography
\renewcommand\thebibliography[1]{
	\OLDthebibliography{#1}
	\setlength{\parskip}{0pt}
	\setlength{\itemsep}{2pt} 
}

\allowdisplaybreaks  
\usepackage{latexsym}
\usepackage{chngcntr}
\usepackage[colorlinks,linkcolor=blue,anchorcolor=blue,linktocpage]{hyperref}
\hypersetup{citecolor=[rgb]{0,0.5,0}}

\usepackage{fullpage}
\counterwithin{figure}{section}

\theoremstyle{definition}
\newtheorem{df}{Definition}[section]

\newtheorem{rem}[df]{Remark}

\newtheorem{cv}[df]{Convention}

\theoremstyle{plain}
\newtheorem{thm}[df]{Theorem}

\newtheorem{pp}[df]{Proposition}
\newtheorem{co}[df]{Corollary}
\newtheorem{lm}[df]{Lemma}

\newcommand{\fk}{\mathfrak}
\newcommand{\mc}{\mathcal}
\newcommand{\wtd}{\widetilde}
\newcommand{\wht}{\widehat}

\newcommand{\ovl}{\overline}

\newcommand{\End}{\mathrm{End}} 
\newcommand{\id}{\mathbf{1}}
\newcommand{\Hom}{\mathrm{Hom}}
\newcommand{\Conf}{\mathrm{Conf}}

\newcommand{\ev}{\mathrm{ev}}
\newcommand{\coev}{\mathrm{coev}}

\newcommand{\diag}{\mathrm{diag}}
\newcommand{\Dom}{\scr D}

\newcommand{\Diffp}{\mathrm{Diff}^+}
\newcommand{\Diff}{\mathrm{Diff}}
\newcommand{\PSU}{\mathrm{PSU}(1,1)}
\newcommand{\UPSU}{\widetilde{\mathrm{PSU}}(1,1)}

\newcommand{\bk}[1]{\langle {#1}\rangle}
\newcommand{\GA}{\mathscr G_{\mathcal A}}

\newcommand{\scr}{\mathscr}

\newcommand{\Jtd}{\widetilde{\mathcal J}}

\newcommand{\Ad}{\mathrm{Ad}}

\newcommand{\im}{\mathbf{i}}
\newcommand{\Co}{\complement}
\newcommand{\RepA}{\mathrm{Rep}(\mc A)}
\newcommand{\RepdA}{\mathrm{Rep}^{\mathrm d}(\mc A)}
\newcommand{\RepfA}{\mathrm{Rep}^{\mathrm f}(\mc A)}

\newcommand{\mbb}{\mathbb}

\newcommand{\Obj}{\mathrm{Obj}}
\newcommand{\Ef}{\mathscr E^{\mathrm f}}
\newcommand{\Ed}{\mathscr E^{\mathrm d}}

\newcommand{\bpr}{{}^\backprime}
\newcommand{\pr}{\mathrm{pr}}

\numberwithin{equation}{section}

\title{Bisognano-Wichmann property for rigid categorical extensions and non-local extensions of conformal nets}
\author{{\sc Bin Gui}
}
\date{}
\begin{document}\sloppy 
	\pagenumbering{arabic}
	\setcounter{section}{-1}

	\maketitle

\begin{abstract}
Given an (irreducible) M\"obius covariant net $\mc A$, we prove a Bisognano-Wichmann  theorem  for its categorical extension $\Ed$ associated to the braided $C^*$-tensor category $\RepdA$ of dualizable (more precisely, ``dualized") M\"obius covariant $\mc A$-modules. As a closely related result, we prove a (modified) Bisognano-Wichmann theorem  for any (possibly) non-local extension of $\mc A$ obtained by a $C^*$-Frobenius algebra $Q$ in $\RepdA$.  As an application, we discuss the relation between the domains of modular operators and the preclosedness of certain unbounded operators in $\Ed$.
\end{abstract}

\tableofcontents
	
\newpage

\section{Introduction}

The notion of categorical extensions of conformal nets was introduced in \cite{Gui21a} to understand the relations between the tensor categories of conformal nets and vertex operator algebras (VOA's). A categorical extension $\scr E$ of an irreducible conformal (or M\"obius covariant) net $\mc A$ is the Haag-Kastler net of bounded charged fields (intertwiners) associated to $\mc A$. $\scr E$ satisfies a list of axioms similar to those of $\mc A$, including, most importantly, the \emph{locality} axiom, which says that bounded charged fields supported in disjoint open intervals commute adjointly. One main observation in \cite{Gui21a} is that, in order to relate the tensor category of a unitary VOA $V$ with the one of the corresponding conformal net $\mc A_V$ (assuming $\mc A_V$ exists), it suffices to show that the (usually) \emph{unbounded} smeared intertwining operators of $V$ give rise to \emph{bounded} intertwiners satisfying the axioms of a categorical extension, especially the locality. Similar to the construction in \cite{CKLW18} of $\mc A_V$ from $V$, proving the locality axiom is the most difficult step, which amounts to proving the strong commutativity of certain adjointly commuting unbounded closed operators.\footnote{Two closed operators $A$ and $B$ on a Hilbert space $\mc H$ are said to commute adjointly if $[A,B]=[A^*,B]=0$ when acting on suitable vectors; they are said to commute strongly if the  von Neumann algebras generated by $A$ and by $B$ commute. Strong commutativity implies adjoint commutativity; the converse may not hold by the famous counterexample of Nelson \cite{Nel59}.}

The Bisognano-Wichmann (B-W) theorem \cite{BW75} is a powerful tool for proving the locality of the conformal net $\mc A_V$ associated to a unitary VOA $V$. In \cite{CKLW18}, Carpi-Kawahigashi-Longo-Weiner used this theorem to show that very often, one only needs the strong commutativity of a small amount of  smeared vertex operators (which ``generate $V$") to prove the strong commutativity of \emph{all} smeared vertex operators supported in disjoint intervals. \footnote{For general quantum field theories, a similar result was proved in \cite{DSW86}.} The main motivation of our present article is to generalize this result to intertwining operators (charged fields) of VOA's. 

Let us first recall the B-W theorem in (algebraic) chiral conformal field theory \cite{BGL93,GF93,FJ96}. Let $\mc A$ be an (irreducible) M\"obius covariant net with vacuum representation $\mc H_0$ and vacuum vector $\Omega$. The representation of $\PSU$ on $\mc H_0$ is denoted by $U$. By Reeh-Schlieder property, $\Omega$ is a cyclic and separating vector of $\mc A(I)$ where $I$ is any open (non-dense non-empty) interval on the unit circle $\mbb S^1$. Thus, one can associate to  the pair $(\mc A(I),\Omega)$ the modular operator $\Delta_I$ and modular conjugation $\fk J_I$ satisfying the Tomita-Takesaki theorem. Now, the B-W theorem for $\mc A$ says that:
\begin{itemize}
\item (Geometric modular theorem) $\Delta_I^{\im t}=\delta_I(-2\pi t)$, where $\delta_I$ is the dilation subgroup of the M\"obius group $\PSU$ associated to the interval $I$ (see section \ref{lb22} for more details).
\item (PCT theorem) The antiunitary map $\Theta:=\fk J_{\mbb S^1_+}$ (which is an involution by Tomita-Takesaki theory) is a PCT operator for $\mc A$, where $\mbb S^1_+$ is the upper semi-circle. More precisely, if we set $\fk r:\mbb S^1\rightarrow\mbb S^1,z\mapsto\ovl z$, then we have $\Theta\mc A(I)\Theta=\mc A(\fk rI)$ and $\Theta U(g)\Theta=U(\fk rg\fk r)$ for any $g\in\PSU$.
\end{itemize}
More generally, one has the B-W theorem for Fermi conformal nets \cite{ALR01,CKL08} and irreducible finite-index non-local extensions of conformal nets \cite{LR04}.

To derive a B-W theorem for categorical extensions of $\mc A$, we first need to define the modular $S$ and $F$ operators for them. Before explaining the definition, we first recall what are categorical extensions.

Let $\mbb S^1_-$ be the lower semi-circle. If $\mc H_i,\mc H_j$ are $\mc A$-modules, then $\mc H_j$ is a left $\mc A(\mbb S^1_+)$ module, and $\mc H_i$ is a right $\mc A(\mbb S^1_+)$ module defined by the action $x\in\mc A(\mbb S^1_+)\mapsto \Theta x^*\Theta$. Then the fusion product $\mc H_i\boxtimes\mc H_j$ is the Connes-Sauvageot relative tensor product of $\mc H_i$ and $\mc H_j$ over $\mc A(\mbb S^1_+)$. $\mc A(\mbb S^1_+)$ and $\mc A(\mbb S^1_-)$ act naturally on $\mc H_i\boxtimes\mc H_j$ by acting respectively on the left and the right components, and can be extended  to a representation of $\mc A$ on $\mc H_i\boxtimes\mc H_j$ using ``path continuations". (See section \ref{lb23} or \cite{Gui21a} chapter 2 for details.) Now, for each  interval $I$, one can define a dense vector space $\mc H_i(I)=\Hom_{\mc A(I')}(\mc H_0,\mc H_i)\Omega$, where $I'$ is the interior of the complement of $I$. $\mc H_j(I)$ is defined similarly. Then we know that $\mc H_i\boxtimes\mc H_j$ has a dense subspace spanned by vectors of the form $\xi\otimes\eta$ where $\xi\in\mc H_i(\mbb S^1_+)$ and $\eta\in\mc H_j(\mbb S^1_-)$. We then have bounded operators
\begin{align*}
L(\xi)\in\Hom_{\mc A(\mbb S^1_-)}(\mc H_j,\mc H_i\boxtimes\mc H_j),\qquad R(\eta)\in\Hom_{\mc A(\mbb S^1_+)}(\mc H_i,\mc H_i\boxtimes\mc H_j)
\end{align*}
defined by $L(\xi)\phi=\xi\otimes\phi$ and $R(\eta)\psi=\psi\otimes\eta$ for any $\phi\in\mc H_j(\mbb S^1_-),\psi\in\mc H_i(\mbb S^1_+)$. We understand $L(\xi),R(\eta) $ as operators acting on any possible $\mc A$-modules. This means that when $\chi\in\mc H_k$, we have $L(\xi)\chi\in\mc H_i\boxtimes\mc H_k,R(\eta)\chi\in\mc H_k\boxtimes\mc H_j$.

The $L$ and $R$ operators defined above should be understood as supported in $\mbb S^1_+$ and $\mbb S^1_-$ respectively. We would like to have them supported in any interval $I$, so that we have  nets of sets of $L$ operators and $R$ operators. It turns out that in general,  such nets can be defined not on $\mbb S^1$ but on its universal cover. So one should consider the $L$ and $R$ operators localized not in intervals, but in arg-valued intervals. If $I$ is an interval of $\mbb S^1$, then one can choose a continuous argument function $\arg_I$. Then the pair $\wtd I=(I,\arg_I)$ is called an arg-valued interval. We choose $\wtd{\mbb S^1_+}$ and $\wtd{\mbb S^1_-}$ such that $\arg_{\mbb S^1_+}(e^{\im t})=t$ ($0<t<\pi$), and that $\arg_{\mbb S^1_-}(e^{\im t})=t$ ($-\pi<t<0$). Then one can define consistently the $L$ and $R$ operators  localized in any given arg-valued interval $\wtd I$. To be more precise, for any $\mc A$-modules $\mc H_i,\mc H_k$ and any $\xi\in\mc H_i(I)$, one can define
\begin{align*}
L(\xi,\wtd I)\in\Hom_{\mc A(I')}(\mc H_k,\mc H_i\boxtimes\mc H_k),\qquad R(\xi,\wtd I)\in\Hom_{\mc A(I')}(\mc H_k,\mc H_k\boxtimes\mc H_i).
\end{align*}
Moreover, when $\wtd I=\wtd{\mbb S^1_+}$ we have $L(\xi,\wtd I)=L(\xi)$; when $\wtd I=\wtd{\mbb S^1_-}$ we have $R(\xi,\wtd I)=R(\xi)$. These $L$ and $R$ operators form a categorical extension of $\mc A$.

We now focus on dualizable $\mc A$-modules $\mc H_i,\mc H_j,\mc H_k$, etc. Since $\mc H_i$ is dualizable, we have an $\mc A$-module $\mc H_{\ovl i}$ (the dual object) and evaluations $\ev_{i,\ovl i}\in\Hom_{\mc A}(\mc H_i\boxtimes\mc H_{\ovl i},\mc H_0)$ and $\ev_{\ovl i,i}\in\Hom_{\mc A}(\mc H_{\ovl i}\boxtimes\mc H_i,\mc H_0)$  satisfying the conjugate equations
\begin{gather*}
(\ev_{i,\ovl i}\otimes\id_i)(\id_i\otimes\coev_{\ovl i,i})=\id_i=(\id_i\otimes\ev_{\ovl i,i})(\coev_{i,\ovl i}\otimes\id_i),\\
(\ev_{\ovl i,i}\otimes\id_{\ovl i})(\id_{\ovl i}\otimes\coev_{i,\ovl i})=\id_{\ovl i}=(\id_{\ovl i}\otimes\ev_{i,\ovl i})(\coev_{\ovl i,i}\otimes\id_{\ovl i}),
\end{gather*}
where we set $\coev_{i,\ovl i}=\ev_{i,\ovl i}^*,\coev_{\ovl i,i}=\ev_{\ovl i,i}^*$. $\ev_{i,\ovl i}$ determines the other $\ev$ and $\coev$. Moreover, we may and do assume that the $\ev$ and $\coev$ are standard, which means $\ev_{i,\ovl i}(F\otimes\id_{\ovl i})\coev_{i,\ovl i}=\ev_{\ovl i,i}(\id_{\ovl i}\otimes F)\coev_{\ovl i,i}$ for each $F\in\End_{\mc A}(\mc H_i)$. (Cf. \cite{LR97,Yam04,BDH14}.) We say that $(\mc H_{\ovl i},\ev_{i,\ovl i})$ is a dualizing data, and that a dualizable object $\mc H_i$ equipped with a dualizing data is called \textbf{dualized}. Dualized representations with M\"obious covariance form a rigid braided $C^*$-tensor category $\RepdA$, and for each dualized $\mc H_i$ we can choose a canonical dual object $\mc H_{\ovl i}$, such that $\mc H_i$ is also the dual object of $\mc H_{\ovl i}$. (See Section \ref{lb30} for details.) 

Now, for any $\wtd I$, we define the categorical $S$ and $F$ operators $S_{\wtd I},F_{\wtd I}$. For any dualizable $\mc H_i$, we have 
\begin{align*}
S_{\wtd I},F_{\wtd I}:\mc H_i\rightarrow\mc H_{\ovl i}
\end{align*} 
with common domain $\mc H_i(I)$ defined by
\begin{equation*}
S_{\wtd I}\xi=L(\xi,\wtd I)^*\coev_{i,\ovl i}\Omega,\qquad F_{\wtd I}\xi=R(\xi,\wtd I)^*\coev_{\ovl i,i}\Omega.
\end{equation*}

These two operators are indeed preclosed. Moreover, they are related by the (unitary) twist operator $\vartheta$ (proposition \ref{lb28}):
\begin{align*}
F_{\wtd I}=\vartheta S_{\wtd I}.
\end{align*}
We can thus define the modular operator $\Delta_{\wtd I}$ and modular conjugation $\fk J_{\wtd I}$ by the polar decompositions:
\begin{align*}
S_{\wtd I}=\fk J_{\wtd I}\cdot \Delta_{\wtd I}^{\frac 12},\qquad F_{\wtd I}=\vartheta\fk J_{\wtd I}\cdot\Delta_{\wtd I}^{\frac 12},
\end{align*}
where, for each $\mc H_i$, $\Delta_{\wtd I}$ is a positive closed operator on $\mc H_i$, and $\fk J_{\wtd I}:\mc H_i\rightarrow\mc H_{\ovl i}$ is antiunitary. Indeed, $\fk J_{\wtd I}:\mc H_i\rightarrow\mc H_{\ovl i}$ is an involution, i.e., $\fk J_{\wtd I}^2=\id$. It turns out that  \emph{$\fk J_{\wtd I}$ implements the conjugations of morphisms}: for any morphism $G\in\Hom_{\mc A}(\mc H_i,\mc H_j)$ one can define its conjugate $\ovl G\in\Hom_{\mc A}(\mc H_{\ovl i},\mc H_{\ovl i})$ to be the adjoint of the transpose $G^\vee$, where $G^\vee$ is the unique morphism in $\Hom_{\mc A}(\mc H_{\ovl j},\mc H_{\ovl i})$ satisfying
\begin{align*}
\ev_{j,\ovl j}(G\otimes\id_{\ovl j})=\ev_{i,\ovl i}(\id_i\otimes G^\vee).
\end{align*}
We will prove that
\begin{align*}
\ovl G=\fk J_{\wtd I}\cdot G\cdot \fk J_{\wtd I}
\end{align*}
(see Proposition \ref{lb14}), which suggests that $\fk J_{\wtd I}$ is the correct modular conjugation for categorical extensions.

The action of $\Delta_{\wtd I}$ on any $\mc H_i$ can be interpreted as a Connes spatial derivative (see remark \ref{lb24}). Indeed, our definition and treatment of $S_{\wtd I}$ and $F_{\wtd I}$ are deeply motivated by the matrix algebra approach to Connes fusion products and Connes spatial derivatives in \cite{Fal00} and \cite[Sec. IX.3]{Tak02}. Those matrices of von Neumann algebras are described in our article by the $C^*$-Frobenius algebra $Q=(\mc H_k\boxtimes\mc H_{\ovl k},\mu,\iota)$ in the representation category $\RepdA$ of dualized M\"obius covariant $\mc A$-modules, where $\iota\in\Hom_{\mc A}(\mc H_0,\mc H_k\boxtimes\mc H_{\ovl k})$ is $\coev_{k,\ovl k}$, and $\mu\in\Hom_{\mc A}(\mc H_k\boxtimes\mc H_{\ovl k}\boxtimes \mc H_k\boxtimes\mc H_{\ovl k},\mc H_k\boxtimes\mc H_{\ovl k})$ is $\ev_{k,\ovl k}(\id_k\otimes\ev_{\ovl k,k}\otimes\id_{\ovl k})$. As we will show, $S_{\wtd I}$ and $F_{\wtd I}$ are closely related to the $S$ and $F$ operators of non-local extensions of $\mc A$ constructed from $C^*$-Frobenius algebras. Thus, using the  Tomita-Takesaki theory for those non-local extensions, we are able to show that $S_{\wtd I}$ and $F_{\wtd I}$ are always preclosed, and that $\Delta_{\wtd I}$ and $\fk J_{\wtd I}$ satisfy similar algebraic relations as those in Tomita-Takesaki theory (see for example Proposition \ref{lb25}). The idea here is the same as in \cite{Fal00} and \cite{Tak02}.

We emphasize that the categorical extensions and the non-local finite index extensions (by $C^*$-Frobenius algebras) of a conformal net $\mc A$ are closely related. So are their modular theories. As we will see, the proof of the B-W theorem for categorical extensions relies on that  for non-local extensions, and vice versa. Let us now point out some key ideas in the proof. 

Just like the proof of many other versions of B-W theorem, the starting point of our proof is the following well known  consequence (cf. \cite{ALR01,LR04}; see also Lemma \ref{lb35}) of Borchers' theorem \cite{Bor92}: for a non-local extension $\mc B$ constructed from the $C^*$-Frobenius algebra $Q$, $z(t)=\Delta_{\wtd I}^{\im t}\delta_{\wtd I}(2\pi t)$ is a one-parameter group independent of the arg-valued interval $\wtd I$. To show that $z(t)=1$ \emph{when $Q$ is standard}, we need to first show that the modular operator for $\mc B(\wtd I)$ is the inverse of that of $\mc B(\wtd I')$; here, $\wtd I'$ is the ``clockwise complement of $\wtd I$, i.e., it is the interior of the complement of $I$ in $\mbb S^1$, equipped with the arg-function $\arg_{I'}$ satisfying $\arg_{I'}<\arg_I<\arg_{I'}+2\pi$. This fact is obvious when $Q$ is commutative (i.e., when $\mc B$ is local extension, which means $\mc B(\wtd I')$ is the commutant of $\mc B(I)$). Especially, this is true when $Q$ is the identity object, namely $\mc B=\mc A$.\footnote{Note that one does not need the B-W theorem for the conformal net $\mc A$ to prove that $z(t)$ is a one-parameter group independent of $\wtd I$. It is rather the opposite: that this fact proves the B-W theorem for $\mc A$. Thus, it should now be clear to the readers that our proofs of B-W theorem for non-local extensions and categorical extensions do not rely on the B-W theorem of $\mc A$, but rather, it follows the same pattern as the proof of the latter.} But this is not an obvious fact when the standard $Q$ is in general non-commutative. It turns out that our Categorical Extensions provide a natural framework for proving this fact: we prove it using the previously mentioned crucial relation $F_{\wtd I}=\vartheta S_{\wtd I}$ for categorical $S$ and $F$ operators (Proposition \ref{lb28}). (In particular, the property that the categorical $S$ and $F$ operators have the same absolute value $\Delta_{\wtd I}$ is remarkable.) We also use the fact that the categorical $S_{\wtd I}$ and $F_{\wtd I}$ can be related directly to the $S$ operators for $\mc B(\wtd I)$ and its commutant $\mc B(\wtd I)'$ (cf. Proposition \ref{lb9}; see also the non-local Haag duality in Theorem \ref{lb10}). 

Once we have proved the B-W theorem for standard non-local extensions, it is not difficult to prove categorical B-W theorem by restricting to $\mc A$-subrepresentations of $Q$. Then, using this categorical B-W theorem, one can also prove easily   the (modified) non-local B-W theorem for non-local extensions of $\mc A$ associated to a non-necessarily standard $C^*$-Frobenius algebra $Q$. These two B-W theorems are the main results of our paper, which are  stated in details in theorems \ref{lb26} and \ref{lb27}. Roughly speaking, the categorical B-W theorem says:
\begin{thm}[Categorical B-W theorem]\label{lb02}
We have
\begin{align}
\Delta_{\wtd I}^{\im t}=\delta_I(-2\pi t)\label{eq01}
\end{align}
when acting on any dualizable $\mc H_i$. Moreover, $\Theta:=\fk J_{\wtd{\mbb S^1_+}}$ is a PCT-operator for the (rigid) categorical extension.
\end{thm}
Let $Q=(\mc H_a,\mu,\iota)$ be a $C^*$-Frobenius algebra in $\RepdA$, where $\iota\in\Hom_{\mc A}(\mc H_0,\mc H_a)$ and $\mu\in\Hom_{\mc A}(\mc H_a\boxtimes\mc H_a,\mc H_a)$. Let $\mc H_{\ovl a}$ and $\ev_{a,\ovl a}$ be the dual object and the standard evaluation. Let $\epsilon$ be the unique invertible morphism in $\Hom_{\mc A}(\mc H_a,\mc H_{\ovl a})$ satisfying
\begin{align*}
\ev_{\ovl a,a}(\epsilon\otimes\id_a)=\iota^*\mu.
\end{align*}
($\epsilon$ will be called \textbf{reflection operator} in our paper.) We remark that $\epsilon^*\epsilon$ is independent of the choice of $\mc H_{\ovl a}$ and $\ev_{a,\ovl a}$ for $\mc H_a$. Then we have:
\begin{thm}[Modified non-local B-W theorem]\label{lb01}
Let $\mc B$ be the non-local extension of $\mc A$	obtained through $Q$. For any $\wtd I\in\Jtd$, Let $D_{\wtd I}$ and $\fk J^Q_{\wtd I}$ be the modular operator and conjugation associated to $(\mc B(\wtd I),\iota\Omega)$. Then
\begin{align}
D_{\wtd I}^{\im t}=(\epsilon^*\epsilon)^{\im t}\delta_I(-2\pi t),\label{eq02}
\end{align}
and $\Theta^Q:=\fk J^Q_{\wtd{\mbb S^1_+}}$ is a PCT operator for $\mc B$ and its ``clockwise dual net" $\mc B'$.
\end{thm}

Some remarks on these two theorems:
\begin{itemize}
\item Equivalent forms of equation \eqref{eq01} already appeared in \cite{FRS92,Joer96} and in \cite{Lon97}. In \cite{FRS92,Joer96}, the $S$ operators are defined for reduced field bundles, which are an alternative model for charged fields (intertwining operators) of conformal nets. For our purpose (see the beginning of the introduction), categorical extensions might be more convenient than reduce field bundles. In \cite{Lon97}, Longo showed that the dilation group $\delta_I$ is related to Connes Radon-Nikodym derivatives, which are in turn related to Connes spatial derivatives and hence related to our $\Delta_{\wtd I}$ (see remark \ref{lb24}).
\item Similar to \cite{Joer96,GL96}, the conformal spin-statistics theorem $\vartheta=e^{2\im\pi L_0}$  is a consequence of the PCT theorem for (rigid) categorical extensions (see theorem \ref{lb20}).
\item The $C^*$-Frobienius algebra $Q$ is standard if and only if $\epsilon$ is unitary. Thus, by \eqref{eq02}, for the non-local extension $\mc B$ of $\mc A$ obtained by $Q$, the standard geometric modular theorem $D_{\wtd I}^{\im t}=\delta_I(-2\pi t)$ holds if and only if $Q$ is standard.
\item When $Q$ is irreducible (as a left $Q$-module), theorem \ref{lb01} was proved by \cite{LR04} proposition 3.5-(ii).
\end{itemize}

This article is organized as follows. The first three sections provide background materials of this article. In section 1 we review the definitions of M\"obius covariant nets and conformal nets and their representations. In section 2 we review the definition and basic properties of categorical extensions of conformal nets. In section 3, we review some important facts about dualizable objects, their duals and standard evaluations. Later on, we will focus on the rigid braided $C^*$-tensor category $\RepdA$ of dualized objects, i.e., dualizable objects with chosen duals and standard evaluations.  

Construction of non-local extensions of conformal nets via $C^*$-Frobenius algebras (or Q-systems) was first studied in \cite{LR95} using endomorphisms of von Neumann algebras. A parallel construction using bimodules and Connes fusions was given in \cite{Mas97}. In section 4, we use categorical extensions as a new method to realize such construction of non-local extensions. Our method emphasizes the close relation between the charged field operators of a conformal net and the field operators of its non-local extensions, and explains the slogan ``non-local extensions are subquotients of categorical extensions" proposed in the Introduction of \cite{Gui21a}. In section 5 we define the $S$ and $F$ operators for rigid categorical extensions, and prove many elementary but important properties for these operators. Especially, we relate them with the $S$ and $F$ operators of non-local extensions. In section 6 we prove the main results of this article, namely theorems \ref{lb02} and \ref{lb01}. In section 7 we use the modular theory of categorical extensions to study the preclosedness of certain unbounded charged fields of conformal nets. Although our main motivation of this article is to study the functional analytic properties of these  charged field operators, here we do not  give a systematic study of this topic but leave it to future works.

Categorical extensions of conformal nets are closely related to Connes fusion. In section A we briefly explain this relation for the convenience of the readers who are not familiar with this topic. We hope that this appendix section would help them understand the axioms in the definition of categorical extensions. In section B we prove that the rigid categorical extensions of M\"obius covariant nets are M\"obius covariant. This result parallels the conformal covariance of the categorical extensions of conformal (covariant) nets proved in \cite{Gui21a} section 2.4 and theorem 3.5. Indeed, our proof of the M\"obius covariance in this article can be adapted to give a simpler proof of the conformal covariance in \cite{Gui21a}; see the end of section B.

\subsection*{Acknowledgment}

We would like to thank Sebastiano Carpi for many enlightening and helpful discussions during the preparation of this article. We also thank the referees for  many helpful comments on the first version of this article.

\section{Backgrounds}\label{lb29}

Let $\mc J$ be the set of all non-empty non-dense open intervals in the unit circle $\mbb S^1$. If $I\in\mc J$, then $I'$ denotes the interior of the complement of $I$, which is also an element in $\mc J$. The group $\Diffp(\mbb S^1)$ of orientation-preserving diffeomorphisms of $\mbb S^1$ contains the subgroup $\PSU$ of M\"obius transforms of $\mbb S^1$. If $I\in\mc J$, we let $\Diff(I)$ be the subgroup of all $g\in\Diffp(\mbb S^1)$ acting as identity on $I'$.

In this article, we always let $\mc A$ be an (irreducible) \textbf{M\"obius covariant net}, which means that for each $I\in\mathcal J$ there is a von Neumann algebra $\mathcal A(I)$ acting on a fixed separable Hilbert space $\mathcal H_0$, such that the following conditions hold:\\
(a) (Isotony) If $I_1\subset I_2\in\mathcal J$, then $\mathcal A(I_1)$ is a von Neumann subalgebra of $\mathcal A(I_2)$.\\
(b) (Locality) If $I_1,I_2\in\mathcal J$ are disjoint, then $\mathcal A(I_1)$ and $\mathcal A(I_2)$ commute.\\
(c) (M\"obius covariance) We have a strongly continuous  unitary representation $U$ of $\PSU$ on $\mc H_0$ such that for any $g\in\PSU, I\in\mc J,$, 
\begin{align*}
U(g)\mc A(I)U(g)^*=\mc A(gI).
\end{align*}
(d) (Positivity of energy) The generator $L_0$ of the rotation subgroup $\varrho$  is positive.\\
(e) There exists a  unique (up to scalar) $\PSU$-invariant unit vector $\Omega\in\mathcal H_0$. Moreover, $\Omega$ is  cyclic under the action of $\bigvee_{I\in\mathcal J}\mathcal M(I)$ (the von Neumann algebra generated by all $\mathcal M(I)$).

$\mc A$ satisfies \textbf{Haag duality} \cite[Thm. 2.19]{GF93}, i.e., that $\mc A(I')=\mc A(I)'$. This result plays a crucial role in the construction of categorical extensions, and is thus an input in the proof of Haag duality for non-local extensions (cf. Proposition \ref{lb39} or Theorem \ref{lb10}). 

We say that $\mc A$ is a \textbf{conformal (covariant) net} if the representation $U$ of $\PSU$ on $\mc H_0$ can be extended to a strongly continuous projective unitary representation $U$ of $\Diffp(\mbb S^1)$ on $\mathcal H_0$, such that for any $g\in \Diffp(\mbb S^1),I\in\mathcal J$, and any representing element $V\in\mathcal U(\mathcal H_0)$ of $U(g)$,
\begin{align*}
V\mathcal A(I)V^*=\mathcal A(gI).
\end{align*}
Moreover, if $g\in\Diff(I)$ and $x\in\mathcal A(I')$, then
\begin{align*}
VxV^*=x.
\end{align*}

Let $\mathcal H_i$ be a separable Hilbert space. Recall that a (normal) representation $(\mathcal H_i,\pi_i)$ of $\mathcal A$ (also called an $\mc A$-module) associates to each $I\in\mathcal J$ a unital *-representation $\pi_{i,I}:\mathcal A(I)\rightarrow B(\mathcal H_i)$, such that for any $I_1,I_2\in\mathcal J$ satisfying $I_1\subset I_2$, and any $x\in\mathcal A(I_1)$, we have $\pi_{i,I_1}(x)=\pi_{i,I_2}(x)$. We write $\pi_{i,I}(x)$  as $\pi_i(x)$ or just $x$ when no confusion arises.

Let $\scr G$ be the universal covering of $\Diffp(\mbb S^1)$. The corresponding projective representation of $\scr G$ on $\mc H_0$ is also denoted by $U$. Then $\scr G$ has a central extension
\begin{align*}
1\rightarrow U(1)\rightarrow\GA\rightarrow \scr G\rightarrow 1
\end{align*}
associated to the projective representation of $\Diffp(\mbb S^1)$ on $\mc H_0$. In other words, we set
\begin{align*}
\GA=\{(g,V)\in\mathscr G\times \mathcal U(\mathcal H_0)| V \textrm{ is a representing element of } U(g) \}.
\end{align*}
$\scr G_{\mc A}$ actually only depends on the central charge of $\scr A$; see \cite{Gui21a}. This fact will not be used in the present article.

Then the projective representation $\Diffp(\mbb S^1)\curvearrowright\mc H_0$ gives rise to an actual unitary (continuous) representation of $\GA$ of $\mc H_0$, also denoted by $U$. For each $I$, we let $\scr G(I)$ be the connected component containing $1$ of the preimage of $\Diff(I)$ under the covering map $\scr G\rightarrow \Diffp(\mbb S^1)$. Similarly, let $\GA(I)$ be the preimage of $\scr G(I)$ under $\GA\rightarrow\scr G$. If $\mc A$ is conformal covariant, then any $\mc A$-module $\mc H_i$ is \textbf{conformal covariant}, which means that there is a unique representation $U_i$ of $\GA$ on $\mc H_i$ such that for any $I\in\mc J$ and $g\in\GA(I)$,
\begin{align}
U_i(g)=\pi_i(U(g)).\label{eq10}
\end{align}
Moreover, this representation of $\scr G_{\mc A}$ continuous. These facts follow from \cite{AFK04} (only for irreducible representations) and \cite{Hen19} theorem 11. (See also \cite{Gui21a}[Thm. 2.2] for a detailed explanation.) Moreover, the generator of the rotation subgroup acting on $\mc H_i$ is positive by \cite{Wei06} theorem 3.8. From \eqref{eq10} and the fact that $\GA$ is algebraically generated by  $\{\GA(I):I\in\mc J \}$ proved in \cite{Hen19} Lemma 17-(ii) (see also \cite{Gui21a} proposition 2.2), it is clear that any homomorphism of conformal net modules is also a homomorphism of representations of $\GA$. Moreover, for any $g\in\GA$ and $x\in\mc A(I)$ one has
\begin{align}
U_i(g)\pi_{i,I}(x)U_i(g)^*=\pi_{i,gI}(U(g)xU(g)^*).\label{eq20}
\end{align}
Very often, we will write $U(g)$ and $U_i(g)$ as $g$ for short.

Let $\UPSU$ be the universal cover of $\PSU$, regarded as a subgroup of $\scr G$. By \cite{Bar54}, the restriction of any strongly continuous projective representation of $\scr G$ to $\UPSU$ can be lifted to a unique strongly continuous unitary representation of $\UPSU$. Thus $\UPSU$ is also a subgroup of $\GA$. Note that  the action of $\UPSU$ on $\mc H_0$ also preserves $\Omega$. We say that an $\mc A$-module $\mc H_i$ is \textbf{M\"obius covariant} if there is a strongly continuous unitary representation $U_i$ of $\UPSU$ on $\mc H_i$ such that  \eqref{eq20} holds for any $g\in\UPSU$ and $I\in\mc J$.

In the literature, a M\"obius covariant representation is often also assumed to have positive generator of rotation group (or equivalently, positive generator of translation group \cite[Lemma 3.1]{Wei06}). Since in our article we are mainly interested in dualizable representations (equivalently, representations with finite indexes), the positive energy condition is automatic by \cite[Cor. 4.4]{BCL98}.

\section{Categorical extensions}

Let $\RepA$ be the $C^*$-category of $\mc A$-modules whose objects are denoted by $\mc H_i,\mc H_j,\mc H_k,\dots$. Then one can equip $\RepA$ with a structure of braided $C^*$-tensor category either via Doplicher-Haag-Roberts (DHR) superselection theory \cite{FRS89,FRS92}, or via Connes fusion \cite{BDH15,BDH17,Gui21a}. These two constructions are equivalent by \cite{Gui21a} chapter 6. The unit of $\RepA$ is $\mc H_0$. We write the tensor (fusion) product of two $\mc A$-modules $\mc H_i,\mc H_j$ as $\mc H_i\boxtimes\mc H_j$. We assume without loss of generality that $\RepA$ is strict, which means that we will not distinguish between $\mc H_0,\mc H_0\boxtimes\mc H_i,\mc H_i\boxtimes\mc H_0$, or $(\mc H_i\boxtimes \mc H_j)\boxtimes\mc H_k$ and $\mc H_i\boxtimes(\mc H_j\boxtimes\mc H_k)$ (abbreviated to $\mc H_i\boxtimes\mc H_j\boxtimes H_k$). In the following, we review the definition and the basic properties of closed vector-labeled categorical extensions of $\mc A$ (abbreviated to ``categorical extensions" for short) introduced in \cite{Gui21a}.

To begin with, if $\mc H_i,\mc H_j$ are $\mc A$-modules and $I\in\mc J$, then $\Hom_{\mc A(I')}(\mc H_i,\mc H_j)$ denotes the vector space of bounded linear operators $T:\mc H_i\rightarrow\mc H_j$ such that $T\pi_{i,I'}(x)=\pi_{j,I'}(x)T$ for any $x\in\mc A(I')$. We then define $\mc H_i(I)=\Hom_{\mc A(I')}(\mc H_0,\mc H_i)\Omega$, which is a dense subspace of $\mc H_i$. Note that $I\subset J$ implies $\mc H_i(I)\subset\mc H_i(J)$. Moreover, if $G\in\Hom_{\mc A}(\mc H_i,\mc H_j)$, then
\begin{align}
G\mc H_i(I)\subset\mc H_j(I).\label{eq62}
\end{align}

If $I\in\mc J$, an arg-function of $I$ is, by definition, a continuous function $\arg_I:I\rightarrow\mbb R$ such that for any $e^{it}\in I$, $\arg_I(e^{it})-t\in 2\pi\mbb Z$. $\wtd I=(I,\arg_I)$ is called an \textbf{arg-valued interval}. Equivalently, $\wtd I$ is a branch of $I$ in the universal cover of $\mbb S^1$. We let $\wtd{\mc J}$ be the set of arg-valued intervals. If $\wtd I=(I,\arg_I)$ and $\wtd J=(J,\arg_J)$ are in $\Jtd$, we say that $\wtd I$ and $\wtd J$ are disjoint if $I$ and $J$ are so. Suppose moreover that for any $z\in I,\zeta\in J$ we have $\arg_J(\zeta)<\arg_I(z)<\arg_J(\zeta)+2\pi$, then we say that $\wtd I$ is \textbf{anticlockwise} to $\wtd J$ (equivalently, $\wtd J$ is \textbf{clockwise} to $\wtd I$). We write $\wtd I\subset\wtd J$ if $I\subset J$ and $\arg_J|_I=\arg_I$. Given $\wtd I\in\Jtd$, we also define $\wtd I'=(I',\arg_{I'})\in\wtd J$ such that $\wtd I$ is anticlockwise to $\wtd I'$. We say that $\wtd I'$ is the \textbf{clockwise complement} of $\wtd I$.

\begin{df}\label{lb1}
A (closed and  vector-labeled) \textbf{categorical extension} $\scr E=(\mc A,\RepA,\boxtimes,\mc H)$ of $\mc A$ associates, to any  $\mc H_i,\mc H_k\in\Obj(\RepA)$ and any $\wtd I\in\Jtd,\fk \xi\in\mc H_i(I)$, bounded linear operators
	\begin{gather*}
	L(\xi,\wtd I)\in\Hom_{\mc A(I')}(\mc H_k,\mc H_i\boxtimes\mc H_k),\\
	R(\xi,\wtd I)\in\Hom_{\mc A(I')}(\mc H_k,\mc H_k\boxtimes\mc H_i),
	\end{gather*}
	such that the following conditions are satisfied:\\
	(a) (Isotony) If $\wtd I_1\subset\wtd I_2\in\Jtd$, and $\xi\in\mc H_i(I_1)$, then $L(\xi,\wtd I_1)=L(\xi,\wtd I_2)$, $R(\xi,\wtd I_1)=R(\xi,\wtd I_2)$ when acting on any  $\mc H_k\in\Obj(\RepA)$.\\
	(b) (Functoriality) If $\mc H_{i},\mc H_k,\mc H_{k'}\in\Obj(\RepA)$, $F\in\Hom_{\mc A}(\mc H_k,\mc H_{k'})$,  the following diagrams commute for any $\wtd I\in\Jtd,\xi\in\mc H_i(I)$.
	\begin{gather}
	\begin{CD}
	\mc H_k @>F>> \mc H_{k'}\\
	@V L(\xi,\wtd I)  VV @V L(\xi,\wtd I)  VV\\
	\mc H_i\boxtimes\mc H_k @> \id_i\otimes F>> \mc H_i\boxtimes\mc H_{k'}
	\end{CD}\qquad\qquad
	\begin{CD}
	\mc H_k @> R(\xi,\wtd I)  >> \mc H_k\boxtimes\mc H_i\\
	@V F VV @V F\otimes\id_i  VV\\
	\mc H_{k'} @>R(\xi,\wtd I) >> \mc H_{k'}\boxtimes\mc H_i
	\end{CD}.
	\end{gather}
	(c) (State-field correspondence\footnote{For general (i.e., non-necessarily closed or vector-labeled) categorical extensions, this axiom is replaced by the neutrality and the Reeh-Schlieder property; see \cite{Gui21a} section 3.1.}) For any $\mc H_i\in\Obj(\RepA)$, under the identifications $\mc H_i=\mc H_i\boxtimes\mc H_0=\mc H_0\boxtimes\mc H_i$, the relation
	\begin{align}
	L(\xi,\wtd I)\Omega=R(\xi,\wtd I)\Omega=\xi
	\end{align}
	holds for any $\wtd I\in\Jtd,\xi\in\mc H_i(I)$. It follows immediately that when acting on $\mc H_0$, $L(\xi,\wtd I)$ equals $R(\xi,\wtd I)$ and is independent of $\arg_I$.\\
	(d) (Density of fusion products) If $\mc H_i,\mc H_k\in\Obj(\RepA),\wtd I\in\Jtd$, then the set $L(\mc H_i(I),\wtd I)\mc H_k$ spans a dense subspace of $\mc H_i\boxtimes\mc H_k$, and $R(\mc H_i(I),\wtd I)\mc H_k$ spans a dense subspace of $\mc H_k\boxtimes\mc H_i$. (Indeed, they span the full space $\mc H_i\boxtimes\mc H_k$ and $\mc H_k\boxtimes\mc H_i$ respectively.)\\
	(e) (Locality) For any $\mc H_k\in\Obj(\RepA)$, disjoint $\wtd I,\wtd J\in\Jtd$ with $\wtd I$ anticlockwise to $\wtd J$, and any $\xi\in\mc H_i(I),\eta\in\mc H_j(J)$, the following diagram \eqref{eq5}  commutes adjointly.
	\begin{align}
	\begin{CD}
	\mc H_k @> \quad R(\eta,\wtd J)\quad   >> \mc H_k\boxtimes\mc H_j\\
	@V L(\xi,\wtd I)   V  V @V L(\xi,\wtd I) VV\\
	\mc H_i\boxtimes\mc H_k @> \quad R(\eta,\wtd J) \quad  >> \mc H_i\boxtimes\mc H_k\boxtimes\mc H_j
	\end{CD}\label{eq5}
	\end{align}
Here, the \textbf{adjoint commutativity} of diagram \eqref{eq5} means that $R(\eta,\wtd J)L(\xi,\wtd I)=L(\xi,\wtd I)R(\eta,\wtd J)$ when acting on $\mc H_k$, and $R(\eta,\wtd J)L(\xi,\wtd I)^*=L(\xi,\wtd I)^*R(\eta,\wtd J)$ when acting on $\mc H_i\boxtimes\mc H_k$.\\
	(f) (Braiding) There is a unitary linear map $\ss_{i,j}:\mc H_i\boxtimes\mc H_j\rightarrow\mc H_j\boxtimes \mc H_i$  for any $\mc H_i,\mc H_j\in\Obj(\RepA)$, such that  
	\begin{align}
	\ss_{i,j} L(\xi,\wtd I)\eta=R(\xi,\wtd I)\eta
	\end{align}
	whenever $\wtd I\in\Jtd,\xi\in\mc H_i(I)$, $\eta\in\mc H_j$.
\end{df}

Note that $\ss_{i,j}$ is unique by the density of fusion products. Moreover, $\ss_{i,j}$ commutes with the actions of $\mc A$, and  is the same as the braid operator of $\RepA$; see \cite{Gui21a} sections 3.2, 3.3. The existence of $\scr E$ is also proved in \cite{Gui21a} sections 3.2.\footnote{In \cite{Gui21a} we assume $\mc A$ to be conformal covariant  for the simplicity of discussions. Most results in that article (for example, the construction of Connes categorical extensions, the uniqueness of braided $C^*$-tensor categories, the uniqueness of vector-labeled closed categorical extensions, etc.) do not rely on the conformal covariance and are also true for M\"obius covariant nets. The only exception is the conformal covariance of categorical extensions, which should be replaced by M\"obius covariance when the $\mc A$ is only M\"obius covariant; see theorem \ref{lb19} and section B.}

\begin{rem}
We see that $L(\xi,\wtd I)$ and $R(\xi,\wtd I)$  can act on any object in $\RepA$. If we want to emphasize that they are acting on a specific object $\mc H_k$, we write $L(\xi,\wtd I)|_{\mc H_k}$ and $R(\xi,\wtd I)|_{\mc H_k}$. It is noteworthy that for any $x\in\mc A(I)$,
\begin{align}
L(x\Omega,\wtd I)|_{\mc H_k}= R(x\Omega,\wtd I)|_{\mc H_k}=\pi_{k,I}(x).\label{eq9}
\end{align}
See the end of \cite{Gui21a} section 3.1. By the locality and the state-field correspondence, it is also easy to see that
\begin{align}
L(\xi,\wtd I)\eta=R(\eta,\wtd J)\xi\label{eq7}
\end{align}
whenever $\xi\in\mc H_i(I)$, $\eta\in\mc H_j(J)$, and $\wtd I$ is anticlockwise to $\wtd J$.
\end{rem}

Another useful fact is that if $F\in\Hom_{\mc A}(\mc H_i,\mc H_{i'})$, $G\in\Hom_{\mc A}(\mc H_j,\mc H_{j'})$,  $\xi\in\mc H_i(I)$, and $\eta\in\mc H_j$, then
\begin{align}
(F\otimes G)L(\xi,\wtd I)\eta=L(F\xi,\wtd I)G\eta,\qquad (G\otimes F)R(\xi,\wtd I)\eta=R(F\xi,\wtd I)G\eta.\label{eq8}
\end{align}
This was proved in \cite{Gui21a} section 3.3 using Connes fusion, but it also follows directly from the axioms of categorical extensions. To prove the first equation, it suffices to assume that $\eta\in\mc H_j(J)$ where $\wtd J$ is clockwise to $\wtd I$. Then, by the functoriality and relation \eqref{eq7},
\begin{align*}
&(F\otimes G)L(\xi,\wtd I)\eta=(\id\otimes G)(F\otimes \id)L(\xi,\wtd I)\eta=(\id\otimes G)(F\otimes \id)R(\eta,\wtd J)\xi\\
=&(\id\otimes G)R(\eta,\wtd J)F\xi=(\id\otimes G)L(F\xi,\wtd I)\eta=L(F\xi,\wtd I)G\eta.
\end{align*}
The second relation follows from the first one and \eqref{eq7}.

We now prove some fusion relations for the $L$ and $R$ operators of $\scr E$.
\begin{pp}\label{lb5}
Let $\mc H_i,\mc H_j,\mc H_k\in\Obj(\RepA)$, $\wtd I\in\Jtd$, and $\xi\in\mc H_i(I)$.\\
(a) If $\eta\in\mc H_j(I)$, then $L(\xi,\wtd I)\eta\in(\mc H_i\boxtimes\mc H_j)(I)$, $R(\xi,\wtd I)\eta\in(\mc H_j\boxtimes\mc H_i)(I)$, and
\begin{gather}
L(\xi,\wtd I)L(\eta,\wtd I)|_{\mc H_k}=L(L(\xi,\wtd I)\eta,\wtd I)|_{\mc H_k},\\
R(\xi,\wtd I)R(\eta,\wtd I)|_{\mc H_k}=R(R(\xi,\wtd I)\eta,\wtd I)|_{\mc H_k}.
\end{gather}
(b) If $\psi\in(\mc H_i\boxtimes H_j)(I)$ and $\phi\in (\mc H_j\boxtimes H_i)(I)$, then $L(\xi,\wtd I)^*\psi\in\mc H_j(I)$, $R(\xi,\wtd I)^*\phi\in\mc H_j(I)$, and
\begin{gather}
L(\xi,\wtd I)^*L(\psi,\wtd I)|_{\mc H_k}=L(L(\xi,\wtd I)^*\psi,\wtd I)|_{\mc H_k},\\
 R(\xi,\wtd I)^*R(\phi,\wtd I)|_{\mc H_k}=R(R(\xi,\wtd I)^*\phi,\wtd I)|_{\mc H_k}.
\end{gather}
\end{pp}
As a special case, we see that if $\xi\in\mc H_i(I)$ and $x\in\mc A(I)$, then $x\xi\in\mc H_i(I)$, and
\begin{gather}
L(x\xi,\wtd I)=xL(\xi,\wtd I),\qquad R(x\xi,\wtd I)=xR(\xi,\wtd I).\label{eq47}
\end{gather}
\begin{proof}
We only prove the first equation of part (b); the second one follows similarly. Part (a) follows either from a similar argument or from \cite{Gui21a} proposition 3.6. Since $L(\xi,\wtd I)^*\psi=L(\xi,\wtd I)^*L(\psi,\wtd I)\Omega$, we clearly have $L(\xi,\wtd I)^*\psi\in\mc H_j(I)$. Choose any $\chi\in\mc H_k(\wtd I')$. Then, by the adjoint commutativity of left and right operators,
\begin{align*}
&L(\xi,\wtd I)^*L(\psi,\wtd I)\chi=L(\xi,\wtd I)^*L(\psi,\wtd I)R(\chi,\wtd I')\Omega=R(\chi,\wtd I')L(\xi,\wtd I)^*L(\psi,\wtd I)\Omega\\
=&R(\chi,\wtd I')L(\xi,\wtd I)^*\psi=R(\chi,\wtd I')L(L(\xi,\wtd I)^*\psi)\Omega=L(L(\xi,\wtd I)^*\psi)R(\chi,\wtd I')\Omega\\
=&L(L(\xi,\wtd I)^*\psi)\chi.
\end{align*}
\end{proof}

Next, we discuss the conformal covariance of $\scr E$. For any $\wtd I=(I,\arg_I)\in\Jtd$ and $g\in\GA$, we have $gI$ defined by the action of $\Diffp(\mbb S^1)$ on $\mbb S^1$. We now set $g\wtd I=(gI,\arg_{gI})$, where $\arg_{gI}$ is defined as follows. Choose any map $\gamma:[0,1]\rightarrow\GA$ satisfying $\gamma(0)=1,\gamma(1)=g$ such that $\gamma$ descends to a (continuous) path in $\scr G$. Then for any $z\in I$ there is a path $\gamma_z:[0,1]\rightarrow \mbb S^1$ defined by $\gamma_z(t)=\gamma(t)z$. The argument $\arg_I(z)$ of $z$ changes continuously along the path $\gamma_z$ to an argument of $gz$, whose value is denoted by $\arg_{gI}(gz)$. 

\begin{thm}[\cite{Gui21a} theorem 3.13]\label{lb2}
If $\mc A$ is conformal covariant, then $\scr E=(\mc A,\RepA,\boxtimes,\mc H)$ is \textbf{conformal covariant}, which means that for any $g\in\GA,\wtd I\in\wtd{\mc J},\mc H_i\in\Obj(\RepA),\xi\in\mc H_i(I)$, there exists an element $g\xi g^{-1}\in\mc H_i(g I)$ such that
\begin{align}
L(g\xi g^{-1},g\wtd I)=gL(\xi,\wtd I)g^{-1},\qquad R(g\xi g^{-1},g\wtd I)=gR(\xi,\wtd I)g^{-1}\label{eq1}
\end{align}
when acting on any $\mc H_j\in\Obj(\RepA)$.
\end{thm}
It's clear that we have
\begin{align}
g\xi g^{-1}=gL(\xi,\wtd I)g^{-1}\Omega=gR(\xi,\wtd I)g^{-1}\Omega.
\end{align}
In particular, when $g\in\UPSU$ one has $g\xi g^{-1}=g\xi$ by the state-field correspondence and that $g\Omega=\Omega$. Therefore
\begin{align}
L(g\xi,g\wtd I)=gL(\xi,\wtd I)g^{-1},\qquad R(g\xi,g\wtd I)=gR(\xi,\wtd I)g^{-1}\qquad(\forall g\in\UPSU).\label{eq2}
\end{align}
The above property is called the \textbf{M\"obius covariance} of $\scr E$.

\section{Rigid categorical extensions and dualized objects}\label{lb30}

Recall that a representation $\mc H_i$ of $\mc A$ is called dualizable if there exists an object $\mc H_{\ovl i}\in\Obj(\RepA)$ (called \textbf{dual object}) and \textbf{evaluations} $\ev_{i,\ovl i}\in\Hom_{\mc A}(\mc H_i\boxtimes\mc H_{\ovl i},\mc H_0)$ and $\ev_{\ovl i,i}\in\Hom_{\mc A}(\mc H_{\ovl i}\boxtimes\mc H_i,\mc H_0)$  satisfying the conjugate equations
\begin{gather}
	(\ev_{i,\ovl i}\otimes\id_i)(\id_i\otimes\coev_{\ovl i,i})=\id_i=(\id_i\otimes\ev_{\ovl i,i})(\coev_{i,\ovl i}\otimes\id_i),\label{eq3}\\
	(\ev_{\ovl i,i}\otimes\id_{\ovl i})(\id_{\ovl i}\otimes\coev_{i,\ovl i})=\id_{\ovl i}=(\id_{\ovl i}\otimes\ev_{i,\ovl i})(\coev_{\ovl i,i}\otimes\id_{\ovl i}),\label{eq4}
\end{gather}
where we set $\coev_{i,\ovl i}=\ev_{i,\ovl i}^*,\coev_{\ovl i,i}=\ev_{\ovl i,i}^*$. It is clear that $\mc H_i$ is also a dual object of $\mc H_{\ovl i}$. 

Note that in each of \eqref{eq3} and \eqref{eq4}, the first equation is equivalent to the second one by taking adjoint. Note also that $\ev_{\ovl i,i}$ is uniquely determined by $\ev_{i,\ovl i}$ since  $\coev_{\ovl i,i}$ is so. Moreover, one can choose the evaluations to be \textbf{standard}, which means that besides the conjugate equations, we also have
\begin{align}
	\ev_{i,\ovl i}(F\otimes\id_{\ovl i})\coev_{i,\ovl i}=\ev_{\ovl i,i}(\id_{\ovl i}\otimes F)\coev_{\ovl i,i}\label{eq57}
\end{align}
for any $F\in\End_{\mc A}(\mc H_i)$. Then there exist positive numbers $d_i=d_{\ovl i}$ satisfying $\ev_{i,\ovl i}\coev_{i,\ovl i}=\ev_{\ovl i,i}\coev_{\ovl i,i}=d_i\id_0=d_{\ovl i}\id_0$, called the \textbf{quantum dimensions} of $\mc H_i$ and $\mc H_{\ovl i}$.  

Standard evaluations exist and are unique up to unitaries. The uniqueness means that if $\mc H_{\wht i}$ is also a dual object of $\mc H_i$, and $\wtd\ev_{i,\wht i}\in\Hom_{\mc A}(\mc H_i\boxtimes\mc H_{\wht i},\mc H_0)$ is a standard evaluation, then $\wtd\ev_{i,\wht i}:=\ev_{i,\ovl i}(\id_i\otimes u)$ for a (necessarily unique) unitary $u\in\Hom_{\mc A}(\mc H_{\wht i},\mc H_{\ovl i})$. Moreover, we also have $\wtd\ev_{\wht i,i}:=\ev_{\ovl i,i}(u\otimes\id_i)$. We refer the reader to \cite{LR97} or \cite{Yam04} or \cite{BDH14} for more details. 

\begin{df}
We let $\RepfA$ be the $C^*$-tensor category of dualizable M\"obius covariant representations of $\mc A$.
\end{df}

Recall from Section \ref{lb29} that when $\mc A$ is conformal covariant, the conformal covariance and hence the M\"obius covariance of dualizable representations are automatic. Then $\RepfA$ is a rigid braided $C^*$-tensor category.\footnote{That $\RepfA$ is closed under fusion product $\boxtimes$ is known to experts. In  section \ref{lb18} we give a proof of this fact.} We remark that $\mc H_{\ovl i}$ is also M\"obius covariant by \cite{GL96} theorem 2.11. Therefore $\mc H_{\ovl i}\in\Obj(\RepfA)$ if $\mc H_i$ is so. That any $\mc A$-subrepresentation of $\mc H_i$ is M\"obius covariant follows from \cite[Prop. 2.2]{GL96}.

We can now restrict $\scr E$ to $\RepfA$ to obtain a (closed, vector-labeled) \textbf{rigid categorical extension} $\Ef=(\mc A,\RepfA,\boxtimes,\mc H)$, which is also conformal covariant when $\mc A$ is so. This means that when $\mc A$ is conformal covariant, definition \ref{lb1} and theorem \ref{lb2} hold verbatim for $\Ef$, except that $\RepA$ should be replaced by $\RepfA$. When $\mc A$ is only M\"obius covariant, these are also true except that theorem \ref{lb2} should be replaced by M\"obius covariance. Note first of all that for any $\mc H_i\in\Obj(\RepfA)$, the strongly continuous unitary representations of $\UPSU$ making $\mc H_i$ M\"obius covariant are unique by \cite{GL96} proposition 2.2. As an easy consequence of this fact, any morphism in $\RepfA$ intertwines the actions of $\UPSU$ (cf. Lemma \ref{lb21}). The following is proved in section \ref{lb18}.
\begin{thm}\label{lb19}
	$\RepfA$ is closed under $\boxtimes$. Moreover, $\Ef=(\mc A,\RepfA,\boxtimes,\mc H)$ is \textbf{M\"obius covariant}, which means that for any $g\in\UPSU,\wtd I\in\wtd{\mc J},\mc H_i\in\Obj(\RepfA),\xi\in\mc H_i(I)$, we have
	\begin{align}
		L(g\xi,g\wtd I)=gL(\xi,\wtd I)g^{-1},\qquad R(g\xi,g\wtd I)=gR(\xi,\wtd I)g^{-1}\label{eq52}
	\end{align}
	when acting on any $\mc H_j\in\Obj(\RepfA)$.
\end{thm}

Since all the $\ev$ and $\coev$ for $\mc H_i$ are determined by $\ev_{i,\ovl i}$, we define the following:

\begin{df}
For any dualizable $\mc H_i$, a \textbf{dualizing data} is $(\mc H_{\ovl i},\ev_{i,\ovl i})$ where $\mc H_{\ovl i}\in\Obj(\RepA)$ is a dual object of $\mc H_i$,  $\ev_{i,\ovl i}\in\Hom_{\mc A}(\mc H_i\boxtimes\mc H_{\ovl i},\mc H_0)$ satisfies \eqref{eq3}, \eqref{eq4}, and $\coev_{i,\ovl i}=\ev_{i,\ovl i}^*,\coev_{\ovl i,i}=\ev_{\ovl i,i}^*$, and the evaluations are standard.
\end{df}

A dualizable $\mc H_i$ have infinitely many different dualizing data. However, for the purpose of this article, we need to choose a canonical dualizing data satisfying some nice property, such as that the double dual object equals the original one. For this purpose, we introduce the following definition.

\begin{df}
A dualizable M\"obius covariant representation $\mc H_i$ of $\mc A$ is called \textbf{dualized}, if $\mc H_i$ is equipped with a dualizing data $(\mc H_{\ovl i},\ev_{i,\ovl i})$.
\end{df}

Equivalently, a dualized object is a triple $(\mc H_i,\mc H_{\ovl i},\ev_{i,\ovl i})$ where $\mc H_i\in\Obj(\RepfA)$ and $(\mc H_{\ovl i},\ev_{i,\ovl i})$ is a dualizing data of $\mc H_i$.

In this article, our main results are stated for the category $\RepdA$ and the categorical extension $\Ed$ modeled on it. They are defined as follows.

\begin{df}\label{lb31}
$\RepdA$ is the rigid braided $C^*$-tensor category of dualized (M\"obius covariant) representations of $\mc A$. Morphisms of $\RepdA$ are the same as morphisms of $\RepfA$ (i.e., homomorphisms of representations of $\mc A$). If $\mc H_i,\mc H_j$ are dualized with dualizing data $(\mc H_{\ovl i},\ev_{i,\ovl i})$ and $(\mc H_{\ovl j},\ev_{j,\ovl j})$, we define their tensor (fusion) product to be $\mc H_i\boxtimes\mc H_j$ as a (dualizable M\"obius covariant) representation of $\mc A$, equipped with dualizing data
\begin{align*}
	(\mc H_{\ovl j}\boxtimes\mc H_{\ovl i},\ev_{i\boxtimes j,\ovl j\boxtimes\ovl i})
\end{align*}
where
\begin{align*}
\ev_{i\boxtimes j,\ovl j\boxtimes\ovl i}	=\ev_{i,\ovl i}(\id_i\otimes\ev_{j,\ovl j}\otimes\id_{\ovl i}).
\end{align*}
Moreover, for any $\mc H_i\in\Obj(\RepdA)$ equipped with dualizing data $(\mc H_{\ovl i},\ev_{i,\ovl i})$, we define its \textbf{(canonical) dual object} to be $\mc H_{\ovl i}$ equipped with dualizing data $(\mc H_i,\ev_{\ovl i,i})$, where $\ev_{\ovl i,i}$ is related to $\ev_{i,\ovl i}$ by \eqref{eq3}, \eqref{eq4}, and $\coev_{i,\ovl i}=\ev_{i,\ovl i}^*,\coev_{\ovl i,i}=\ev_{\ovl i,i}^*$.

We let $\mc H_0\in\Obj(\RepdA)$ be the vacuum representation of $\mc A$ with dualizing data $(\mc H_0,\lambda_{0,0})$ where, in general, $\lambda_{0,i}:\mc H_0\boxtimes\mc H_i\xrightarrow{\simeq}\mc H_i$ is the left unitor.
\end{df}

The following is easy to see:

\begin{pp}
Let $\mc H_i,\mc H_j$ be dualized, and let $\mc H_{\ovl i},\mc H_{\ovl j}$ be their respecitve dual objects. Then the dual object of $\mc H_0$ is $\mc H_0$, the dual object of $\mc H_{\ovl i}$ is $\mc H_i$, and the dual object of $\mc H_i\boxtimes \mc H_j$ is $\mc H_{\ovl j}\boxtimes \mc H_{\ovl i}$. If we identify $\mc H_i$ with $\mc H_0\boxtimes\mc H_i$ (resp. $\mc H_i\boxtimes\mc H_0$) using the left (resp. right) unitor, then their dualizing data are automatically identified. If we identify $(\mc H_i\boxtimes\mc H_j)\boxtimes\mc H_k$ with $\mc H_i\boxtimes(\mc H_j\boxtimes\mc H_k)$ using the associator, then their dualizing data are also automatically identified.
\end{pp}

\begin{df}
The forgetful functor $\RepdA\rightarrow\RepfA$ defined by forgetting the dualizing data is clearly an equivalence of $C^*$-tensor categories. We pullback the categorical extension $\scr E$ along the forgetful functor to get a rigid categorical extension
\begin{align*}
	\Ed=(\mc A,\RepdA,\boxtimes,\mc H).
\end{align*}
\end{df}


The following theorem follows obviously from Theorem \ref{lb19}.

\begin{thm}
Theorem \ref{lb19} holds verbatim for $\RepdA$ and $\Ed$.
\end{thm}

We close this section by recalling some well-known facts about rigid (braided) $C^*$-tensor categories. We state these facts only for $\RepdA$.

Choose $\mc H_i,\mc H_j\in\Obj(\RepdA)$ with dualizing date $(\mc H_{\ovl i},\ev_{i,\ovl i}),(\mc H_{\ovl j},\ev_{j,\ovl j})$ respectively. Suppose that $G\in\Hom_{\mc A}(\mc H_i,\mc H_j)$. Then there exists a unique $G^\vee\in\Hom(\mc H_{\ovl j},\mc H_{\ovl i})$ satisfying
\begin{align}
	\ev_{j,\ovl j}(G\otimes\id_{\ovl j})=\ev_{i,\ovl i}(\id_i\otimes G^\vee),\label{eq36}
\end{align}
called the transpose of $G$. We have $G^{\vee\vee}=G$. Thus 
\begin{align}
	\ovl G:=(G^\vee)^*=(G^*)^\vee  \label{eq58}
\end{align}
which is in $\Hom_{\mc A}(\mc H_{\ovl i},\mc H_{\ovl j})$ and called the \textbf{conjugate} of $G$. See for example \cite{Yam04}. It is easy to see that $\ovl{\ovl G}=G$, that if $F$ is another morphism and $FG$ is definable then $\ovl{FG}=\ovl F\ovl G$, that $(\ovl G)^*=\ovl{G^*}$, and that $G$ is a projection (resp. unitary, an isometry, a partial isometry) if and only if $\ovl G$ is so. 


Since $\RepdA$ is a rigid braided $C^*$-tensor category, we can define the \textbf{twist operator} $\vartheta_i$ on any $\mc H_i\in\Obj(\RepdA)$ to be the unique operator in $\End_{\mc A}(\mc H_i)$ such that 
\begin{align}
	\ev_{i,\ovl i}=\ev_{\ovl i,i}\ss_{i,\ovl i}(\vartheta_i\otimes\id_{\ovl i}),\label{eq33}
\end{align}
where we recall that $\ss$ is the braid operator of $\RepA$, and the evaluations are assumed to be standard. Then, by \cite{Mueg00},  $\vartheta_i$ is a unitary operator independent of the dualizing data of $\mc H_i$, and the actions of $\vartheta$ on all $\mc H_i\in\Obj(\RepdA)$ (as $\vartheta_i$) define a ribbon structure compatible with the braided $C^*$-tensor structure of $\RepdA$. (Indeed, $\vartheta_i$ is unitary if and only if the evaluations are standard.)  This means, among other things, that $\vartheta$ is functorial (i.e. natural, which means that it  commutes with homomorphisms),  that for any $\mc H_i,\mc H_j\in\Obj(\RepdA)$,
\begin{gather}
	\ss_{i,j}^2=\vartheta_{i\boxtimes j}(\vartheta_i^{-1}\otimes\vartheta_j^{-1}),\label{eq32}\\
	\vartheta_i^\vee=\vartheta_{\ovl i},
\end{gather}
and (hence) that
\begin{align}
	(\id_i\otimes\vartheta_{\ovl i})\coev_{i,\ovl i}=\ss_{i,\ovl i}^{-1}\coev_{\ovl i,i}.\label{eq15}
\end{align}
(See also \cite{Gui21b} the paragraph containing equation (3.29) for some explanations.)

\section{Non-local extensions}\label{lb34}

$Q$-systems were introduced by R.Longo \cite{Lon94} and are powerful tools for studying local and non-local extensions of conformal and M\"obius covariant nets \cite{LR95,KL04,BKLR15}. In this section, we give a construction of non-local extensions by Q-systems under the framework of categorical extensions. We shall work with a general $C^*$-Frobenius algebra $Q$ in $\RepdA$, and construct a non-local extension $\mc B$ of $\mc A$ via $Q$. The results in this section hold verbatim for $\Ef,\RepfA$, and their $C^*$-Frobenius algebras.

Recall that $Q=(\mc H_a,\mu,\iota)$ is called a \textbf{$C^*$-Frobenius algebra} in $\RepdA$ if $\mc H_a\in\Obj(\RepdA),\mu\in\Hom_{\mc A}(\mc H_a\boxtimes \mc H_a,\mc H_a),\iota\in\Hom_{\mc A}(\mc H_0,\mc H_a)$, and the following conditions are satisfied:
\begin{itemize}
\item (Unit) $\mu(\iota\otimes\id_a)=\id_a=\mu(\id_a\otimes\iota)$.
\item (Associativity) $\mu(\mu\otimes\id_a)=\mu(\id_a\otimes\mu)$.
\item (Frobenius relation) $(\id_a\otimes\mu)(\mu^*\otimes\id_a)=\mu^*\mu$.
\end{itemize}
Note that the associativity and the Frobenius relation are equivalent to the \emph{adjoint commutativity} of the following diagram
	\begin{align}
\begin{CD}
\mc H_a\boxtimes \mc H_a\boxtimes \mc H_a @> \quad \id_a\otimes\mu \quad   >> \mc H_a\boxtimes \mc H_a\\
@V \mu\otimes\id_a   V  V @V \mu VV\\
\mc H_a\boxtimes \mc H_a @> \qquad \mu \qquad  >> \mc H_a
\end{CD}.\label{eq6}
\end{align}

Let us fix a $C^*$-Frobenius algebra $Q$. For any $\xi\in\mc H_a(I)$, we define bounded linear operators on $\mc H_a$:
\begin{align*}
A(\xi,\wtd I)=\mu\cdot L(\xi,\wtd I)|_{\mc H_a},\qquad B(\xi,\wtd I)=\mu\cdot R(\xi,\wtd I)|_{\mc H_a}.
\end{align*}

\begin{df}
For any $\wtd I\in\Jtd$,  $\mc B(\wtd I)$ (resp. $\mc B'(\wtd I)$) is defined to be the set of all $A(\xi,\wtd I)$ (resp. $B(\xi,\wtd I)$) where $\xi\in\mc H_a(I)$.
\end{df}
We shall show that $\mc B:\wtd I\in\Jtd\mapsto \mc B(\wtd I)$ and $\mc B':\wtd I\in\Jtd\mapsto \mc B'(\wtd I)$ are two nets of von Neumann algebras extending $\mc A$, and that the Haag duality $\mc B(\wtd I)'=\mc B'(\wtd I')$ is satisfied. First, notice that $\mc A(I)$ is also acting on $\mc H_a$. We also denote by $\mc A(I)$ the image of $\mc A(I)$ under $\pi_{a,I}$. The following lemma shows that $\mc B$ and $\mc B'$ are  extensions of $\mc A$.

\begin{pp}\label{lb4}
We have $\mc A(I)\subset\mc B(\wtd I)$ and $\mc A(I)\subset\mc B'(\wtd I)$.
\end{pp}
\begin{proof}
Choose any $x\in\mc A(I)$. Then one has
\begin{align}
A(\iota x\Omega,\wtd I)=\pi_{a,I}(x)=B(\iota x\Omega,\wtd I).
\end{align}
Indeed, for any $\eta\in\mc H_a$,
\begin{align*}
A(\iota x\Omega,\wtd I)\eta=\mu\cdot L(\iota x\Omega,\wtd I)\eta=\mu(\iota\otimes\id_a)L(x\Omega,\wtd I)\eta=L(x\Omega,\wtd I)\eta=x\eta,
\end{align*}
where we have used \eqref{eq8}, the unit property, and \eqref{eq9}. The other relation is proved in a similar manner.
\end{proof}

\begin{pp}
If $\wtd I$ is anticlockwise to $\wtd J$, then for any $\xi\in\mc H_a(I)$ and $\eta\in\mc H_a(J)$, $A(\xi,\wtd I)$ \textbf{commutes adjointly} with $B(\eta,\wtd J)$, which means that $A(\xi,\wtd I)B(\eta,\wtd J)=B(\eta,\wtd J)A(\xi,\wtd I)$ and $A(\xi,\wtd I)^*B(\eta,\wtd J)=B(\eta,\wtd J)A(\xi,\wtd I)^*$.
\end{pp}

\begin{proof}
Consider the following matrix of diagrams.
\begin{align}
\begin{CD}
\mc H_a @>\quad R(\eta,\wtd J) \quad >> \mc H_a\boxtimes\mc H_a @> ~~~\quad  \mu \quad~~~  >> \mc H_a\\
@V L(\xi,\wtd I)  VV @V L(\xi,\wtd I) VV @V L(\xi,\wtd I) VV\\
\mc H_a\boxtimes\mc H_a @> \quad  R(\eta,\wtd J) \quad  >> \mc H_a\boxtimes \mc H_a\boxtimes\mc H_a   @> \quad \id_a\otimes \mu\quad  >>\mc H_a\boxtimes\mc H_a\\
@V \mu  VV @V  \mu\otimes\id_a VV @V  \mu  VV\\
\mc H_a @> \quad  R(\eta,\wtd J) \quad  >> \mc H_a\boxtimes\mc H_a   @>~~~ \quad \mu \quad~~~  >>\mc H_a
\end{CD}
\end{align}
The $(1,1)$-diagram commutes adjointly by the locality of $\Ed$. The $(2,1)$- and $(1,2)$-diagrams commute adjointly by the functoriality of $\Ed$. The $(2,2)$-diagram is just  \eqref{eq6}, which we know is commuting adjointly by the associativity and the Frobenius property of $Q$. Thus the largest diagram commutes adjointly, which is exactly the adjoint commutativity of $A(\xi,\wtd I)$ and $B(\eta,\wtd J)$.
\end{proof}

\begin{df}
If $\fk S$ is a set of bounded linear operators on a Hilbert space $\mc H$, its \textbf{commutant} $\fk S'$ is defined to be the set of bounded linear operators on $\mc H$ which commute adjointly with the  operators in $\fk S$. Then $\fk S'$ is a von Neumann algebra. The double commutant $\fk S''$ is called the von Neumann algebra generated by $\fk S$.
\end{df}

\begin{pp}\label{lb39}
For any $\wtd I\in\Jtd$ we have $\mc B(\wtd I)'=\mc B'(\wtd I')$ and $\mc B'(\wtd I')'=\mc B(\wtd I)$. As a consequence, $\mc B(\wtd I)$ and $\mc B'(\wtd I')$ are von Neumann algebras acting on $\mc H_a$.
\end{pp}
We remind the reader that $\wtd I'$ is the clockwise complement of $\wtd I$.
\begin{proof}
We only prove $\mc B(\wtd I)'=\mc B'(\wtd I')$ as the other relation can be proved in a similar way. Note that by the previous proposition, we have $\mc B(\wtd I)'\supset\mc B'(\wtd I')$. To prove $\mc B(\wtd I)'\subset\mc B'(\wtd I')$, we choose any $Y\in\mc B(\wtd I)'$ and show that $Y\in\mc B'(\wtd I')$.

Set $\eta=Y\iota\Omega$. Since $\iota\in\Hom_{\mc A}(\mc H_0,\mc H_a)$ and  $Y\in\Hom_{\mc A(I)}(\mc H_a,\mc H_a)$ (by proposition \ref{lb4}), one has $Y\iota\in\Hom_{\mc A(I)}(\mc H_0,\mc H_a)$. Therefore $\eta\in\mc H_a(I')$.  Choose any $\xi\in\mc H_a(I)$. Then, by proposition \ref{lb3}, we have
\begin{align*}
&Y\xi=YA(\xi,\wtd I)\iota\Omega=A(\xi,\wtd I)Y\iota\Omega=A(\xi,\wtd I)\eta=A(\xi,\wtd I)B(\eta,\wtd I')\iota\Omega\\
=&B(\eta,\wtd I')A(\xi,\wtd I)\iota\Omega=B(\eta,\wtd I')\xi.
\end{align*}
This shows $Y=B(\eta,\wtd I')$ and hence that $Y\in\mc B'(\wtd I')$.
\end{proof}

\begin{pp}\label{lb3}
For any $\wtd I\in\Jtd$ and $\xi\in\mc H_i(I)$,
\begin{align}
A(\xi,\wtd I)\iota\Omega=\xi=B(\xi,\wtd I)\iota\Omega.\label{eq38}
\end{align}
\end{pp}

As a consequence, we see that $\iota\Omega$ is a cyclic separating vector for $\mc B(\wtd I)$ and $\mc B'(\wtd I)$.
\begin{proof}
We shall prove the following more general relations:
\begin{align}
A(\xi,\wtd I)\iota=L(\xi,\wtd I)|_{\mc H_0},\qquad B(\xi,\wtd I)\iota=R(\xi,\wtd I)|_{\mc H_0}.
\end{align}
Again, we only prove the first one as the second one can be argued similarly. We compute that
\begin{align*}
A(\xi,\wtd I)\iota=\mu\cdot L(\xi,\wtd I)\cdot\iota=\mu(\id_a\otimes\iota)L(\xi,\wtd I)|_{\mc H_0}=L(\xi,\wtd I)|_{\mc H_0}
\end{align*}
where we have used the functoriality of $\Ed$ and the unit property of $Q$.
\end{proof}

Finally, if $\mc A$ is conformal covariant, we notice that for any $g\in\GA$, we have $g\mc B(\wtd I)g^{-1}=\mc B(g\wtd I)$ and $g\mc B'(\wtd I)g^{-1}=\mc B'(g\wtd I)$. Indeed, we notice that the actions of $g$ commute with $\mu$ (see the discussion after \eqref{eq10}). Therefore the conformal covariance of $\Ed$ implies the two equations. If $\mc A$ is only M\"obius covariant, we also have similar relations for $g\in\UPSU$. We summarize the above results as follows. (Note that \eqref{eq11} follows from lemma \ref{lb3}.)

\begin{thm}\label{lb10}
$\mc B:\wtd I\in\Jtd\mapsto \mc B(\wtd I)$ and $\mc B':\wtd I\in\Jtd\mapsto \mc B'(\wtd I)$ are families of von Neumann algebras satisfying the following properties for any $\wtd I,\wtd J\in\Jtd$.\\
(a) (Extension property)  $\mc A(I)\subset\mc B(\wtd I)\cap\mc B'(\wtd I)$.\\
(b) (Isotony) If $\wtd I\subset\wtd J$, then $\mc B(\wtd I)\subset\mc B(\wtd J)$ and $\mc B'(\wtd I)\subset\mc B'(\wtd J)$.\\
(c) (Reeh-Schlieder property) $\mc B(\wtd I)\iota\Omega$ and $\mc B'(\wtd I)\iota\Omega$ are dense subspaces of $\mc H_a$. Indeed, we have
\begin{align}
\mc B(\wtd I)\iota\Omega=\mc B'(\wtd I)\iota\Omega=\mc H_a(I).\label{eq11}
\end{align}
(d) (Haag duality) $\mc B(\wtd I)'=\mc B'(\wtd I')$.\\
(e) (M\"obius/conformal covariance) For any $g\in\UPSU$ one has
\begin{align}
g\mc B(\wtd I)g^{-1}=\mc B(g\wtd I),\qquad g\mc B'(\wtd I)g^{-1}=\mc B'(g\wtd I).
\end{align}
When $\mc A$ is conformal covariance, the above relations are also true when $g\in\GA$.
\end{thm}
We say that $\mc B$ and $\mc B'$ are the \textbf{non-local extensions of $\mc A$} associated to the $C^*$-Frobenius algebra $Q$, and that $\mc B'$ is the \textbf{clockwise dual net} of $\mc B$. The latter name is justified by the Haag duality $\mc B(\wtd I)'=\mc B'(\wtd I')$.

Given $Q=(\mc H_a,\mu,\iota)$ and the associated non-local extensions $\mc B,\mc B'$, we define $Q'=(\mc H_a,\mu',\iota)$, where $\mu'=\mu\ss_{a,a}$.
\begin{pp}
The non-local extensions of $\mc A$ associated to the $C^*$-Frobenius algebra $Q'$ are $\mc B'$ and $\mc B''$, where $\mc B''$ is the clockwise dual net of $\mc B'$.
\end{pp}

\begin{proof}
We let $\mc B_1,\mc B_1'$ be the non-local extensions associated to $Q'$, and show that $\mc B_1=\mc B'$. The operators in $\mc B_1(\wtd I)$ are written as $\mu'L(\xi,\wtd I)|_{\mc H_a}$ where $\xi\in\mc H_a(I)$. By the braiding axiom of $\Ed$, $\mu'L(\xi,\wtd I)|_{\mc H_a}=\mu\ss L(\xi,\wtd I)|_{\mc H_a}=\mu R(\xi,\wtd I)|_{\mc H_a}$ which is inside $\mc B'(\wtd I)$. Moreover, any operator in $\mc B'(\wtd I)$ is of this form. This proves $\mc B_1(\wtd I)=\mc B'(\wtd I)$.
\end{proof}

We describe the relation between $\mc B$ and its  clockwise double dual net $\mc B''$. Let $\wtd I''$ be the clockwise complement of $\wtd I'$.  Then we have $I''=I$ and $\arg_{I''}=\arg_I-2\pi$.

\begin{pp}
$\mc B''(\wtd I'')=\mc B(\wtd I)$ for any $\wtd I\in\Jtd$.
\end{pp}
\begin{proof}
We have $\mc B(\wtd I)=\mc B'(\wtd I')'$  and, similarly, $\mc B'(\wtd I')=\mc B''(\wtd I'')'$.
\end{proof}

\section{Categorical modular operators and conjugations}

We first recall the Tomita-Takesaki theory for von Neumann algebras associated with cyclic seperating vectors; details can be found in \cite{Tak02} or \cite{Tak70}. Let $\mc M$ be a von Neumann algebra acting on a Hilbert space $\mc H$, and assume that $\Omega\in\mc H$ is a cyclic and separating vector of $\mc M$. (We do not require $\lVert\Omega\lVert=1$.) One defines unbounded antilinear operators $S:\mc M\Omega\rightarrow\mc M\Omega$ and $F:\mc M'\Omega\rightarrow\mc M'\Omega$ such that for any $x\in\mc M,y\in\mc M'$,
\begin{align*}
Sx\Omega=x^*\Omega,\qquad Fy\Omega=y^*\Omega.
\end{align*}
$S$ and $F$ are indeed preclosed operators, whose closures are also denoted by the same symbols $S$ and $F$ respectively. Moreover, $S^*=F$. Let $S=\fk J\Delta^{\frac 12}$ be the polar decomposition of $S$, where the positive operator $\Delta=S^*S$ is called the modular operator, and the antiunitary map $\fk J$ is called the modular conjugation. We have $\Delta^{\im t}\Omega=\fk J\Omega=\Omega,S=S^{-1}$, $\fk J^2=\id$,  $S=\fk J\Delta^{\frac 12}=\Delta^{-\frac 12}\fk J$. Let $\im=\sqrt{-1}$. For any $t\in\mbb R$, we have
\begin{gather*}
\Delta^{\im t}\mc M\Delta^{-\im t}=\mc M, \qquad \fk J\mc M\fk J=\mc M'.
\end{gather*}

Tomita-Takesaki theory can be applied to non-local extensions without difficulty. This will be used to derive a categorical Tomita-takesaki theory in this section. 

\begin{df}
For each $\mc H_i\in\Obj(\RepdA)$ with dualizing data $(\mc H_{\ovl i},\ev_{i,\ovl i})$ (see Section \ref{lb30}), and for any $\wtd I\in\Jtd$, we define unbounded antilinear operators $S_{\wtd I},F_{\wtd I}:\mc H_i\rightarrow\mc H_{\ovl i}$ with domains $\mc H_i(I)$ such that for any $\xi\in\mc H_i(I)$,
\begin{equation}
	S_{\wtd I}\xi=L(\xi,\wtd I)^*\coev_{i,\ovl i}\Omega,\qquad F_{\wtd I}\xi=R(\xi,\wtd I)^*\coev_{\ovl i,i}\Omega\label{eq12}.
\end{equation}
\end{df}
Recall (cf. Definition \ref{lb31}) that $\mc H_{\ovl i}$ is assumed to have dualizing data $(\mc H_i,\ev_{\ovl i,i})$.

\begin{rem}
We also understand $S_{\wtd I}$ and $F_{\wtd I}$ as categorical operators, which means that they can act on any object of $\RepdA$. We write $S_{\wtd I}$ as $S_{\wtd I}|_{\mc H_i}$ and similarly $F_{\wtd I}$ as $F_{\wtd I}|_{\mc H_i}$ if we want to emphasize that $S_{\wtd I}$ and $F_{\wtd I}$ are acting on the object $\mc H_i\in\RepdA$. Then their domains are $\mc H_i(I)$. Later we will show that $S_{\wtd I}|_{\mc H_i},F_{\wtd I}|_{\mc H_i}$ are preclosed (with the same absolute value), and will use the same symbols to denote the closures. Then the domains of $S_{\wtd I}|_{\mc H_i},F_{\wtd I}|_{\mc H_i}$ will be a dense subspace of $\mc H_i$ containing $\mc H_i(I)$.
\end{rem}

The following proposition describes how $S_{\wtd I},F_{\wtd I}$ depends on the dualizing data.

\begin{pp}\label{lb8}
Suppose that $\mc H_i,\mc H_j\in\Obj(\RepdA)$ are equal as representations of $\mc A$, with possibly different dualizing data  given by $(\mc H_{\ovl i},\ev_{i,\ovl i})$ and $(\mc H_{\ovl j},\ev_{j,\ovl j})$ respectively. By the uniqueness of standard evaluations (cf. Section \ref{lb30}), there is a unique unitary $u\in\Hom_{\mc A}(\mc H_{\ovl i},\mc H_{\ovl j})$ such that
\begin{align*}
	\ev_{j,\ovl j}=\ev_{i,\ovl i}(\id_i\otimes u^{-1}).
\end{align*}
Then for each $\wtd I\in\Jtd$,
\begin{align*}
S_{\wtd I}|_{\mc H_j}=uS_{\wtd I}|_{\mc H_i},\qquad F_{\wtd I}|_{\mc H_j}=uF_{\wtd I}|_{\mc H_i}.
\end{align*}
\end{pp}

Note that we do not assume $\mc H_{\ovl i}$ and $\mc H_{\ovl j}$ are the same objects of $\RepA$, even though $\mc H_i,\mc H_j$ are so.

\begin{proof}
Recall the functoriality: for any $\xi\in\mc H_i(I)=\mc H_j(I)$ and $\eta\in\mc H_{\ovl j}$, we have $(\id_i\otimes u^{-1})L(\xi,\wtd I)=L(\xi,\wtd I)u^{-1}$ and $(u^{-1}\otimes\id_i)R(\xi,\wtd I)=R(\xi,\wtd I)u^{-1}$ when acting on $\mc H_{\ovl j}$. Therefore
\begin{align*}
&S_{\wtd I}|_{\mc H_j}\xi=L(\xi,\wtd I)^*\coev_{j,\ovl j}\Omega=	L(\xi,\wtd I)^*(\id_i\otimes u)\coev_{i,\ovl i}\Omega\\
=&uL(\xi,\wtd I)^*\coev_{i,\ovl i}\Omega=uS_{\wtd I}|_{\mc H_i}\xi.
\end{align*}
The relation for $F$-operators is proved in the same way.
\end{proof}

We shall show that  $S_{\wtd I}$ and $F_{\wtd I}$ are involutions. First of all, we need:

\begin{pp}\label{lb7}
Choose $\mc H_i\in\Obj(\RepdA)$. If $\xi\in\mc H_i(I)$, then $S_{\wtd I}\xi\in\mc H_{\ovl i}(I),F_{\wtd I}\xi\in\mc H_{\ovl i}(I)$. Moreover, for any $\mc H_k\in\Obj(\RepA)$ we have
\begin{gather}
L(S_{\wtd I}\xi,\wtd I)|_{\mc H_k}=L(\xi,\wtd I)^*(\coev_{i,\ovl i}\otimes\id_k),\label{eq13}\\
R(F_{\wtd I}\xi,\wtd I)|_{\mc H_k}=R(\xi,\wtd I)^*(\id_k\otimes\coev_{\ovl i,i}).\label{eq14}
\end{gather}
\end{pp}
Note that in the above two equations, $L(\xi,\wtd I)^*$ is a bounded linear operator from $\mc H_i\boxtimes\mc H_{\ovl i}\boxtimes\mc H_k$ to $\mc H_{\ovl i}\boxtimes\mc H_k$, and $R(\xi,\wtd I)^*$ from $\mc H_k\boxtimes\mc H_{\ovl i}\boxtimes\mc H_i$ to $\mc H_k\boxtimes\mc H_{\ovl i}$. Also, $\mc H_k$ is not assumed to be dualizable or dualized. Then $\mc H_i\boxtimes\mc H_k$ and $\mc H_k\boxtimes\mc H_i$ have no dualizing data if $\mc H_k$ is not dualized.

\begin{proof}
That $S_{\wtd I}\xi$ and $F_{\wtd I}\xi$ are inside $\mc H_{\ovl i}(I)$ follows from the obvious fact that $L(\xi,\wtd I)^*\coev_{i,\ovl i}$ and $R(\xi,\wtd I)^*\coev_{\ovl i,i}$ are in $\Hom_{\mc A(I')}(\mc H_0,\mc H_{\ovl i})$. For any $\chi\in\mc H_k$, we use proposition \ref{lb5} to compute that
\begin{align*}
&L(S_{\wtd I}\xi,\wtd I)\chi=L\big(L(\xi,\wtd I)^*\coev_{i,\ovl i}\Omega,\wtd I\big)\chi=L(\xi,\wtd I)^*L\big(\coev_{i,\ovl i}\Omega,\wtd I\big)\chi\\
\xlongequal{\eqref{eq8}}&L(\xi,\wtd I)^*(\coev_{i,\ovl i}\otimes\id_k)L\big(\Omega,\wtd I\big)\chi\xlongequal{\eqref{eq9}}L(\xi,\wtd I)^*(\coev_{i,\ovl i}\otimes\id_k)\chi.
\end{align*}
The other equation is proved similarly.	
\end{proof}

\begin{pp}\label{lb6}
For any $\mc H_i\in\Obj(\RepdA)$ and $\xi\in\mc H_i(I)$, we have $S_{\wtd I}^2\xi=F_{\wtd I}^2\xi=\xi$.
\end{pp}

\begin{proof}
We compute
\begin{align*}
&S_{\wtd I}^2\xi\xlongequal{\eqref{eq12}}L(S_{\wtd I}\xi,\wtd I)^*\coev_{\ovl i,i}\Omega\xlongequal{\eqref{eq13}} (\ev_{i,\ovl i}\otimes\id_i)L(\xi,\wtd I)\coev_{\ovl i,i}\Omega\\
=&(\ev_{i,\ovl i}\otimes\id_i)(\id_i\otimes\coev_{\ovl i,i})L(\xi,\wtd I)\Omega\xlongequal{\eqref{eq3}} L(\xi,\wtd I)\Omega=\xi.
\end{align*}
Similarly, we may use \eqref{eq4} to show $F_{\wtd I}^2\xi=\xi$.
\end{proof}

The above two propositions imply immediately the following result.

\begin{co}
For any $\mc H_i\in\Obj(\RepdA),\mc H_k\in\Obj(\RepA),\xi\in\mc H_i(I),\phi\in\mc H_i\boxtimes\mc H_k,\psi\in\mc H_k\boxtimes\mc H_i$,
\begin{gather}
L(\xi,\wtd I)^*\phi=(\ev_{\ovl i,i}\otimes\id_k)L(S_{\wtd I}\xi,\wtd I)\phi,\label{eq16}\\
R(\xi,\wtd I)^*\psi=(\id_k\otimes\ev_{i,\ovl i})R(F_{\wtd I}\xi,\wtd I)\psi.\label{eq17}
\end{gather}
\end{co}

Next, we relate $S_{\wtd I}$ and $F_{\wtd I}$. 

\begin{pp}\label{lb28}
We have
\begin{align}
F_{\wtd I}=\vartheta S_{\wtd I}.
\end{align}
More precisely, for any $\mc H_i\in\Obj(\RepdA)$ and $\xi\in\mc H_i(I)$ we have $F_{\wtd I}\xi=\vartheta_{\ovl i}S_{\wtd I}\xi$.
\end{pp}
Recall that $\vartheta_{\ovl i}$ is the (unitary) twist operator of $\mc H_{\ovl i}$.
\begin{proof}
By the braiding axiom of $\Ed$ we have $R(\xi,\wtd I)|_{\mc H_{\ovl i}}=\ss_{i,\ovl i}L(\xi,\wtd I)|_{\mc H_{\ovl i}}$. Therefore
\begin{align*}
&F_{\wtd I}\xi=R(\xi,\wtd I)^*\coev_{\ovl i,i}\Omega=L(\xi,\wtd I)^*\ss_{i,\ovl i}^{-1}\coev_{\ovl i,i}\Omega \xlongequal{\eqref{eq15}} L(\xi,\wtd I)^*(\id_i\otimes\vartheta_{\ovl i})\coev_{i,\ovl i}\Omega\\
=&\vartheta_{\ovl i}L(\xi,\wtd I)^*\coev_{i,\ovl i}\Omega=\vartheta_{\ovl i}S_{\wtd I}\xi.
\end{align*}
\end{proof}

We will see later that $S_{\wtd I}$ and $F_{\wtd I}$ are preclosed. Therefore $S_{\wtd I}^*S_{\wtd I}=F_{\wtd I}^*F_{\wtd I}$ will be denoted by $\Delta_{\wtd I}$. This fact is  crucial for proving the geometric modular theorems. 

We now show the M\"obius covariance of $S_{\wtd I}$ and $F_{\wtd I}$.

\begin{pp}\label{lb12}
For any $g\in\UPSU$,
\begin{align}
gS_{\wtd I}g^{-1}=S_{g\wtd I},\qquad gF_{\wtd I}g^{-1}=F_{g\wtd I}.
\end{align}	
\end{pp}	
\begin{proof}
The domain of $S_{g\wtd I}$ is $\mc H_i(gI)$, whereas the domain of $gS_{\wtd I}g^{-1}$ is $g\mc H_i(I)$. From \eqref{eq20} one clearly has $g\mc H_i(I)=\mc H_i(gI)$. Now choose any $\xi\in\mc H_i(gI)$. Then $g^{-1}\xi\in\mc H_i(I)$, and $L(\xi,g\wtd I)=gL(g^{-1}\xi,\wtd I)g^{-1}$ by the M\"obius covariance of $\Ed$. Notice that $\Omega$ is $\UPSU$-invariant, and that the morphisms intertwine the actions of $\UPSU$. Therefore
\begin{align*}
&S_{g\wtd I}\xi=L(\xi,g\wtd I)^*\coev_{i,\ovl i}\Omega=gL(g^{-1}\xi,\wtd I)^*g^{-1}\coev_{i,\ovl i}\Omega\\
=&gL(g^{-1}\xi,\wtd I)^*\coev_{i,\ovl i}g^{-1}\Omega=gL(g^{-1}\xi,\wtd I)^*\coev_{i,\ovl i}\Omega=gS_{\wtd I}g^{-1}\xi.
\end{align*}
\end{proof}

The following proposition says that we can use morphisms to relate the $S_{\wtd I}$ and $F_{\wtd I}$ on different dualized objects. Recall the conjugates of morphisms defined by \eqref{eq58}.

\begin{pp}\label{lb32}
For any $\mc H_i,\mc H_j\in\Obj(\RepdA)$ with (canonical) dual objects $\mc H_{\ovl i},\mc H_{\ovl j}$ respectively,  and for any $G\in\Hom_{\mc A}(\mc H_i,\mc H_j)$, we have
\begin{align*}
\ovl G\cdot S_{\wtd I}|_{\mc H_i}\subset S_{\wtd I}|_{\mc H_j}\cdot G,\qquad \ovl G\cdot F_{\wtd I}|_{\mc H_i}\subset F_{\wtd I}|_{\mc H_j}\cdot G,
\end{align*}
where the conjugate  of $G\in\Hom_{\mc A}(\mc H_i,\mc H_j)$ is defined using the dualizing data of $\mc H_i,\mc H_j$.
\end{pp}

\begin{proof}
Notice $G\mc H_i(I)\subset\mc H_j(I)$ and $\ovl G\mc H_{\ovl i}(I)\subset\mc H_{\ovl j}(I)$. For any $\xi\in\mc H_i(I)$, we have $L(G\xi,\wtd I)=(G\otimes\id)L(\xi,\wtd I)$ by \eqref{eq8}. Therefore
\begin{align*}
	&S_{\wtd I}|_{\mc H_j}G\xi=L(G\xi,\wtd I)^*\coev_{j,\ovl j}\Omega=L(\xi,\wtd I)^*(G^*\otimes\id_{\ovl j})\coev_{j,\ovl j}\Omega\\
	\xlongequal{\eqref{eq36}}&L(\xi,\wtd I)^*(\id_i\otimes\ovl G)\coev_{i,\ovl i}\Omega=\ovl GL(\xi,\wtd I)^*\coev_{i,\ovl i}\Omega=\ovl G S_{\wtd I}|_{\mc H_i}\xi.
\end{align*}
We conclude $\ovl G\cdot S_{\wtd I}|_{\mc H_i}\subset S_{\wtd I}|_{\mc H_j}\cdot G$. The second relation follows from a similar argument or from Proposition \ref{lb28} and the fact that the twist intertwines morphisms.
\end{proof}

To prove further properties of $S_{\wtd I}$ and $F_{\wtd I}$, we have to relate them with the $S$ and $F$ operators of non-local extensions of $\mc A$. First of all, if $Q=(\mc H_a,\mu,\iota)$ is a $C^*$-Frobenius algebra in $\RepdA$, then $\mc H_a$ has a dualizing data $(\mc H_{\ovl a},\ev_{a,\ovl a})$. On the other hand, notice that $\mc H_a$ is self dual and $\ev_{a,a}:=\iota^*\mu$ defines an evaluation satisfying the conjugate equation
\begin{align*}
(\ev_{a,a}\otimes\id_a)(\id_a\otimes\coev_{a,a})=\id_a=(\id_a\otimes\ev_{a,a})(\coev_{a,a}\otimes\id_a).
\end{align*}
We say that $Q$ is \textbf{standard} if $\ev_{a,a}$ is a standard evaluation. Therefore, when $Q$ is standard,  $(\mc H_a,\ev_{a,a}=\iota^*\mu)$ is another dualizing data for the representation $\mc H_a$ of $\mc A$. By the uniqueness of standard evaluations, we have a unitary $\epsilon\in\Hom_{\mc A}(\mc H_a,\mc H_{\ovl a})$ satisfying
\begin{align}
	\ev_{\ovl a,a}(\epsilon\otimes\id_a)=\iota^*\mu=\ev_{a,\ovl a}(\id_a\otimes\epsilon),\label{eq18}
\end{align}
called the \textbf{reflection operator} of $Q$.

In the remaining part of this section, we shall always assume that $Q$ is a standard  $C^*$-Frobenius algebra in $\RepdA$, and $\mc H_a$ has dualizing data $(\mc H_{\ovl a},\ev_{a,\ovl a})$ which defines the reflection operator $\epsilon$. Let $\mc B$ and $\mc B'$ be the pair of non-local extensions of $\mc A$ obtained by $Q$. Then, according to Proposition \ref{lb3}, for each $\wtd I\in\Jtd$, $\iota\Omega$ is a cyclic separating vector for $\mc B(\wtd I)$ and $\mc B'(\wtd I)$. 

\begin{pp}\label{lb9}
Assume that $Q$ is standard. Then $\epsilon^{-1}S_{\wtd I}|_{\mc H_a}$ and $\epsilon^{-1}F_{\wtd I}|_{\mc H_a}$ are respectively the (preclosed) $S$ operators of $\mc B(\wtd I)$ and $\mc B'(\wtd I)$ with respect to $\iota\Omega$. More precisely, for any $X\in\mc B(\wtd I)$ and $Y\in\mc B'(\wtd I)$,
\begin{align}
\epsilon^{-1}S_{\wtd I}|_{\mc H_a}\cdot X\iota\Omega=X^*\iota\Omega,\qquad \epsilon^{-1}F_{\wtd I}|_{\mc H_a}\cdot Y\iota\Omega=Y^*\iota\Omega.
\end{align}
\end{pp}
\begin{proof}
We prove the second equation. To simplify discussions, we suppress the $|_{\mc H_a}$ after $F_{\wtd I}$. Choose any $\xi\in\mc H_a(I)$. We want to show that $F_{\wtd I}B(\xi,\wtd I)\iota\Omega=\epsilon B(\xi,\wtd I)^*\iota\Omega$. By proposition \ref{lb3} we have $F_{\wtd I}B(\xi,\wtd I)\iota\Omega=F_{\wtd I}\xi$. On the other hand,
\begin{align*}
&\epsilon B(\xi,\wtd I)^*\iota\Omega=\epsilon R(\xi,\wtd I)^*\mu^*\iota\Omega \xlongequal{\eqref{eq17}} \epsilon(\id_a\otimes\ev_{a,\ovl a})R(F_{\wtd I}\xi,\wtd I)\mu^*\iota\Omega\\
=&\epsilon(\id_a\otimes\ev_{a,\ovl a})(\mu^*\iota\otimes\id_{\ovl a})R(F_{\wtd I}\xi,\wtd I)\Omega=\epsilon(\id_a\otimes\ev_{a,\ovl a})(\mu^*\iota\otimes\id_{\ovl a})\cdot F_{\wtd I}\xi.
\end{align*}

By \eqref{eq18} we have $(\epsilon\otimes\id_a)\mu^*\iota=\coev_{\ovl a,a}$. Therefore
\begin{align*}
&\epsilon(\id_a\otimes\ev_{a,\ovl a})(\mu^*\iota\otimes\id_{\ovl a})=(\epsilon\otimes\ev_{a,\ovl a})(\mu^*\iota\otimes\id_{\ovl a})\\
=&(\id_{\ovl a}\otimes\ev_{a,\ovl a})(\epsilon\otimes\id_a\otimes\id_{\ovl a})(\mu^*\iota\otimes\id_{\ovl a})=(\id_{\ovl a}\otimes\ev_{a,\ovl a})((\epsilon\otimes\id_a)\mu^*\iota\otimes\id_{\ovl a})\\
=&(\id_{\ovl a}\otimes\ev_{a,\ovl a})(\coev_{\ovl a,a}\otimes\id_{\ovl a})=\id_{\ovl a}.
\end{align*}
This proves the second equation. A similar argument proves the first one.
\end{proof}

\begin{rem}\label{lb33}
Since $\vartheta$ commutes with homomorphisms, it commutes in particular with $\epsilon$. Therefore the $S$ operators of $\mc B(\wtd I)$ and $\mc B'(\wtd I)$ differ by a twist: $\epsilon^{-1}F_{\wtd I}|_{\mc H_a}=\vartheta\cdot\epsilon^{-1}S_{\wtd I}|_{\mc H_a}$.
\end{rem}

By Tomita-Takesaki theory \cite[Chapter VI]{Tak02}, we know that $S_{\wtd I}|_{\mc H_a}$ and $F_{\wtd I}|_{\mc H_a}$ are preclosed since $\epsilon^{-1} S_{\wtd I}|_{\mc H_a}$ and $\epsilon ^{-1}F_{\wtd I}|_{\mc H_a}$ are so. Using this fact, we can show
\begin{lm}
Let $Q=(\mc H_a,\mu,\iota)$ be a standard $C^*$-Frobenius algebra in $\RepdA$, and let $\mc H_i\in\Obj(\RepdA)$ be equivalent to a subrepresentation of $\mc H_a$.  Then $S_{\wtd I}|_{\mc H_i}, F_{\wtd I}|_{\mc H_i}$ are preclosed. 
\end{lm}

\begin{proof}
Since the proofs for the two operators are the same, we only prove the preclosedness of $S_{\wtd I}|_{\mc H_i}$. By Proposition \ref{lb9} and Tomita-Takesaki theory for the von Neumann algebras $\mc B(\wtd I)$, we know that $S_{\wtd I}|_{\mc H_a}$ is preclosed. Let $\varphi\in\End_{\mc A}(\mc H_i,\mc H_a)$ be an isometry embedding $\mc H_i$ into $\mc \mc H_a$. Let $\mc H_{\ovl i}$ be the (canonical) dual object of $\mc H_i$.  Since $\varphi^*\varphi=\id_i$, we have $\ovl\varphi^*\cdot \ovl\varphi=\id_{\ovl i}$.

By Proposition \ref{lb32}, we have $\ovl \varphi\cdot S_{\wtd I}|_{\mc H_i}\subset S_{\wtd I}|_{\mc H_a}\cdot \varphi$. Let $\bk{\cdot|\cdot}$ denote the inner product of $\mc H_{\ovl i}$. Then for each $\xi,\eta\in\mc H_i(I)$,
\begin{align*}
\Big\langle S_{\wtd I}|_{\mc H_i}\xi\Big|S_{\wtd I}|_{\mc H_i}\eta\Big\rangle\xlongequal{\ovl\varphi^*\cdot \ovl\varphi=\id_{\ovl i}}\Big\langle \ovl\varphi S_{\wtd I}|_{\mc H_i}\xi\Big|\ovl\varphi S_{\wtd I}|_{\mc H_i}\eta\Big\rangle=\Big\langle S_{\wtd I}|_{\mc H_a}\varphi\xi\Big|S_{\wtd I}|_{\mc H_a}\varphi\eta\Big\rangle.
\end{align*}
Choose any sequence $\xi_n$ in $\mc H_i(I)$ such that $\xi_n\rightarrow 0$ and $S_{\wtd I}|_{\mc H_i}\xi_n$ converges. Then $\varphi\xi_n\rightarrow 0$.  By the above relation, $S_{\wtd I}|_{\mc H_a}\varphi\xi$ converges. So, by the preclosedness of $S_{\wtd I}|_{\mc H_a}$, $S_{\wtd I}|_{\mc H_a}\varphi\xi$ converges to $0$. By the above relation again, we conclude $S_{\wtd I}|_{\mc H_i}\xi\rightarrow 0$. So $S_{\wtd I}|_{\mc H_i}$ must be preclosed.
\end{proof}

Thus, to prove that $S_{\wtd I}$ and $F_{\wtd I}$ are preclosed on any dualized object $\mc H_i$, it suffices to show that any $\mc H_i$ can be embedded into $\mc H_a$  for some standard $C^*$-Frobenius algebra $Q=(\mc H_a,\mu,\iota)$. This is well-known. We review such construction below, and recall some important properties that will be used in later sections.

First, assume $\mc H_k\in\Obj(\RepdA)$ with dualizing data $(\mc H_{\ovl k},\ev_{k,\ovl k})$.  Then $Q=(\mc H_k\boxtimes\mc H_{\ovl k},\mu,\iota)$ is a standard  $C^*$-Frobenius algebra, where
\begin{align}
	\iota=\coev_{k,\ovl k},\qquad \mu=\id_k\otimes\ev_{\ovl k,k}\otimes\id_{\ovl k}.\label{eq49}
\end{align}
Moreover, the dualizing data of $\mc H_k\boxtimes\mc H_{\ovl k}$ defined as in  Definition \ref{lb31} is clearly $(\mc H_k\boxtimes\mc H_{\ovl k},\iota^*\mu)$. Therefore, the reflection operator is just
\begin{align*}
	\epsilon=\id_{k\otimes\ovl k}.
\end{align*}

Now assume that $\{\mc H_i:i\in\mc E\}$ is a finite set of distinct objects in $\RepdA$, indexed by $\mc E$. Let $(\mc H_{\ovl i},\ev_{i,\ovl i})$ be the dualizing data of $\mc H_i$. Assume $\mc H_k=\bigoplus_{i\in\mc E}\mc H_i$, and let the morphism $\varphi_i:\mc H_i\rightarrow\mc H_k$ be the inclusion. Fix a dualizing data $(\mc H_{\ovl k},\ev_{k,\ovl k})$ for $\mc H_k$. Since $\ovl{\varphi_j}=(\varphi_j^*)^\vee$ and hence $\ev_{k,\ovl k}(\varphi_i\otimes\ovl{\varphi_j})=\ev_{k,\ovl k}(\varphi_j^*\varphi_i\otimes\id_{\ovl j})=\delta_{i,j}\ev_{i,\ovl i}$, we clearly have
\begin{gather}
	\qquad \mc H_{\ovl k}\simeq\bigoplus_{i\in\mc E}\mc H_{\ovl i},\qquad \ev_{k,\ovl k}=\sum_{i\in\mc E}\ev_{i,\ovl i}(\varphi_i^*\otimes\ovl{\varphi_i}^*),\label{eq51}
\end{gather}
because the right hand side of the second equation of \eqref{eq51} times $(\varphi_i\otimes\ovl{\varphi_j})$ also equals $\delta_{i,j}\ev_{i,\ovl i}$. Let $Q$ be the corresponding standard $C^*$-Frobenius algebra for $\mc H_k\boxtimes\mc H_{\ovl k}$ defined by \eqref{eq49}. Then we have a unitary equivalence of $\mc A$-representations
\begin{align}
\bigoplus_{i,j\in\mc E}\psi_{i,\ovl j}:\bigoplus_{i,j\in\mc E}\mc H_i\boxtimes\mc H_{\ovl j}\xrightarrow{\simeq}	\mc H_k\boxtimes\mc H_{\ovl k}\label{eq19}
\end{align}
where
\begin{align}
\psi_{i,\ovl j}=\varphi_i\otimes\ovl{\varphi_j}\label{eq64}
\end{align}
is an isometry. We have
\begin{align}
	\ovl{\psi_{i,\ovl j}}=\psi_{j,\ovl i}.\label{eq65}
\end{align}
Moreover, using \eqref{eq49} and \eqref{eq51}, we see that for each $i,j,l,m\in\mc E$, 
\begin{align}
	\mu(\psi_{i,\ovl j}\otimes\psi_{l,\ovl m})=\delta_{j,l}\cdot \psi_{i,\ovl m}(\id_i\otimes\ev_{\ovl j,j}\otimes\id_{\ovl m}),\label{eq48}
\end{align}
and hence
\begin{align}
\iota^*\mu(\psi_{i,\ovl j}\otimes\psi_{j,\ovl l})=\delta_{i,l}\cdot \ev_{i\boxtimes\ovl j, j\boxtimes\ovl l}.\label{eq70}
\end{align}

In the case that $0\in\mc E$, i.e., $\{\mc H_i:i\in\mc E\}$ contains $\mc H_0$. Notice that $\mc H_{\ovl 0}=\mc H_0$. Then for each $i,j\in\mc E$, by identifying $\mc H_i$ with $\mc H_i\boxtimes\mc H_0$ and $\mc H_{\ovl j}$ with $\mc H_0\boxtimes\mc H_{\ovl j}$ using the right and left unitors respectively and noticing \eqref{eq48}, we obtain an isometry
\begin{gather}
\mu(\psi_{i,0}\otimes\psi_{0,\ovl j})=\psi_{i,\ovl j}:\mc H_i\boxtimes\mc H_{\ovl j}\rightarrow \mc H_k\boxtimes\mc H_{\ovl k}.	\label{eq66}
\end{gather}


Now, for any $\mc H_i\in\Obj(\RepdA)$, we set $\mc H_k=\mc H_0\bigoplus\mc H_i$. Recall $\mc H_{\ovl 0}=\mc H_0$. Then $\mc H_k\boxtimes\mc H_{\ovl k}$ defines a standard $C^*$-Frobenius algebra. By \eqref{eq19}, $\mc H_i\simeq\mc H_i\boxtimes\mc H_0$ is a sub-representation of $\mc H_k\boxtimes\mc H_{\ovl k}$. Thus we conclude:

\begin{pp}
$S_{\wtd I}$ and $F_{\wtd I}$ are preclosed operators on any $\mc H_i\in\Obj(\RepdA)$.
\end{pp}

\begin{rem}
In the following, we will always let $S_{\wtd I},F_{\wtd I}$  denote the closures of the preclosed operators in \eqref{eq12}. Then it is clear that Propositions \ref{lb8}, \ref{lb28}, \ref{lb12}, \ref{lb32} and Remark \ref{lb33} still hold for $S_{\wtd I},F_{\wtd I}$. The original operators (without taking closures) will be denoted by $S_{\wtd I}|_{\mc H_i(I)},F_{\wtd I}|_{\mc H_i(I)}$ (whose domains are precisely $\mc H_i(I)$).
\end{rem}

We have a positive closed operator $\Delta_{\wtd I}:=S_{\wtd I}^*S_{\wtd I}=F_{\wtd I}^*F_{\wtd I}$ definable on any object $\mc H_i\in\Obj(\RepdA)$. (We will write $\Delta_{\wtd I}$ as $\Delta_I$ after showing that $\Delta_{\wtd I}$ is independent of the choice of $\arg_I$.) So $\Delta_{\wtd I}|_{\mc H_i}$ is a positive closed operator on $\mc H_i$. We call $\Delta_{\wtd I}$ the  \textbf{modular operator} of $\Ed$. Note that by proposition \ref{lb8}, $\Delta_{\wtd I}|_{\mc H_i}$ does not rely on the dualizing data of $\mc H_i$. We define the categorical domain $\Dom(\Delta_{\wtd I}^{\frac 12})$ which associates to each $\mc H_i\in\Obj(\RepdA)$ the dense subspace $\Dom(\Delta_{\wtd I}^{\frac 12}|_{\mc H_i})$ of $\mc H_i$. Then $\Dom(\Delta_{\wtd I}^{\frac 12})=\Dom(S_{\wtd I})=\Dom(F_{\wtd I})$. We also have (categorical) polar decompositions
\begin{align}
	S_{\wtd I}=\fk J_{\wtd I}\cdot\Delta_{\wtd I}^{\frac 12},\qquad F_{\wtd I}=\vartheta\fk J_{\wtd I}\cdot\Delta_{\wtd I}^{\frac 12}\label{eq41}
\end{align}
where $\fk J_{\wtd I}$, when restricted to any $\mc H_i\in\Obj(\RepdA)$, is an anti-unitary operator
\begin{align*}
	\fk J_{\wtd I}|_{\mc H_i}:\mc H_i\rightarrow\mc H_{\ovl i}.
\end{align*}
We call $\fk J_{\wtd I}$ the  \textbf{(left) modular conjugation} of $\Ed$. (The right modular conjugation is $\vartheta\fk J_{\wtd I}$ for the obvious reason.) 

Our next goal is to use Proposition \ref{lb32} to study the relations of modular operators and conjugations on different dualized objects. To prepare for the proof, we recall that two closed operators $A$ and $B$ on a Hilbert space $\mc H$ are said to \textbf{commute strongly} if the von Neumann algebras generated by $A$ and by $B$ commute.\footnote{The von Neumann algebra generated by $A$ is the one generated by $U$ and all $e^{\im tH}$ where $A=UH$ is the polar decomposition of $A$.} If $A$ is bounded, then $A$ and $B$ commute strongly if and only if
\begin{align}
AB\subset BA,\qquad A^*B\subset BA^*.	
\end{align}
(Cf. for instance \cite[Sec. B.1]{Gui19}.) We also recall the definition of strongly commuting diagrams of closed operators introduced in \cite{Gui21a}.
\begin{df}
	Let $\mc P,\mc Q, \mc R,\mc S$ be Hilbert spaces, and  $A:\mc P\rightarrow\mc R,B:\mc Q\rightarrow\mc S,C:\mc P\rightarrow\mc Q,D:\mc R\rightarrow\mc S$ be unbounded closed operators. By saying that the diagram of closed operators
	\begin{align}
		\begin{CD}
			\mc P @>C>> \mc Q\\
			@V A VV @V B VV\\
			\mc R @>D>> \mc S
		\end{CD}\label{eq37}
	\end{align}
	\textbf{commutes strongly}, we mean the following: Let $\mc H=\mc P\oplus\mc Q\oplus\mc R\oplus\mc S$. Define closed operators $R,S$ on $\mc H$ with domains $\Dom(R)=\Dom(A)\oplus\Dom(B)\oplus\mc R\oplus \mc S$, $\Dom(S)=\Dom(C)\oplus\mc Q\oplus\Dom(D)\oplus \mc S$, such that
	\begin{gather*}
		R(\xi\oplus\eta\oplus\chi\oplus\varsigma)=0\oplus 0\oplus A\xi\oplus B\eta\qquad(\forall \xi\in\Dom(A),\eta\in\Dom(B),\chi\in\mc R,\varsigma\in \mc S),\\
		S(\xi\oplus\eta\oplus\chi\oplus\varsigma)=0\oplus C\xi\oplus 0\oplus D\chi   \qquad(\forall \xi\in\Dom(C),\eta\in\mc Q,\chi\in \Dom(D),\varsigma\in\mc S).
	\end{gather*}
	(Such construction is called the \textbf{extension} from $A,B$ to $R$, and from $C,D$ to $S$.) Then  $R$ and $S$ commute strongly. In the case that $A$ and $B$ are preclosed antilinear operators, we choose anti-unitary operators $U_1$ on $\mc R$ and $U_2$ on $\mc S$. We say that \eqref{eq37} commutes strongly if the following diagram of closed linear operators commutes strongly:
	\begin{align}
		\begin{CD}
			\mc P @>C>> \mc Q\\
			@V U_1A VV @V U_2B VV\\
			\mc R @>U_2DU_1^{-1}>> \mc S
		\end{CD}.
	\end{align}
	This definition is independent of the choice of $U_1,U_2$.
\end{df}

\begin{pp}\label{lb14}
For any $\alpha\in\mbb C$, $\mc H_i,\mc H_j\in\Obj(\RepdA)$,	and  $G\in\Hom_{\mc A}(\mc H_i,\mc H_j)$, the following equation holds when acting on $\mc H_i$.
\begin{gather}
G\cdot \Delta_{\wtd I}^\alpha|_{\mc H_i}\subset \Delta_{\wtd I}^\alpha|_{\mc H_j}\cdot G,\qquad  \ovl G\cdot\fk J_{\wtd I}|_{\mc H_i}=\fk J_{\wtd I}|_{\mc H_j}\cdot G.\label{eq60}
\end{gather}
Moreover, if $G$ is an isometry (i.e. $G^*G=\id_i$), then for each $I\in\mc J$,
\begin{align}
G\cdot \Delta_{\wtd I}^{\frac 12}|_{\mc H_i(I)}=\Delta_{\wtd I}^{\frac 12}|_{\mc H_j(I)}\cdot G,\qquad G\cdot \Delta_{\wtd I}^\alpha|_{\mc H_i}=\Delta_{\wtd I}^\alpha|_{\mc H_j}\cdot G.\label{eq61}
\end{align}
\end{pp}

Thus, in the case that $G$ is an isometry, we conclude
\begin{align}
\ovl G\cdot S_{\wtd I}|_{\mc H_i}=S_{\wtd I}|_{\mc H_j}\cdot G,\qquad 	\ovl G\cdot F_{\wtd I}|_{\mc H_i}=F_{\wtd I}|_{\mc H_j}\cdot G.
\end{align}

\begin{proof}
Recall Proposition \ref{lb32}. Since $\ovl{G^*}=G^\vee={\ovl G}^*$, we also have $G^*S_{\wtd I}|_{\mc H_j}\subset S_{\wtd I}|_{\mc H_i}{\ovl G}^*$. Therefore the following diagram of closed operators commute strongly.
	\begin{align}
		\begin{CD}
			\mc H_i @>G>> \mc H_j\\
			@V S_{\wtd I} VV @V S_{\wtd I} VV\\
			\mc H_{\ovl i} @>\ovl G>> \mc H_{\ovl j}
		\end{CD}.
	\end{align}
	If we take the polar decomposition of the two vertical $S_{\wtd I}$, then its phase and any power of its absolute value commute strongly with the horizontal $G$ and $\ovl G$, i.e., the following diagrams commute strongly
	\begin{align}
\begin{CD}
	\mc H_i @>G>> \mc H_j\\
	@V \Delta_{\wtd I}^\alpha VV @V \Delta_{\wtd I}^\alpha VV\\
	\mc H_i @>G>> \mc H_j
\end{CD}\qquad
		\begin{CD}
			\mc H_i @>G>> \mc H_j\\
			@V \fk J_{\wtd I} VV @V \fk J_{\wtd I} VV\\
			\mc H_{\ovl i} @>\ovl G>> \mc H_{\ovl j}
		\end{CD}~~.
	\end{align}
This proves \eqref{eq60}.

We already know $G\cdot \Delta_{\wtd I}^\alpha|_{\mc H_i}\subset\Delta_{\wtd I}^\alpha|_{\mc H_j}\cdot G$. Similarly, we have $G^*\cdot \Delta_{\wtd I}^\alpha|_{\mc H_j}\subset\Delta_{\wtd I}^\alpha|_{\mc H_i}\cdot G^*$. Apply $G(\cdot)G$ to both sides and set $p=G^*G$, we obtain $p\cdot \Delta_{\wtd I}^\alpha|_{\mc H_j} G\subset G\Delta_{\wtd I}^\alpha|_{\mc H_i}$. So the domain of $\Delta_{\wtd I}^\alpha|_{\mc H_j} G$ (which equals the domain of $p\Delta_{\wtd I}^\alpha|_{\mc H_j} G$) is inside the domain of $G\Delta_{\wtd I}^\alpha|_{\mc H_i}$. This proves the second relation of \eqref{eq61}. A similar argument proves the first of \eqref{eq61}. 
\end{proof}

In the case that $G$ is an isometry, the second relation of \eqref{eq62} has a more geometric interpretation:

\begin{co}
Let $\mc H_j\in\Obj(\RepdA)$, and let $\mc H_i$ be an $\mc A$-invariant subspace of $\mc H_j$, which is therefore an object of $\RepfA$. Then
\begin{align*}
\Dom(\Delta_{\wtd I}^\alpha|_{\mc H_i})=\Dom(\Delta_{\wtd I}^\alpha|_{\mc H_j})\cap\mc H_i,
\end{align*}
and for any $\xi$ inside this vector space, we have
\begin{align*}
\Delta_{\wtd I}^\alpha|_{\mc H_i}\cdot \xi=	\Delta_{\wtd I}^\alpha|_{\mc H_j}\cdot \xi.
\end{align*}
\end{co}

The dualizing data of $\mc H_i$ is irrelevant since the modular operator does not rely on it.

\begin{proof}
Let $G\in\Hom_{\mc A}(\mc H_i,\mc H_j)$ be the inclusion, and apply the second relation of \eqref{eq61}.
\end{proof}

The following are some easy consequences of the Tomita-Takesaki theory \cite[Chapter VI]{Tak02}.

\begin{pp}\label{lb25}
The following are true when acting on any $\mc H_i\in\Obj(\RepdA)$.
\begin{gather}
\fk J_{\wtd I}^2=\id.\label{eq22}\\
\Delta_{\wtd I'}=\Delta_{\wtd I}^{-1},\qquad\vartheta\fk J_{\wtd I'}=\fk J_{\wtd I}.\label{eq23}\\
\fk J_{\wtd I}\Delta_{\wtd I}^{\frac 12}=\Delta_{\wtd I}^{-\frac 12}\fk J_{\wtd I}.\label{eq24}
\end{gather}
\end{pp}

Recall that $\wtd I'$ is the clockwise complement of $\wtd I$.
\begin{proof}
We know that $\mc H_i$ must be equivalent to a subobject of $\mc H_a$ where $Q=(\mc H_a,\mu,\iota)$ is a standard $C^*$-Frobenius algebra in $\RepdA$. We let $(\mc H_a,\iota^*\mu)$ be the dualizing data of $\mc H_a$. Then the reflection operator $\epsilon$ equals $\id_a$. Let $\mc B$ and $\mc B'$ be the non-local extensions of $\mc A$ associated to $Q$ (cf. Section \ref{lb34}).  By Propositions \ref{lb9} and \ref{lb28}, $S_{\wtd I}|_{\mc H_a}$ and $F_{\wtd I'}|_{\mc H_a}=\vartheta_aS_{\wtd I'}|_{\mc H_a}$ are respectively the $S$ operators for $\mc B(\wtd I)$ and $\mc B'(\wtd I')$ associated to the cyclic vector $\iota\Omega$.  So, by the Tomita-Takesaki Theory for von-Neumann algebras, we have $\fk J_{\wtd I}|_{\mc H_a}=\fk J_{\wtd I}^{-1}|_{\mc H_a}$ and polar decomposition $S_{\wtd I}|_{\mc H_a}=\fk J_{\wtd I}|_{\mc H_a}\Delta_{\wtd I}^{\frac 12}|_{\mc H_a}=\Delta_{\wtd I}^{-\frac 12}|_{\mc H_a}\fk J_{\wtd I}|_{\mc H_a}$ (which proves \eqref{eq22} and \eqref{eq24} on $\mc H_a$), and we also have $\vartheta_aS_{\wtd I'}|_{\mc H_a}=\vartheta_a\fk J_{\wtd I'}|_{\mc H_a}\Delta_{\wtd I'}^{\frac 12}|_{\mc H_a}$. Since $\mc B(I)$ and $\mc B'(I')$ are commutants of each other (Theorem \ref{lb10}), the $S$ operator for $\mc B'(I')$ (which is $\vartheta_aS_{\wtd I'}|_{\mc H_a}$) is also the $F$ operator for $\mc B(\wtd I)$, which equals $\fk J_{\wtd I}|_{\mc H_a}\Delta_{\wtd I}^{-\frac 12}|_{\mc H_a}$. So, the uniqueness of polar decomposition implies \eqref{eq23} on $\mc H_a$.

We have proved \eqref{eq22}-\eqref{eq24} on $\mc H_a$. Let $\varphi\in\Hom_{\mc A}(\mc H_i,\mc H_a)$ be an isometry. It is now easy to prove these three on $\mc H_i$ using  Proposition \ref{lb14}. For instance,
\begin{align*}
\ovl\varphi\cdot \fk J_{\wtd I}|_{\mc H_i}\Delta_{\wtd I}^{\frac 12}|_{\mc H_i}=	\fk J_{\wtd I}|_{\mc H_a}\Delta_{\wtd I}^{\frac 12}|_{\mc H_a}\cdot \varphi=\Delta_{\wtd I}^{-\frac 12}|_{\mc H_a}\fk J_{\wtd I}|_{\mc H_a}\cdot \varphi=\ovl\varphi\cdot \Delta_{\wtd I}^{-\frac 12}|_{\mc H_{\ovl i}}\fk J_{\wtd I}|_{\mc H_i}.
\end{align*}
This proves \eqref{eq24} on $\mc H_i$. A similar argument proves \eqref{eq22} and \eqref{eq23} on $\mc H_i$.
\end{proof}

\begin{co}\label{lb36}
For each $\wtd I\in\Jtd$ and $\mc H_i\in\Obj(\RepdA)$, we have $F_{\wtd I'}|_{\mc H_i}=(S_{\wtd I}|_{\mc H_{\ovl i}})^*$.
\end{co}

\begin{proof}
We compute
\begin{align*}
S_{\wtd I}^*=(\fk J_{\wtd I}\Delta_{\wtd I}^{\frac 12})^*\xlongequal{\eqref{eq22}}\Delta_{\wtd I}^{\frac 12}\fk J_{\wtd I}\xlongequal{\eqref{eq24}} \fk J_{\wtd I}\Delta_{\wtd I}^{-\frac 12}\xlongequal{\eqref{eq23}}\vartheta\fk J_{\wtd I'}\Delta_{\wtd I'}^{\frac 12}=F_{\wtd I'}.
\end{align*}
\end{proof}

\begin{rem}\label{lb24}
We close this section with a brief discussion of the relations between $\Delta_{\wtd I}$ and Connes spatial derivatives \cite{Con80}. Fix $I\in\mc J$. Choose a non-empty $\mc H_i\in\Obj(\RepdA)$ and let $\mc N=\pi_{i,I'}(\mc A(I'))$. The (normalized) state on $\mc A(I')$ defined by $\bk{\cdot\Omega|\Omega}$ is transported through the isomorphism $\pi_{i,I'}$ to a state $\varphi$ on $\mc N$. Let $\mc M=\mc N'$ be the commutant of $\mc N$ (acting on $\mc H_i$). Then $\mc M$ can be described by the left representation of $Q=(\mc H_i\boxtimes\mc H_{\ovl i},\mu,\iota)$ on $\mc H_i$ as follows: Let $\mc B$ be the non-local extension associated to $Q$. Then any $X\in\mc B(\wtd I)$ can be expressed as $A(\chi,\wtd I)=\mu L(\chi,\wtd I)$ for some $\chi\in(\mc H_i\boxtimes\mc H_{\ovl i})(I)$. We then define a representation of $\mc B(\wtd I)$ on $\mc H_i$ by defining the action of $A(\chi,\wtd I)$ on any $\xi\in\mc H_i$ to be $(\id_i\otimes\ev_{\ovl i,i})L(\chi,\wtd I)\xi$. This representation is indeed faithful (since $\mc B(\wtd I)$ is indeed a (type III) factor). Moreover, the image of $\mc B(\wtd I)$ under this representation is exactly $\mc M$. The state of $\mc B(\wtd I)$ defined by $\bk{\cdot\iota\Omega|\iota\Omega}$ is transported through this representation to a state $\psi$ of $\mc M$. Then we actually have
\begin{align}
\Delta_{\wtd I}|_{\mc H_i}=\frac{d\psi}{d\varphi}.
\end{align}

Note that $\psi$ is in general not normalized: By the fact that $\iota=\coev_{i,\ovl i}$ we have $\psi(1)=d_i$ where $d_i$ is the quantum dimension of $\mc H_i$, i.e. $d_i\id_0=\ev_{i,\ovl i}\coev_{i,\ovl i}$ where the $\ev$ and $\coev$ are standard. We give another description of $\psi$: Let $\mc E:\mc M\rightarrow\pi_{i,I}(\mc A(I))$ be the minimal conditional expectation of the subfactor $\pi_{i,I}(\mc A(I))\subset\mc M$. Transport the state $\bk{\cdot\Omega|\Omega}$ of $\mc A(I)$ to $\pi_{i,I}(\mc A(I))$ and denote it by $\psi_0$. Then the normalized state $d_i^{-1}\psi$ equals $\psi_0\circ\mc E$.
\end{rem}

\section{Categorical and non-local Bisognano-Wichmann theorems}\label{lb22}

Let $\fk r:z\in \mbb S^1\mapsto \ovl z=z^{-1}\in\mbb S^1$ be the reflection, represented by $\left( \begin{array}{cc}
0& 1 \\
1 &0
\end{array} \right)$. (I.e., $\fk r(z)=\frac {0z+1}{1z+0}$.) Then $\fk r=\fk r^{-1}$, and $g\in\PSU\mapsto \fk rg\fk r\in\PSU$ is an automorphism of $\PSU$, and acts on $\mbb S^1$ as $g(z)=\frac{\alpha z+\beta}{\ovl\beta z+\ovl\alpha}$. Recall that any element in $\PSU$ takes the form $g=\left( \begin{array}{cc}
\alpha  & \beta \\
\ovl\beta &\ovl\alpha
\end{array} \right)$ where $|\alpha|^2-|\beta|^2=1$. Then we have $\fk rg\fk r=\left( \begin{array}{cc}
\ovl\alpha  &\ovl \beta \\
\beta &\alpha
\end{array} \right)$. We lift this automorphism to $\UPSU$ and also denote it by $\fk r(\cdot)\fk r$. For $\wtd I=(I,\arg I)\in\Jtd$, we define $\fk r\wtd I=(\fk rI,\arg_{\fk rI})$ where $\arg_{\fk rI}(z)=-\arg_I(\ovl z)$ for any $z\in \fk rI$.

Define $\mbb S^1_+=\{a+\im b\in\mbb S^1:b>0\}$ to  be the upper semi-circle. Define $\wtd{\mbb S^1_+}$ such that $\arg_{\mbb S^1_+}$ takes values in $(0,\pi)$. Then $\wtd{\mbb S^1_-}:=\fk r\wtd{\mbb S^1_+}$ is the lower semi-circle with arg values in $(-\pi,0)$. Note that $\wtd{\mbb S^1_-}$ is the clockwise complement of $\wtd{\mbb S^1_+}$. We write $\Delta_{\wtd{\mbb S^1_+}},\Delta_{\wtd{\mbb S^1_-}},$ as $\Delta_+,\Delta_-$ respectively. We also define $\Theta=\fk J_{\wtd{\mbb S^1_+}}$, called the  \textbf{PCT operator} of $\Ed$, which is an anti-unitary map from each $\mc H_i$ to its dual object $\mc H_{\ovl i}$.  

The action of $\Theta$ on $\mc H_i\in\Obj(\RepdA)$ is written as $\Theta|_{\mc H_i}$ as usual, or $\Theta_i$ for short. Note that $\Theta$ is an involution by Proposition \ref{lb25}, i.e.,
\begin{align*}
\Theta_i=\Theta_{\ovl i}^{-1}	
\end{align*}
for each $\mc H_i\in\Obj(\RepdA)$. Also, if $g\in\UPSU$ and $\wtd I=g\wtd{\mbb S^1_+}$, then by proposition \ref{lb12},
\begin{align}
\fk J_{\wtd I}=g\Theta g^{-1}.
\end{align}
The following noteworthy result follows from proposition \ref{lb14}, which says that conjugates of morphisms are implemented by the PCT operator.
\begin{thm}
For any morphism $G$ of objects in $\RepdA$, we have
\begin{align*}
\ovl G=\Theta\cdot G\cdot\Theta.
\end{align*}
\end{thm}

Consider the rotation subgroup $
\varrho(t)=\left( \begin{array}{cc}
e^{\frac {\im t}2} & 0 \\
0 &e^{\frac {-\im t}2}
\end{array} \right)$ and dilation subgroup 
$\delta(t)=\left( \begin{array}{cc}
\cosh\frac t2 & -\sinh\frac t2 \\
-\sinh\frac t2 &\cosh\frac t2
\end{array} \right)$ of $\PSU$. For each $I\in\mc J$, define $\delta_I(t)=g\delta(t)g^{-1}$ where $g\in\PSU$ and $g\mbb S^1_+=I$. Then $\delta_I$ is well defined, and $\delta(t)=\delta_{\mbb S^1_+}(t)$. We lift $\varrho$ and $\delta$ to one-parameter subgroups of $\UPSU$ and denote them by the same symbols.

Let $Q=(\mc H_a,\mu,\iota)$ be a $C^*$-Frobenius algebra in $\RepdA$, and let $\mc B,\mc B'$ be the pair of non-local extensions of $\mc A$ associated to $Q$. The following result is well-known (cf. \cite{ALR01} theorem 2.1 or \cite{LR04} proposition 3.2). We present the details of the proof for the reader's convenience.

\begin{lm}\label{lb35}
Let $D_{\wtd I}$ and $\Theta^Q$ be the modular operator and conjugation for $\mc B(\wtd{\mbb S_+^1})$ and the cyclic separating vector $\iota\Omega$. Then $z(t):=\delta_I(2\pi t)D_{\wtd I}^{\im t}$ is a one-parameter group of unitary operators on $\mc H_a$ independent of $\wtd I$, and $\Theta^Qg\Theta^Q$ equals the action of $\fk r g\fk r$ when acting on $\mc H_a$. 
\end{lm}

\begin{proof}
Set $A=\left( \begin{array}{cc}
-1  & \im \\
1 & \im
\end{array} \right)$ corresponding to the Cayley transform sending the (resp. positive) real line to the  unit (resp. upper semi-) circle. Its inverse is   $A^{-1}=\left( \begin{array}{cc}
	-\frac 12  & \frac 12 \\
	-\frac \im 2 & -\frac \im 2
\end{array} \right)$. Let $\tau_\pm(t)\in\PSU$ such that $A^{-1}\tau_+(t)A=\left( \begin{array}{cc}
1  & t \\
0 & 1
\end{array} \right)$ and $A^{-1}\tau_-(t)A=\left( \begin{array}{cc}
1  & 0 \\
-t & 1
\end{array}\right)$. Then $\tau_+$ and $\tau_-$ are the one-parameter translation subgroups of $\PSU$ associated to $\mbb S_+^1$ and $\mbb S_-^1$ respectively. $\tau_\pm$ fixes $\mp1$, and rotates the other points anticlockwisely when $t$ increases. It is well-known that these subgroups $A^{-1}\tau_\pm(t)A$ generate $\mathrm{PSL}_2(\mbb R)=A^{-1}\PSU A$, and hence $\tau_\pm$ generate $\PSU$. Indeed, one computes $A^{-1}\varrho(t)A=\left( \begin{array}{cc}
\cos (t/2)  & \sin(t/2) \\
-\sin(t/2) & \cos (t/2)
\end{array}\right)$ and the dilation group $\delta(t)$ satisfies $A^{-1}\delta(t)A=\left( \begin{array}{cc}
e^{t/2}  & 0 \\
0 & e^{-t/2}
\end{array}\right)$. Using these relations, one computes
\begin{gather*}
\varrho(4t)=\tau_+(\tan t)\tau_-(\sin 2t)\tau_+(\tan t),\\
\delta(2t)=\tau_+(-e^t)\tau_-(1-e^{-t})\tau_+(1)\tau_-(1-e^t).	
\end{gather*}
This proves the claim since $\tau_+,\delta,\varrho$ generate $\PSU$ by KAN decomposition \cite{Lon08}.

We also let $\tau_\pm$ denote their lifts to $\UPSU$. We claim that $\tau_\pm$ generate $\UPSU$. Since the preimage of $1\in\PSU$ under the covering map is $\varrho(2\pi\mbb Z)$, it suffices to show $\varrho(2\pi)$ is inside the subgroup  generated by $\tau_\pm$. This is true because the above relation for $\varrho(4t)$ holds not just in $\PSU$, but also in $\UPSU$ due to the uniqueness of lifting whose initial value is $1\in\UPSU$.

We compute that for each $s,t\in\mbb R$,
\begin{gather*}
\delta(-2\pi s)\tau_\pm(t)\delta(2\pi s)=\tau_\pm(e^{\mp 2\pi s}t).	
\end{gather*}
These relations hold in $\PSU$. Consider these as one-parameter subgroups of the variable $t$ (and fixed $s$) and note that their lifts to $\UPSU$ are unique, we see that these relations also hold in $\UPSU$. Similarly, one checks that
\begin{align*}
\fk r\tau_\pm(t)\fk r=\tau_\pm(-t).
\end{align*}
By \cite[Cor.4.4]{BCL98}, the generators of $\tau_\pm$ are positive. (In the case that $\mc A$ is conformal covariant, this also follows from \cite{Wei06} Lemma 3.1 and Theorem 3.8.) For any $t\geq 0$, since $\tau_+(t)(\wtd{\mbb S_+^1})\subset \wtd{\mbb S_+^1}$, by Theorem \ref{lb10}, we have $\Ad\tau_+(t)(\mc B( \wtd{\mbb S_+^1}))\subset \mc B( \wtd{\mbb S_+^1})$.  Then, by Borchers' theorem \cite[Thm. II.9]{Bor92}, we obtain for any $t\in\mbb R$ (when acting on $\mc H_a$) that
\begin{gather}
D_+^{\im t}gD_+^{-\im t}=\delta(-2\pi s)g\delta(2\pi s),\qquad \fk J^Q g \fk J^Q=\fk r g\fk r\label{eq63}
\end{gather}
when $g$ is in $\tau_+$. Similarly, since $\tau_-(t)(\wtd{\mbb S_+^1})\subset \wtd{\mbb S_+^1}$ when $t\leq 0$, we conclude that \eqref{eq63} holds when $g\in\tau_-$. Thus, it holds for all $g\in\UPSU$.

The first half of \eqref{eq63} shows that $z(t):=\delta_{\wtd {\mbb S_+^1}}(2\pi t)D_{\wtd {\mbb S_+^1}}^{\im t}|_{\mc H_a}$ commutes with any $g\in\UPSU$. For each $\wtd I\in\Jtd$, choose $g$ sending $\wtd {\mbb S_+^1}$ to $\wtd I$. Note that $gD_{\wtd {\mbb S_+^1}}g^{-1}=D_{\wtd I}$ since $g$ preserves $\iota\Omega$ and $g\mc B(\wtd{\mbb S_+^1})g^{-1}=\mc B(\wtd I)$ by Theorem \ref{lb10}. Thus, we conclude that $z(t)=\delta_I(2\pi t)D_{\wtd I}^{\im t}|_{\mc H_a}$. In particular, $z(t)$ is independent of $\wtd I$. Apply the first half of \eqref{eq63} to the case $g=\delta(2\pi s)$, we see that $D_+^{\im t}$ commutes with $\delta(2\pi s)$. This proves that $z(t)$ is a one-parameter group.
\end{proof}

\begin{thm}\label{lb13}
Let $Q=(\mc H_a,\mu,\iota)$ be a standard $C^*$-Frobenius algebra in $\RepdA$ with reflection operator $\epsilon\in\Hom_{\mc A}(\mc H_a,\mc H_{\ovl a})$   (see \eqref{eq18}). $\mc H_{\ovl a}$ is the dual object of $\mc H_a$. Let $\mc B$ (and also $\mc B'$) be the non-local extension of $\mc A$ associated to $Q$. For any $\wtd I\in\Jtd$, $\Delta_{\wtd I}|_{\mc H_a}$ (resp. $\epsilon^{-1}\Theta$) is the modular operator (resp. conjugation) for $\mc B(\wtd I)$ (resp. $\mc B(\wtd{\mbb S_+^1})$) and the cyclic separating vector $\iota\Omega$. Moreover, the following are true when acting on $\mc H_a$.

(a) (Geometric modular theorem) For any $t\in\mathbb R$,
\begin{align}
\Delta_{\wtd I}^{\im t}=\delta_I(-2\pi t).\label{eq26}
\end{align}

(b) (PCT theorem) For any $g\in\UPSU$,
\begin{gather}
\epsilon^{-1}\Theta \cdot g \cdot\Theta\epsilon=\fk r g\fk r\label{eq27}\\
\epsilon^{-1}\Theta\cdot\mc B(\wtd I)\cdot\Theta\epsilon=\mc B'(\fk r\wtd I). \label{eq28}
\end{gather}
\end{thm}
Note that in equation \eqref{eq27}, $\fk r g\fk r$ is in $\UPSU$ and is acting on $\mc H_a$.
\begin{proof}
By Proposition \ref{lb9}, the modular operator and conjugation for $\mc B(\wtd I)$ are $\Delta_{\wtd I}|_{\mc H_a}$ and $\epsilon^{-1}\fk J_{\wtd I}|_{\mc H_a}$ respectively.	Thus,  Lemma \ref{lb35} implies \eqref{eq27} and the fact that	$z(t):=\delta_I(2\pi t)\Delta_{\wtd I}^{\im t}$ (acting on $\mc H_a$) is a one-parameter group independent of $\wtd I$. Thus $\delta(2\pi t)\Delta_{+}^{\im t}=\delta_{-}(2\pi t)\Delta_{-}^{\im t}$ where $\delta_-=\delta_{\mbb S^1_-}$. By \eqref{eq23}, we have $\Delta_{-}^{\im t}=\Delta_{+}^{-\im t}$. We also have $\delta_-(2\pi t)=\varrho(\pi)\delta(2\pi t)\varrho(-\pi)$, which equals $\delta(-2\pi t)$ by an easy calculation. Thus $z(t)=z(-t)$, which forces $z(t)$ to be $\id$. This proves the geometric modular theorem. By the non-local Haag-duality (Theorem \ref{lb10}), $\mc B'(\wtd{\mbb S^1_-})$ is the commutant of $\mc B(\wtd{\mbb S^1_+})$. Thus, by Tomita-Takesaki theorem, \eqref{eq28} holds in the special case that $\wtd I=\wtd{\mbb S^1_+}$. The general case follows from the special case, the M\"obius covariance of $\mc B$, and equation \eqref{eq27}.
\end{proof}

\begin{thm}\label{lb26}
For any $\wtd I\in\Jtd$, the following are true when acting on any $\mc H_j\in\Obj(\RepdA)$.

(a) (Geometric modular theorem) For any $t\in\mathbb R$,
\begin{align}
\Delta_{\wtd I}^{\im t}=\delta_I(-2\pi t).\label{eq29}
\end{align}

(b) (PCT theorem) For any $g\in\UPSU$,
\begin{gather}
\Theta \cdot g \cdot\Theta=\fk r g\fk r.\label{eq30}
\end{gather}
Moreover, for any $\mc H_i\in\Obj(\RepdA)$ and $\xi\in\mc H_i(I)$, we have
\begin{gather}
\Theta\cdot\mc H_i(I)=\mc H_{\ovl i}(\fk r I),\label{eq34}\\
\Theta \cdot L(\xi,\wtd I) \cdot \Theta=R(\Theta\xi,\fk r\wtd I).\label{eq35}
\end{gather}
\end{thm}

\begin{proof}
Let $Q=(\mc H_a,\mu,\iota)$ where $\mc H_a=\mc H_k\boxtimes\mc H_{\ovl k}$ and $\mc H_k=\mc H_0\oplus\mc H_j$. (See the discussions starting from \eqref{eq49}, especially the choice of dualizing data.) Since $\epsilon=1$, \eqref{eq29} and \eqref{eq30} hold on $\mc H_a$ by Theorem \ref{lb13}. Note that $\Theta_a$ is an involution on $\mc H_a$ (since $\mc H_a$ is the canonical dual object of $\mc H_a$). $\psi_{j,0}:\mc H_j\simeq\mc H_{j,0}\rightarrow\mc H_a$ (defined by \eqref{eq64}) is an isometry, whose conjugate is $\psi_{0,\ovl j}$ by \eqref{eq65}. Since morphisms intertwine the actions of $\UPSU$ (\cite[Prop. 2.2]{GL96} and Lemma \ref{lb21}), by Proposition \ref{lb14},
\begin{align*}
\psi_{j,0}\Theta_{\ovl j}g\Theta_j=\Theta_a\psi_{0,\ovl j}g\Theta_j=\Theta_ag\Theta_a\psi_{j,0}=(\fk r g\fk r)\psi_{j,0}=\psi_{j,0}(\fk r g\fk r),
\end{align*}
which proves \eqref{eq30} on $\mc H_j$. \eqref{eq29} on $\mc H_j$ can be proved in a similar way.

Now we take $\mc H_k=\mc H_0\oplus\mc H_i\oplus\mc H_j$ and define $Q=(\mc H_a,\mu,\iota)$ in the same way. Again, we use freely the notations starting from \eqref{eq49}. Identify $\mc H_i$ with $\mc H_i\boxtimes\mc H_0$ using the right unitor. Choose any $\xi\in\mc H_i(I)$. Recall the definition of $A,B$ in Section \ref{lb34}. Then
\begin{align*}
\Theta_a A(\psi_{i,0}\xi,\wtd I)\Theta_a\iota\Omega=\Theta_a A(\psi_{i,0}\xi,\wtd I)\iota\Omega \xlongequal{\eqref{eq38}} \Theta_a\psi_{i,0}\xi=\psi_{0,\ovl i}\Theta_i\xi.
\end{align*}
By the non-local PCT theorem,  $\Theta_a A(\psi_{i,0}\xi,\wtd I)\Theta$ is inside $\mc B'(\fk r\wtd I)$. So there exists $\eta\in\mc H_a(\fk r\wtd I)$ such that $\Theta_a A(\psi_{i,0}\xi,\wtd I)\Theta_a=B(\eta,\fk r\wtd I)$. Again, by \eqref{eq38}, we must have $\eta=\psi_{0,\ovl i}\Theta_i\xi$. Therefore $\Theta_i\xi=\psi_{0,\ovl i}^*\eta\in\psi_{0,\ovl i}^*\mc H_a(\fk rI)\subset \mc H_{\ovl i}(\fk rI)$. This proves $\Theta\cdot\mc H_i(I)\subset\mc H_{\ovl i}(\fk r I)$. Similarly, $\Theta\cdot\mc H_{\ovl i}(\fk rI)\subset\mc H_i(I)$. Thus we have proved \eqref{eq34}.

We have proved that $\Theta_a A(\psi_{i,0}\xi,\wtd I)\Theta_a=B(\psi_{0,\ovl i}\Theta_i\xi,\fk r\wtd I)$. Let both sides act on $\psi_{j,0}$.  The left hand side is
\begin{align*}
&\Theta_a A(\psi_{i,0}\xi,\wtd I)\Theta_a\psi_{j,0}=\Theta_a A(\psi_{i,0}\xi,\wtd I)\psi_{0,\ovl j}\Theta_j=\Theta_a \mu L(\psi_{i,0}\xi,\wtd I)\psi_{0,\ovl j}\Theta_j\\
\xlongequal{\eqref{eq8}}&\Theta_a \mu(\psi_{i,0}\otimes\psi_{0,\ovl j})L(\xi,\wtd I)\Theta_j\xlongequal{\eqref{eq66}} \Theta_a \psi_{i,\ovl j}L(\xi,\wtd I)\Theta_j\xlongequal{\eqref{eq65}} \psi_{j,\ovl i}\Theta_{i\boxtimes\ovl j} L(\xi,\wtd I)\Theta_j,
\end{align*}
which equals the right hand side (acting on $\mc H_j$)
\begin{align*}
&B(\psi_{0,\ovl i}\Theta_i\xi,\fk r\wtd I)\psi_{j,0}=\mu R(\psi_{0,\ovl i}\Theta_i\xi,\fk r\wtd I)\psi_{j,0} \xlongequal{\eqref{eq8}} \mu(\psi_{j,0}\otimes\psi_{0,\ovl i})R(\Theta_i\xi,\fk r\wtd I)\\
\xlongequal{\eqref{eq66}}& \psi_{j,\ovl i}R(\Theta_i\xi,\fk r\wtd I).
\end{align*}
Equation \eqref{eq35} is now proved.
\end{proof}

\begin{cv}
By \eqref{eq29},  $\Delta_{\wtd I}$ depends only on $I$ but not on $\arg_I$. Thus we will write $\Delta_{\wtd I}$ as $\Delta_I$ in the future.
\end{cv}

\begin{co}[\cite{GL96} Thm.2.11]
For any $\mc H_i\in\Obj(\RepdA)$, $I\in\mc J$, and $x\in\mc A(I)$, we have
\begin{align}
\Theta\pi_{i,I}(x)\Theta=\pi_{\ovl i,\fk rI}(\Theta x\Theta).
\end{align}
\end{co}
Notice that $\Theta x \Theta\in\mc A(\fk rI)$ by the PCT theorem for $\mc A$.
\begin{proof}
Choose an $\arg_I$. Then, by \eqref{eq9}, we have $\pi_{i,I}(x)=L(x\Omega,\wtd I)|_{\mc H_i}$ and $\pi_{\ovl i,\fk rI}(\Theta x\Theta)=R(\Theta x\Theta\Omega,\fk r\wtd I)|_{\mc H_{\ovl i}}=R(\Theta x\Omega,\fk r\wtd I)|_{\mc H_{\ovl i}}$. We may now apply \eqref{eq35} to prove the desired equation.
\end{proof}

\begin{rem}
The above corollary gives an explicit construction of dual representation of any $\mc H_i\in\Obj(\RepdA)$. Namely, we choose any seperable Hilbert space $\mc H_{\ovl i}$, and choose an arbitrary anti-unitary map $\Co:\mc H_i\rightarrow\mc H_{\ovl i}$. Define a representation $\pi_{\ovl i}$ of $\mc A$ on $\mc H_{\ovl i}$ such that for any $I\in\mc J$ and $x\in\mc A(I)$,
\begin{align}
\pi_{\ovl i,I}(x)=\Co\cdot\pi_{i,\fk rI}(\Theta x\Theta)\cdot\Co^{-1}.
\end{align}
Then $(\mc H_{\ovl i},\pi_{\ovl i})$ is equivalent to the a dual object of $(\mc H_i,\pi_i)$.
\end{rem}

The conformal spin-statistics theorem is also an easy consequence of the categorical PCT theorem:

\begin{thm}[\cite{GL96} Thm. 3.13, \cite{Joer96} Sect. 4.1]\label{lb20}
On any $\mc H_i\in\Obj(\RepdA)$ we have
\begin{align}
\vartheta=\varrho(2\pi).
\end{align}
\end{thm}

\begin{proof}
By \eqref{eq23} we have $\Theta=\fk J_{\wtd{\mbb S^1_+}}=\vartheta\fk J_{\wtd{\mbb S^1_-}}$. So $\fk J_{\wtd{\mbb S^1_-}}=\vartheta^{-1}\Theta$. Since $\wtd{\mbb S^1_-}=\varrho(-\pi)\wtd{\mbb S^1_+}$, by proposition \ref{lb12} we also have $S_{\wtd{\mbb S^1_-}}=\varrho(-\pi)S_{\wtd{\mbb S^1_+}}\varrho(\pi)$, and hence $\fk J_{\wtd{\mbb S^1_-}}=\varrho(-\pi)\Theta\varrho(\pi)$. So $\vartheta^{-1}\Theta=\varrho(-\pi)\Theta\varrho(\pi)$. By the categorical PCT theorem, $\Theta\varrho(t)=(\fk r\varrho(t)\fk r)\Theta=\varrho(-t)\Theta$. Therefore $\vartheta^{-1}\Theta=\varrho(-2\pi)\Theta$ and hence $\vartheta=\varrho(2\pi)$.
\end{proof}

We now want to generalize theorem \ref{lb13} to any (non-necessarily standard) $C^*$-Frobenius algebra $Q=(\mc H_a,\mu,\iota)$ in $\RepdA$. Again, we choose a dualizing data $(\mc H_{\ovl a},\ev_{a,\ovl a})$ for $\mc H_a$. (Recall that $\ev_{a,\ovl a},\ev_{\ovl a,a}$ are standard.) There is a unique invertible $\epsilon\in\Hom_{\mc A}(\mc H_a,\mc H_{\ovl a})$, also called \textbf{(left) reflection operator}, such that
\begin{align}
\ev_{\ovl a,a}(\epsilon\otimes\id_a)=\iota^*\mu.\label{eq67}
\end{align}
By the uniqueness of standard evaluations up to unitaries, it is clear that $Q$ is standard if and only if $\epsilon$ is unitary.

\begin{lm}
We have
\begin{align}
	\epsilon^\vee=(\epsilon^{-1})^*\quad\text{and hence}\quad\epsilon^{-1}=\ovl\epsilon.\label{eq44}
\end{align}
\end{lm}

\begin{proof}
The adjoint of \eqref{eq67}  is $(\epsilon^*\otimes\id_a)\coev_{\ovl a,a}=\mu^*\iota$. \eqref{eq67} also implies $\ev_{a,\ovl a}(\id_a\otimes\epsilon^\vee)=\iota^*\mu$. By the Frobenius relation and the unit property of $Q$, we have
\begin{align*}
(\iota^*\mu\otimes \id_a)(\id_a\otimes\mu^*\iota)=(\iota^*\otimes\id_a)(\id_a\otimes\mu)(\mu^*\otimes\id_a)(\id_a\otimes\iota)=(\iota^*\otimes\id_a)\mu^*\mu(\id_a\otimes\iota)=\id_a.
\end{align*}
Thus $(\ev_{a,\ovl a}(\id_a\otimes\epsilon^\vee)\otimes \id_a)(\id_a\otimes(\epsilon^*\otimes\id_a)\coev_{\ovl a,a})$ must be $\id_a$. But
\begin{align*}
&(\ev_{a,\ovl a}(\id_a\otimes\epsilon^\vee)\otimes \id_a)(\id_a\otimes(\epsilon^*\otimes\id_a)\coev_{\ovl a,a})=(\ev_{a,\ovl a}\otimes \id_a)(\id_a\otimes\epsilon^\vee\epsilon^*\otimes\id_a)(\id_a\otimes\coev_{\ovl a,a})\\
=&(\ev_{a,\ovl a}\otimes \id_a)((\epsilon^\vee\epsilon^*)^\vee\otimes\id_{\ovl a}\otimes\id_a)(\id_a\otimes\coev_{\ovl a,a})=(\epsilon^\vee\epsilon^*)^\vee.
\end{align*}
So we must have $\epsilon^\vee=(\epsilon^{-1})^*$, and hence $\epsilon^{-1}=(\epsilon^\vee)^*=\ovl\epsilon$.
\end{proof}

From this, one easily sees that $\coev_{\ovl a,a}=(\epsilon^\vee\otimes\id_a)\mu^*\iota$ and $\coev_{a,\ovl a}=(\id_a\otimes\epsilon)\mu^*\iota$. Using these two equations, the following can be proved in essentially the same way as proposition \ref{lb9}.

\begin{pp}
$\epsilon^{-1} S_{\wtd I}|_{\mc H_a}$ and $\epsilon^* F_{\wtd I}|_{\mc H_a}$ are respectively the $S$ operators of $\mc B(\wtd I)$ and $\mc B'(\wtd I)$ with respect to $\iota\Omega$.
\end{pp}

\begin{proof}
In Proposition \ref{lb9}, we assumed $Q$ is standard, and used $\coev_{\ovl a,a}=(\epsilon\otimes\id_a)\mu^*\iota$ (see the second paragraph of the proof) to show $\epsilon^{-1}F_{\wtd I}|_{\mc H_a}$ is the $S$ operator for $\mc B'(I)$. In the current general case, since we have $\coev_{\ovl a,a}=(\epsilon^\vee\otimes\id_a)\mu^*\iota$, we apply the proof of Proposition \ref{lb9} verbatim, except that we replace $\epsilon$ by $\epsilon^\vee$. So $\epsilon^{-1}$ is replaced by $(\epsilon^\vee)^{-1}$, which equals $\epsilon^*$ by \eqref{eq44}. This proves $\epsilon^* F_{\wtd I}|_{\mc H_a}$ is the $S$ operator for $\mc B'(\wtd I)$.

Similarly, one may use $\coev_{a,\ovl a}=(\id_a\otimes\epsilon)\mu^*\iota$ to show $\epsilon^{-1}S_{\wtd I}|_{\mc H_a}$ is the $S$ operator for $\mc B(I)$. Alternatively, we know that $\epsilon^*F_{\wtd I'}|_{\mc H_a}$ is the $S$ operator of $\mc B'(\wtd I')=\mc B(\wtd I)'$ (recall non-local Haag duality in Theorem \ref{lb10}), thus it is the $F$ operator of $\mc B(\wtd I)$. It's adjoint should be the $S$ operator of $\mc B(\wtd I)$. We use Corollary \ref{lb36} and Proposition \ref{lb32} to compute this adjoint:
\begin{align*}
(\epsilon^*F_{\wtd I'}|_{\mc H_a})^*=	S_{\wtd I}|_{\mc H_{\ovl a}}\cdot \epsilon=\ovl\epsilon S_{\wtd I}|_{\mc H_a}=\epsilon^{-1} S_{\wtd I}|_{\mc H_a}.
\end{align*}
\end{proof}

Hence we have:
\begin{pp}\label{lb15}
The modular conjugations and operators of $\mc B(\wtd I)$ and $\mc B'(\wtd I)$ with respect to $\iota\Omega$ are described by the following polar decompositions:
\begin{gather}
\epsilon^{-1} S_{\wtd I}\big|_{\mc H_a}=\big(\fk J_{\wtd I}\epsilon(\epsilon^*\epsilon)^{-\frac 12}\big)\cdot\big(\Delta_I^{\frac 12}(\epsilon^*\epsilon)^{\frac 12}\big)\Big|_{\mc H_a},\label{eq42}\\
\epsilon^* F_{\wtd I}\Big|_{\mc H_a}=\big(\vartheta\fk J_{\wtd I}\epsilon^\vee(\epsilon^*\epsilon)^{\frac 12}\big)\cdot\big(\Delta_I^{\frac 12}(\epsilon^*\epsilon)^{-\frac 12}\big)\Big|_{\mc H_a}.\label{eq43}
\end{gather}
\end{pp}

\begin{proof}
Recall the polar decompositions \eqref{eq41}.  By proposition \ref{lb14}, $\Delta_I$ commutes with any homomorphism, and we have $\epsilon^{-1}\fk J_{\wtd I}\Delta_I^{\frac 12}=\fk J_{\wtd I}\ovl{\epsilon^{-1}}\Delta_I^{\frac 12}=\fk J_{\wtd I}\epsilon\Delta_I^{\frac 12}$. This proves \eqref{eq42}. Similarly we have $\epsilon^*\vartheta\fk J_{\wtd I}\Delta_I^{\frac 12}=\vartheta\fk J_{\wtd I}\ovl\epsilon^*\Delta_I^{\frac 12}=\vartheta\fk J_{\wtd I}\epsilon^\vee\Delta_I^{\frac 12}$. In addition, using \eqref{eq44} we have
\begin{align}
(\epsilon^\vee)^*\epsilon^\vee=\epsilon^{-1}\epsilon^\vee=\epsilon^{-1}(\epsilon^*)^{-1}=(\epsilon^*\epsilon)^{-1}.
\end{align}
This proves the polar decomposition \eqref{eq43}.
\end{proof}

We now prove the following modified Bisognano-Wichmann theorem.

\begin{thm}\label{lb27}
Let $Q=(\mc H_a,\mu,\iota)$ be a (non-necessarily standard) $C^*$-Frobenius algebra in $\RepdA$, let $\mc B$ and $\mc B'$ be the associated non-local extensions of $\mc A$, and let $\epsilon$ be the reflection operator (see \eqref{eq67}). Choose any $\wtd I\in\Jtd$. The following are true when acting on $\mc H_a$.

(a) (Modified geometric modular theorem) Let $D_{\wtd I}$ and $D'_{\wtd I}$ be respectively the modular operators of $\mc B(\wtd I)$ and $\mc B'(\wtd I)$ with respect to $\iota\Omega$. Then for any $t\in\mathbb R$,
\begin{align}
D_{\wtd I}^{\im t}=(\epsilon^*\epsilon)^{\im t}\delta_I(-2\pi t),\qquad (D'_{\wtd I})^{\im t}=(\epsilon^*\epsilon)^{-\im t}\delta_I(-2\pi t).\label{eq40}
\end{align}

(b) (PCT theorem) Let $\Theta^Q=\big(\fk J_{\wtd {\mbb S^1_+}}\epsilon(\epsilon^*\epsilon)^{-\frac 12}\big)\big|_{\mc H_a}$ be the modular conjugation of $\mc B(\wtd {\mbb S^1_+})$ with respect to $\iota\Omega$. Then for any $g\in\UPSU$,
\begin{gather}
\Theta^Q \cdot g \cdot\Theta^Q=\fk r g\fk r\\
\Theta^Q\cdot\mc B(\wtd I)\cdot\Theta^Q=\mc B'(\fk r\wtd I). 
\end{gather}
\end{thm}

\begin{proof}
The proof of PCT theorem is exactly the same as in the standard case (Theorem \ref{lb13}), which uses Lemma \ref{lb35}. By proposition \ref{lb15}, we have $D_{\wtd I}^{\frac 12}=\big(\Delta_I^{\frac 12}(\epsilon^*\epsilon)^{\frac 12}\big)\Big|_{\mc H_a}$ and $(D'_{\wtd I})^{\frac 12}=\big(\Delta_I^{\frac 12}(\epsilon^*\epsilon)^{-\frac 12}\big)\Big|_{\mc H_a}$. Thus, equations \eqref{eq40} follow directly from \eqref{eq26}.
\end{proof}	

\begin{co}
For the non-local extension $\mc B$ obtained by a $C^*$-Frobenius algebra $Q$ in $\RepdA$, the standard geometric modular theorem \eqref{eq26} holds if and only if $Q$ is standard.
\end{co}

\begin{proof}
$Q$ is standard if and only if the invertible homomorphism $\epsilon$ is unitary, if and only if $\epsilon^*\epsilon=\id_a$.
\end{proof}

\section{Unbounded field operators in rigid categorical extensions}

In this section, we discuss the relation between the domain of $\Delta_I^{\frac 12}$ and the preclosedness of certain unbounded operators in $\Ef$. Since the modular conjugations do not rely on dualizing data, we work with dualizable objects instead of dualized ones. Recall that $\Ef$ is the categorical extension associated to the braided $C^*$-tensor category $\RepfA$ of dualizable M\"obius covariant representations of $\mc A$. (Cf. Section \ref{lb30}.)

First, we recall the following well-known fact (cf. \cite{Tak02} section VI.1). A proof is included for the reader's convenience.

\begin{pp}\label{lb17}
Let $\mc M$ be a von Neumann algebra on a Hilbert space $\mc H$ with commutant $\mc M'$, and let $\Omega$ be a cyclic separating vector of $\mc M$. Let $\Delta,\fk J$ be the modular operator and conjugation of $(\mc M,\Omega)$, and set $S=\fk J\Delta^{\frac 12}$. For any $\xi\in\mc H$, define an unbounded operator $\scr L(\xi)$ with domain $\mc M'\Omega$ such that $\scr L(\xi)y\Omega=y\xi$ for any $y\in\mc M'$. Then the following two conditions are equivalent.

(a) $\Omega\in\Dom(\scr L(\xi)^*)$.

(b) $\xi\in\Dom(\Delta^{\frac 12})$. \\
If either  (a) or (b) is true, then $\scr L(\xi)$ is preclosed, and $S\xi=\scr L(\xi)^*\Omega$.
\end{pp}

\begin{proof}
Let $\fk J$ be the modular conjugation. Recall that $S:=\fk J\Delta^{\frac 12}$ has core $\mc M\Omega$, $F=S^*=\fk J\Delta^{-\frac 12}$ has core $\mc M'\Omega$, and $Sx\Omega=x^*\Omega,Fy\Omega=y^*\Omega$ for any $x\in\mc M,y\in\mc M'$.

First, we assume that (a) is true. Then, for any $y\in\mc M'$, we compute
\begin{align*}
\bk{S^*y\Omega|\xi}=\bk{y^*\Omega|\xi}=\bk{\Omega|y\xi}=\bk{\Omega|\scr L(\xi)y\Omega}=\bk{\scr L(\xi)^*\Omega|y\Omega},
\end{align*}
which shows that $\xi\in\Dom(S)=\Dom(\Delta^{\frac 12})$ and $S\xi=\scr L(\xi)^*\Omega$.

Next, assume that (b) is true. Choose any $y_1,y_2\in\mc M'$. Then
\begin{align*}
&\bk{\scr L(S\xi)y_1\Omega|y_2\Omega}=\bk{y_1S\xi|y_2\Omega}=\bk{S\xi|y_1^*y_2\Omega}=\bk{Fy_1^*y_2\Omega|\xi}=\bk{y_2^*y_1\Omega|\xi}\\
=&\bk{y_1\Omega|y_2\xi}=\bk{y_1\Omega|\scr L(\xi)y_2\Omega},
\end{align*}
which shows $\scr L(S\xi)\subset\scr L(\xi)^*$. Thus  $\Omega\in\Dom(\scr L(S\xi))\subset\Dom(\scr L(\xi)^*)$, and $S\xi=\scr L(S\xi)\Omega=\scr L(\xi)^*\Omega$. Since $\scr L(S\xi)$ has dense domain, so does $\scr L(\xi)^*$. Therefore $\scr L(\xi)$ is preclosed.
\end{proof}

We would like to generalize the above proposition to $\Ef$. For any $\wtd I\in\Jtd$, recall that $\wtd I'$ is the clockwise complement of $\wtd I$. We define ${\bpr\wtd I}\in\Jtd$ such that $(\bpr\wtd I)'=\wtd I$, and call ${\bpr\wtd I}$ the \textbf{anticlockwise complement of $\wtd I$}. Choose $\mc H_i\in\Obj(\RepA)$. (We do not assume $\mc H_i$ to be dualizable.)  For any $\xi\in\mc H_i$, we let $\scr L(\xi,\wtd I)$ (resp. $\scr R(\xi,\wtd I)$) act on any $\mc H_j\in\Obj(\RepA)$ as an unbounded operator $\mc H_j\rightarrow\mc H_i\boxtimes\mc H_j$ (resp. $\mc H_j\rightarrow\mc H_j\boxtimes\mc H_i$) with domain $\mc H_j(I')$ such that for any $\eta\in\mc H_j(I')$,
\begin{align}
\scr L(\xi,\wtd I)\eta=R(\eta,\wtd I')\xi,\quad \text{resp.}  \quad \scr R(\xi,\wtd I)\eta=L(\eta,\bpr\wtd I)\xi.
\end{align}
It is clear that $\Omega$ is inside the domains of $\scr L(\xi,\wtd I)|_{\mc H_0}$ and $\scr R(\xi,\wtd I)|_{\mc H_0}$, and the state-field correspondence
\begin{align}
\scr L(\xi,\wtd I)\Omega=\scr R(\xi,\wtd I)\Omega=\xi
\end{align}
is satisfied. We also have that
\begin{align}
\scr L(\xi,\wtd I)|_{\mc H_0}=\scr R(\xi,\wtd I)|_{\mc H_0},
\end{align}
and that they depend only on $I$ but not on the choice of $\arg_I$. Indeed, both operators send any $y\Omega\in\mc A(I')\Omega$ to $y\xi$.

The following is obvious.

\begin{lm}
Choose $\mc H_i,\mc H_j,\mc H_k,\mc H_l\in\Obj(\RepA)$ and $\xi\in\mc H_i$. If $G\in\Hom_{\mc A}(\mc H_i,\mc H_j)$ and $K\in\Hom_{\mc A}(\mc H_k,\mc H_l)$, then for each $\eta\in\mc H_k(I')$,
\begin{gather}
	\scr L(G\xi,\wtd I)K\eta=(G\otimes K)\scr L(\xi,\wtd I)\eta,\label{eq68}\\
\scr R(G\xi,\wtd I)K\eta=(K\otimes G)\scr R(\xi,\wtd I)\eta
\end{gather}
\end{lm}
\begin{proof}
Note that $K\eta\in\mc H_l(I')$. We have
\begin{align*}
\scr L(G\xi,\wtd I)K\eta=R(K\eta,\wtd I')G\xi\xlongequal{\eqref{eq8}} (G\otimes K) R(\eta,\wtd I')\xi=(G\otimes K)\scr L(\xi,\wtd I)\eta.
\end{align*}
This proves the first relation. A similar argument proves the second one.
\end{proof}

\begin{df}
For any $\mc H_i\in\Obj(\RepA)$ and $I\in\mc J$,  $\mc H_i^\pr(I)$  is the set of all $\xi\in\mc H_i$ such that $\scr L(\xi,\wtd I)|_{\mc H_0}=\scr R(\xi,\wtd I)|_{\mc H_0}$ is preclosed. It is clear that $\mc H_i(I)\subset\mc H_i^\pr(I)$.
\end{df}

It turns out that for any $\xi\in\mc H_i^\pr(I)$, $\scr L(\xi,\wtd I)$ is preclosed on any $\mc H_j\in\Obj(\RepA)$. To prove this, we first need a lemma.

\begin{lm}
Let $\mc H_i,\mc H_j\in\Obj(\RepA),\wtd I,\wtd J\in\Jtd$, and assume that $\wtd J$ is clockwise  to $\wtd I$. If $\xi\in\mc H_i$, $\xi_0\in\Dom((\scr L(\xi,\wtd I)|_{\mc H_0})^*)$, and $\eta,\eta_0\in\mc H_j(J)$, then
\begin{align}
\bk{\scr L(\xi,\wtd I)\eta|R(\eta_0,\wtd J)\xi_0}=\bk{R(\eta_0,\wtd J)^*\eta|\scr L(\xi,\wtd I)^*\xi_0}.\label{eq45}
\end{align}	
\end{lm}

\begin{proof}
Choose $\xi,\xi_0,\eta,\eta_0$ as in the lemma.  Recall that by proposition \ref{lb5}, we have $R(\eta_0,\wtd J)^*\eta\in\mc H_0(J)$ and
\begin{gather*}
R(R(\eta_0,\wtd J)^*\eta,\wtd I')|_{\mc H_i}=R(R(\eta_0,\wtd J)^*\eta,\wtd J)|_{\mc H_i}=R(\eta_0,\wtd J)^*R(\eta,\wtd J)|_{\mc H_i}.
\end{gather*}
Thus
\begin{align*}
&\bk{\scr L(\xi,\wtd I)\eta|R(\eta_0,\wtd J)\xi_0}=\bk{R(\eta,\wtd J)\xi|R(\eta_0,\wtd J)\xi_0}=\bk{R(\eta_0,\wtd J)^*R(\eta,\wtd J)\xi|\xi_0}\\
=&\bk{R(R(\eta_0,\wtd J)^*\eta,\wtd I')\xi|\xi_0}=\bk{\scr L(\xi,\wtd I)R(\eta_0,\wtd J)^*\eta|\xi_0}=\bk{R(\eta_0,\wtd J)^*\eta|\scr L(\xi,\wtd I)^*\xi_0}.
\end{align*}
\end{proof}

\begin{thm}\label{lb16}
Choose any $\mc H_i\in\Obj(\RepA)$, $\wtd I\in\Jtd$, and $\xi\in\mc H_i^\pr(I)$. Then $\scr L(\xi,\wtd I)|_{\mc H_j}$ and $\scr R(\xi,\wtd I)|_{\mc H_j}$ are preclosed for any $\mc H_j\in\Obj(\RepA)$.
\end{thm}

\begin{proof}
Let $\wtd J=\wtd I'$. Assume that $\eta_n$ is a sequence of vectors in $\mc H_j(J)$ converging to $0$ such that $\scr L(\xi,\wtd I)\eta_n$ converges to $\chi\in\mc H_i\boxtimes\mc H_j$. We shall show that $\chi=0$. Since $\scr L(\xi,\wtd I)|_{\mc H_0}$ is preclosed, $\mc W:=\Dom((\scr L(\xi,\wtd I)|_{\mc H_0})^*)$ is a dense subspace of $\mc H_i$. Since $R(\eta_0,\wtd J)$ is bounded for any $\eta_0\in\mc H_j(J)$, we conclude that $R(\mc H_j(J),\wtd J)\mc W$ is  dense in $R(\mc H_j(J),\wtd J)\mc H_i$ which spans a dense subspace of $\mc H_i\boxtimes\mc H_j$ by the density axiom of $\scr E$. Therefore, it suffices to show that $\bk{\chi|R(\eta_0,\wtd J)\xi_0}=0$ for any $\xi_0\in\mc W$ and $\eta_0\in\mc H_j(J)$. We notice that $\bk{\chi|R(\eta_0,\wtd J)\xi_0}$ is the limit of
\begin{align*}
\bk{\scr L(\xi,\wtd I)\eta_n|R(\eta_0,\wtd J)\xi_0}\xlongequal{\eqref{eq45}} \bk{R(\eta_0,\wtd J)^*\eta_n|\scr L(\xi,\wtd I)^*\xi_0},
\end{align*}
which converges to $0$ since $R(\eta_0,\wtd J)$ is bounded. This proves that $\scr L(\xi,\wtd I)|_{\mc H_j}$ is precloded. As for $\scr R(\xi,\wtd I)|_{\mc H_j}$, the argument is similar.
\end{proof}

We need two lemmas before we prove the main result of this section.

\begin{lm}\label{lb37}
Let $A,B$ be densely defined unbounded linear operators on a Hilbert space $\mc H$ with common domain $\Dom(A)=\Dom(B)=\Dom$. Assume that  the ranges of $A$ and $B$ are mutually orthogonal, and that there exist mutually orthogonal projections $p,q$ on $\mc H$ satisfying
\begin{align*}
A\subset Ap,\qquad 	B\subset Bq.
\end{align*}
Then $A+B$ is preclosed if and only if both $A$ and $B$ are so.

\begin{proof}
First, assume $A,B$ are preclosed. Let $e,f$ be the projections of $\mc H$ onto the ranges of $A,B$ respectively. Choose $\xi_n\in\Dom$ converging to $0$ such that $(A+B)\xi_n$ converges. Then $A\xi_n=eA\xi_n+eB\xi_n$, which converges. Since $A$ is preclosed, $A\xi_n\rightarrow 0$. Similarly, $B\xi_n\rightarrow 0$. So $(A+B)\xi_n\rightarrow 0$.

Now, assume $A+B$ is preclosed. Choose $\xi_n\in\scr D$ converging to $0$ such that $A\xi_n$ converges. Since $A\subset Ap$, we have $p\xi_n\in p\Dom\subset\Dom$, and $Ap\xi_n=A\xi_n$. Similarly, $qp\xi_n\in\Dom$ and $Bqp\xi_n=Bp\xi_n$. But $pq=0$. So $Bp\xi_n=0$. So $A\xi_n=(A+B)p\xi_n$, which we know converges. Since $A+B$ is preclosed and $p\xi_n\rightarrow 0$, we must have $(A+B)p\xi_n\rightarrow 0$. So $A\xi_n\rightarrow 0$. So $A$ is preclosed. Similarly, $B$ is also preclosed.
\end{proof}
 
\end{lm}

\begin{lm}\label{lb38}
Let $\mc H_i,\mc H_j,\mc H_k\in\Obj(\RepfA)$ and $\xi\in\mc H_i$. Suppose that $\varphi\in\Hom_{\mc A}(\mc H_j,\mc H_k)$ is an isometry. Then $\scr L(\xi,\wtd I)|_{\mc H_j(I')}$ (resp. $\scr R(\xi,\wtd I)|_{\mc H_j(I')}$) is preclosed if and only if $\scr L(\xi,\wtd I)|_{\mc H_j(I')}\cdot \varphi^*|_{\mc H_k(I')}$ (resp.  $\scr R(\xi,\wtd I)|_{\mc H_j(I')}\cdot \varphi^*|_{\mc H_k(I')}$) is so.
\end{lm}

\begin{proof}
We prove this for $\scr L$. It is an easy exercise that a preclosed operator multiplied by a bounded operator from the right is preclosed. Thus, the ``only if" part follows immediately. The ``if" part follows from that $\scr L(\xi,\wtd I)|_{\mc H_j(I')}$ equals $\scr L(\xi,\wtd I)|_{\mc H_j(I')}\cdot \varphi^*|_{\mc H_k(I')}$ times $\varphi|_{\mc H_j(I')}$.
\end{proof}

We now relate the preclosedness of $\scr L(\xi,\wtd I)$ with the domain of $\Delta_I^{\frac 12}$. 

\begin{thm}
Let $\mc H_i\in\Obj(\RepfA)$ and $\wtd I\in\Jtd$. Choose an arbitrary dual object $\mc H_{\ovl i}$ and standard $\ev_{i,\ovl i},\ev_{\ovl i,i}$  with adjoints $\coev_{i,\ovl i},\coev_{\ovl i,i}$. Then for any $\xi\in\mc H_i$, the following are equivalent:

(a) $\Omega$ is in the domain of $\scr L(\xi,\wtd I)^*\coev_{i,\ovl i}$.

(a') $\Omega$ is in the domain of $\scr R(\xi,\wtd I)^*\coev_{\ovl i,i}$.

(b) $\xi$ is in the domain of $\Delta_I^{\frac 12}|_{\mc H_i}$.\\
If any of them is true, then $\xi\in\mc H_i^\pr(I)$, and 
\begin{align}
S_{\wtd I}\xi=\scr L(\xi,\wtd I)^*\coev_{i,\ovl i}\Omega,\qquad F_{\wtd I}\xi=\scr R(\xi,\wtd I)^*\coev_{\ovl i,i}\Omega.\label{eq53}
\end{align}
\end{thm}

\begin{proof}
We consider dualized objects. In particular, we let $(\mc H_{\ovl i},\ev_{i,\ovl i})$ be the dualizing data for $\mc H_i$. Choose any $\xi\in\mc H_i$. Let $\mc H_k=\mc H_0\oplus\mc H_i$, let $\mc H_a=\mc H_k\boxtimes\mc H_{\ovl k}$, and consider the standard $C^*$-Frobenius algebra $Q=(\mc H_a,\mu,\iota)$. (See the discussions starting from \eqref{eq49}.) Let $\mc B,\mc B'$ be the associated non-local extensions of $\mc A$.	Identify $\mc H_i$ with $\mc H_i\boxtimes\mc H_0$ using the right unitor. Then $\psi_{i,0}:\mc H_i=\mc H_i\boxtimes\mc H_0\rightarrow\mc H_a$ defined in \eqref{eq64} is an isometry. Let $\mc M=\mc B(\wtd I)$. Since $\psi_{i,0}\xi\in\mc H_a$, we can define an unbounded operator $\scr L^Q(\psi_{i,0}\xi)$ on $\mc H_a$ with domain $\mc M'\iota\Omega$ such that $\scr L^Q(\psi_{i,0}\xi)y\iota\Omega=y\psi_{0,i}\xi$ for any $y\in\mc M'=\mc B'(\wtd I')$. 

Step 1. We first shown that $\xi\in\mc H_i^\pr(I)$ is equivalent to that $\scr L^Q(\psi_{i,0}\xi)$ is preclosed.  Note that $\mc M'\iota\Omega=\mc H_a(I')$ by \eqref{eq11}, and that elements in $\mc M'$ are of the form $B(\eta,\wtd I')=\mu R(\eta,\wtd I')|_{\mc H_a}$ where $\eta\in\mc H_a(I')$.  We compute that for any $\eta\in\mc H_a(I')$,
\begin{align*}
&\scr L^Q(\psi_{i,0}\xi)\eta\xlongequal{\eqref{eq38}}\scr L^Q(\psi_{i,0}\xi)B(\eta,\wtd I')\iota\Omega=B(\eta,\wtd I')\psi_{i,0}\xi=\mu R(\eta,\wtd I')\psi_{i,0}\xi\\
=&\mu\scr L(\psi_{i,0}\xi,\wtd I)\eta\xlongequal{\eqref{eq68}}\mu(\psi_{i,0}\otimes\id_a)\scr L(\xi,\wtd I)\eta,
\end{align*}
which shows
\begin{align}
\scr L^Q(\psi_{i,0}\xi)=\mu(\psi_{i,0}\otimes\id_a)\scr L(\xi,\wtd I)|_{\mc H_a(I')}\label{eq46}
\end{align}
with common domain $\mc H_a(I')$.

Since $\id_a=\sum_{s,t\in\{0,i\}}\psi_{s,\ovl t}\psi_{s,\ovl t}^*$ where each $\psi_{s,\ovl t}\psi_{s,\ovl t}^*$ is the projection of $\mc H_a$ onto a subrepresentation equivalent to $\mc H_s\boxtimes\mc H_{\ovl t}$, we have
\begin{align*}
&\scr L^Q(\psi_{i,0}\xi)=\sum_{s,t\in\{0,i\}}\mu(\psi_{i,0}\otimes\psi_{s,\ovl t})(\id_i\otimes\psi_{s,\ovl t}^*)\scr L(\xi,\wtd I)|_{\mc H_a(I')}\\
\xlongequal{\eqref{eq68}}&	\sum_{s,t\in\{0,i\}}\mu(\psi_{i,0}\otimes\psi_{s,\ovl t})\scr L(\xi,\wtd I)\psi_{s,\ovl t}^*|_{\mc H_a(I')}.
\end{align*}
By \eqref{eq48}, we have 
\begin{align}
\scr L^Q(\psi_{i,0}\xi)=\psi_{i,0}\scr L(\xi,\wtd I)\psi_{0,0}^*|_{\mc H_a(I')}+\psi_{i,\ovl i}\scr L(\xi,\wtd I)\psi_{0,\ovl i}^*|_{\mc H_a(I')}.\label{eq69}
\end{align}
Apply Lemma \ref{lb37} by choosing  $\Dom=\mc H_a(I')$, $p=\psi_{0,0}\psi_{0,0}^*$, and $q=\psi_{0,\ovl i}\psi_{0,\ovl i}^*$. We see that $\scr L^Q(\psi_{i,0}\xi)$ is preclosed if and only if the two terms on the right hand side of \eqref{eq69} are both preclosed. The latter condition is, by Lemma \ref{lb38} and Theorem \ref{lb16}, equivalent to that $\scr L(\xi,\wtd I)|_{\mc H_0(I')}$ is preclosed, i.e., that $\xi\in\mc H_i^\pr(I)$. 

Step 2. We now show that the (a) of this theorem is equivalent to that of Proposition \ref{lb17}, namely, that $\iota\Omega$ is in the domain of $\scr L^Q(\psi_{i,0}\xi)^*$. We have
\begin{align}
\psi_{i,\ovl i}^*\iota\Omega\xlongequal{\eqref{eq48}}(\psi_{i,0}\otimes\psi_{0,\ovl i})^*\mu^*\iota\Omega\xlongequal{\eqref{eq70}} \coev_{i\boxtimes 0,0\boxtimes\ovl i}\Omega=\coev_{i,\ovl i}\Omega,\label{eq71}
\end{align}
and similarly
\begin{align}
\psi_{i,0}^*\iota\Omega\xlongequal{\eqref{eq48}}(\psi_{i,0}\otimes\psi_{0,0})^*\mu^*\iota\Omega\xlongequal{\eqref{eq70}} 0.\label{eq72}
\end{align}
Thus, by \eqref{eq69}, that $\iota\Omega$ is in the domain of $\scr L^Q(\psi_{i,0}\xi)^*$ is equivalent to that $\coev_{i,\ovl i}\Omega$ is in the domain of $(\scr L(\xi,\wtd I)\psi_{0,\ovl i}^*|_{\mc H_a(I')})^*$. This is true if and only if the function
\begin{align*}
	\eta\in\mc H_a(I')\mapsto \bk{\scr L(\xi,\wtd I)\psi_{0,\ovl i}^*\eta|\coev_{i,\ovl i}\Omega}
\end{align*}
is continuous. Using \eqref{eq62}, it is easy to see that the above is equivalent to the continuity of 
\begin{align*}
	\eta\in\mc H_{\ovl i}(I')\mapsto \bk{\scr L(\xi,\wtd I)\eta|\coev_{i,\ovl i}\Omega},
\end{align*}
which is precisely (a) of our theorem. 

Step 3. By Proposition \ref{lb9}, the modular operator for $\mc B(I)$ and $\iota\Omega$ is ${\Delta_I}|_{\mc H_a}$. The vector $\psi_{i,0}\xi$ is inside the domain of ${\Delta_I}^{\frac 12}|_{\mc H_a}$ if and only if $\xi$ is in the domain of ${\Delta_I}^{\frac 12}|_{\mc H_i}$, since we have $\psi_{i,0}{\Delta_I}^{\frac 12}|_{\mc H_i}={\Delta_I}^{\frac 12}|_{\mc H_a}\psi_{i,0}$ due to Proposition \ref{lb14}. Thus,  by Proposition \ref{lb17}, we see that (a), (b), and $\xi\in\mc H_i^\pr(I)$ are equivalent. Moreover, if they are true, then by the last statement of Proposition \ref{lb17}, (recall Proposition \ref{lb9} and note $\epsilon=1$) we have
\begin{align*}
S_{\wtd I}|_{\mc H_a}\psi_{i,0}\xi=\scr L^Q(\psi_{i,0}\xi)^*\iota\Omega,
\end{align*}
which by \eqref{eq69}, \eqref{eq71}, \eqref{eq72} equals
\begin{align*}
	\psi_{0,\ovl i}\scr L(\xi,\wtd I)^*\coev_{i,\ovl i}\Omega.
\end{align*}
Since $\ovl{\psi_{i,0}}=\psi_{0,\ovl i}$ and hence $\psi_{0,\ovl i} S_{\wtd I}|_{\mc H_i}\xi=S_{\wtd I}|_{\mc H_a}\psi_{i,0}\xi$, we obtain $S_{\wtd I}\xi=\scr L(\xi,\wtd I)^*\coev_{i,\ovl i}\Omega$.

To prove the equivalence of (a') and (b) and the second equation of \eqref{eq53}, construct $Q=(\mc H_{\ovl k}\boxtimes\mc H_k ,\mu,\iota)$ where $\mc H_k=\mc H_0\oplus\mc H_i$, and use similar arguments by choosing $\mc M=\mc B'(\wtd I)$ and considering $\scr R^Q(\psi_{0,i}\xi)$ defined by sending each $y\iota\Omega\in\mc M'\iota\Omega$ to $y\psi_{0,i}\xi$.
\end{proof}

\appendix

\section{Connes categorical extensions}\label{lb23}

In this appendix section, we sketch the constuction of categorical extensions via Connes fusion. Details can be found in \cite{Gui21a}. We do not assume the representations to be dualizable in this section.

Recall that we set $\mc H_i(I)=\Hom_{\mc A(I')}(\mc H_0,\mc H_i)\Omega$. For any $\xi\in\mc H_i(I)$, we define $Z(\xi,\wtd I)$ to be the unique element in $\Hom_{\mc A(I')}(\mc H_0,\mc H_i)$ satisfying $\xi=Z(\xi,\wtd I)\Omega$. If $J$ is disjoint from $I$, we define a (degenerate) inner product $\bk{\cdot|\cdot}$ (antilinear on the second variable) on the algebraic tensor product $\mc H_i(I)\otimes\mc H_j(J)$ such that for any $\xi_1,\xi_2\in\mc H_i(I)$ and $\eta_1,\eta_2\in\mc H_j(J)$,
\begin{align}
\bk{\xi_1\otimes\eta_1|\xi_2\otimes\eta_2}=\bk {Z(\eta_2,J)^*Z(\eta_1,J)Z(\xi_2,I)^*Z(\xi_1, I)\Omega|\Omega}.
\end{align}
This is nothing but the formula of Connes relative tensor product. We let $\mc H_i(I)\boxtimes\mc H_j(J)$ be the Hilbert space completion of $\mc H_i(I)\otimes\mc H_j(J)$ under this inner product. Note that $\mc A(I)$ and $\mc A(J)$ can naturally act on $\mc H_i(I)\otimes\mc H_j(J)$ by acting on the first resp. second component of the tensor product.

If $I_0\subset I$ and $J_0\subset J$, then $\mc H_i(I_0)$ and $\mc H_j(J_0)$ are dense in $\mc H_i(I)$ and $\mc H_j(J)$. Note that $\mc H_i(I_0)\otimes\mc H_j(J_0)$ is also dense in $\mc H_i(I)\otimes\mc H_j(J)$ under the above inner product. (Indeed, $\mc H_i(I)\otimes\mc H_j(J_0)$ is  dense in $\mc H_i(I)\otimes\mc H_j(J)$ since the $Z$ operators associated to the vectors of $\mc H_i(I)$ are bounded. Likewise, $\mc H_i(I_0)\otimes\mc H_j(J_0)$ is also dense in $\mc H_i(I)\otimes\mc H_j(J_0)$ since the $Z$ operators associated to the vectors of $\mc H_j(J_0)$ are bounded.) Thus we have  a natural unitary map $\mc H_i(I_0)\boxtimes\mc H_j(J_0)\xrightarrow{\simeq} \mc H_i(I)\boxtimes\mc H_j(J)$ induced by inclusion of intervals. Its adjoint $\mc H_i(I)\boxtimes\mc H_j(J)\xrightarrow{\simeq} \mc H_i(I_0)\boxtimes\mc H_j(J_0)$ is the natural unitary map induced by restriction of intervals. Both maps commute with the actions of $\mc A(I_0)$ and $\mc A(J_0)$.

Let $I_1,I_2$ be disjoint respectively from $J_1,J_2$. First, assume that these two pairs of intervals "overlap properly", which means that $I_1\cap I_2,J_1\cap J_2\in\mc J$. Then
\begin{align}
\mc H_i(I_1)\boxtimes \mc H_j(J_1)\xrightarrow{\simeq}\mc H_i(I_1\cap I_2)\boxtimes \mc H_j(J_1\cap J_2)\xrightarrow{\simeq}\mc H_i(I_2)\boxtimes \mc H_j(J_2)\label{eq50}
\end{align}
defines a natural unitary map $\mc H_i(I_1)\boxtimes \mc H_j(J_1)\xrightarrow{\simeq}\mc H_i(I_2)\boxtimes \mc H_j(J_2)$. In general, we need to choose a path $\gamma:[0,1]\rightarrow \Conf_2(\mbb S^1)$, where $\Conf_2(\mbb S^1)=\{(z,w)\in\mbb S^1:z\neq w \}$. We assume that $\gamma(0)\in I_1\times J_1$ and $\gamma(1)\in I_2\times J_2$. Then we can define a natural unitary map $\gamma^\bullet:\mc H_i(I_1)\boxtimes \mc H_j(J_1)\xrightarrow{\simeq}\mc H_i(I_2)\boxtimes \mc H_j(J_2)$ by covering $\gamma$ by a chain of pairs of open intervals $(K_1,L_1),\dots (K_n,L_n)$, such that $K_1=I_1,L_1=J_1,K_n=I_2,L_n=J_2$, and that for each $l=1,2,\dots,n$, $(K_{l-1},L_{l-1})$ and $(K_l,K_l)$ overlap properly. Then we can use a chain of unitary maps induced by restriction and inclusion of intervals (as in \eqref{eq50}) to define $\gamma^\bullet$. We say that $\gamma^\bullet$ is the path continuation induced by $\gamma$. If we have paths $\gamma_1,\gamma_2$, then $(\gamma_1\ast\gamma_2)^\bullet=\gamma_1^\bullet\gamma_2^\bullet$. Moreover, $\gamma^\bullet$ depends only on the homotopy class of $\gamma$.  Then we can transport the actions of $\mc A(I_1),\mc A(J_1)$ from $\mc H_i(I_1)\boxtimes \mc H_j(I_1)$ to $\mc H_i(I_2)\boxtimes \mc H_j(I_2)$ through the map $\gamma^\bullet$. Indeed, the result of transportation is independent of the choice of $\gamma$. It turns out that we have a well defined action of $\mc A$ on any Connes fusion $\mc H_i(I)\boxtimes\mc H_j(J)$ so that it restricts to the standard actions of $\mc A(I),\mc A(J)$ on $\mc H_i(I)\boxtimes\mc H_j(J)$, and that the actions of $\mc A$ commute with all path continuations. Thus $\mc H_i(I)\boxtimes\mc H_j(J)$ becomes a representation of $\mc A$, and the path continuations are unitary isomorphisms of $\mc A$-modules.

In the construction of the tensor category $\RepA$, we let $\mc H_i\boxtimes \mc H_j$ be $\mc H_i(\mbb S^1_+)\boxtimes\mc H_j(\mbb S^1_-)$, where $\mbb S^1_+,\mbb S^1_-\in\mc J$ are respectively the upper and lower semi-circle. Choose $\arg_{\mbb S^1_+}$ (resp. $\arg_{\mbb S^1_-}$) to be the one whose values are inside $(0,\pi)$ (resp. $(-\pi,0)$). This defines $\wtd{\mbb S^1_+}$ and $\wtd{\mbb S^1_-}$.  Now, for each $\wtd I=(I,\arg_I)\in\Jtd$ and $\xi\in\mc H_i(I)$, we  describe the operator $L(\xi,\wtd I)$. Recall that $\wtd I'$ is the clockwise complement of $\wtd I$, which means that $I'$ is (the interior of) the complement of $I$, and $\wtd I'$ is clockwise to $\wtd I$. We write $\wtd I'=(I',\arg_{I'})$.  The action $Z(\xi,I):\eta\in\mc H_j(I')\mapsto \xi\otimes\eta\in\mc H_i(I)\boxtimes\mc H_j(I')$ is a bounded operator which intertwines the actions of $\mc A(I')$. Thus $Z(\xi,I)\in\Hom_{\mc A(I')}(\mc H_j,\mc H_i(I)\boxtimes\mc H_j(I'))$. We now choose a path $\gamma:[0,1]\rightarrow\Conf_2(\mbb S^1)$ from $I\times I'$ to $\mbb S^1_+\times\mbb S^1_-$ such that the arguments of $\wtd I$ and $\wtd I'$ are changing continuously to those of $\wtd{\mbb S^1_+}$ and $\wtd{\mbb S^1_-}$ respectively along $\gamma$. Then for each $\eta\in\mc H_j$,  $L(\xi,\wtd I)\eta$ is defined to be
\begin{align}
L(\xi,\wtd I)\eta=\gamma^\bullet Z(\xi,I)\eta\in\mc H_i\boxtimes\mc H_j.
\end{align}
Define a path $\rho:[0,1]\rightarrow\Conf_2(\mbb S^1)$ from $\mbb S^1_+\times\mbb S^1_-$ to $\mbb S^1_-\times\mbb S^1_+$ by $\rho(t)=(e^{\im\pi (\frac 12-t)},e^{\im\pi (-\frac 12-t)})$. Then
\begin{align}
\rho^\bullet: \mc H_i\boxtimes\mc H_j=\mc H_i(\mbb S^1_+)\boxtimes\mc H_j(\mbb S^1_-)\rightarrow \mc H_i(\mbb S^1_-)\boxtimes\mc H_j(\mbb S^1_+)\simeq\mc H_j\boxtimes\mc H_i
\end{align}
is the braiding $\ss_{i,j}$. We define $R(\xi,\wtd I)\eta=\ss_{i,j}L(\xi,\wtd I)\eta$. That these operators define a categorical extension $\scr E=(\mc A,\RepA,\boxtimes,\mc H)$ of $\mc A$ (called Connes categorical extension) was proved in \cite{Gui21a}.

\section{M\"obius covariance of categorical extensions}\label{lb18}

Let $\Ef=(\mc A,\RepfA,\boxtimes,\mc H)$ be the rigid (vector-labeled and closed) categorical extension of the M\"obius covariant net $\mc A$ as in Section \ref{lb30}. Recall that we assume objects in $\RepfA$ (which are dualizable) to be M\"obius covariant, which means there is a strongly continuous unitary representation $U_i$ of $\UPSU$ on $\mc H_i$ such that  \eqref{eq20} holds. (Recall that such representation is unique by \cite{GL96}.) In this section we prove theorem \ref{lb19}, namely, that the fusion of two M\"obius covariant representations is also M\"obius covariant, and that $\Ef$ is M\"obius covariant. We remark that the arguments in this section can also be used to show that $\scr E$ is conformal covariant when $\mc A$ is so; see the end of the section.

We first notice the following easy fact:

\begin{lm}\label{lb21}
Any morphism in $\RepfA$ commutes with the actions of $\UPSU$.
\end{lm}
\begin{proof}
Let $\mc H_i,\mc H_j\in\RepfA$ and $G\in\Hom_{\mc A}(\mc H_i,\mc H_j)$. Let $\mc H_k=\mc H_i\oplus\mc H_j$. Then the unique representation $U_k$ of $\UPSU$ on $\mc H_k$ is described by $U_k(g)=\diag(U_i(g),U_j(g))$ for any $g\in\UPSU$. We regard $G$ as an endomorphism of $\mc H_k$ by acting trivially on $\mc H_j$. Then it suffices to show that any endomorphism of $\mc H_k$ commutes with the action of $\UPSU$ on $\mc H_k$. By linearity, it suffices to prove this for any unitary $V\in\End_{\mc A}(\mc H_k)$. Then, by the uniqueness of the representation of $\UPSU$ on $\mc H_k$, we have $U_k(g)=VU_k(g)V^*$ for any $g\in\UPSU$. Therefore $V$ commutes with the action of $\UPSU$. 
\end{proof}

Choose any $\mc H_j\in\Obj(\RepfA),\wtd I\in\Jtd$, and define a unitary representation $V_{\wtd I}$ of $\UPSU$ on $\mc H_i\boxtimes\mc H_j$ by setting
\begin{align}
V_{\wtd I}(g)L(\xi,\wtd I)\eta=L(g\xi,g\wtd I)g\eta\label{eq54}
\end{align}  
for any $\xi\in\mc H_i(I),\eta\in\mc H_j(I'),g\in\UPSU$. Note that we have $g\mc H_i(I)=\mc H_i(gI)$ since, by \eqref{eq20}, we have $g\Hom_{\mc A(I')}(\mc H_0,\mc H_i)g^{-1}=\Hom_{\mc A(gI')}(\mc H_0,\mc H_i)$. Thus we have
\begin{align*}
L(g\xi,g\wtd I)=gL(\xi,\wtd I)g^{-1}
\end{align*}
when acting on $\mc H_0$. Now,   we choose any $\xi_1,\xi_2\in\mc H_i(I),\eta_1,\eta_2\in \mc H_j(I')$, and use the locality of $\Ef$ and the fact that $g\Omega=\Omega$ to compute that
\begin{align}
&\bk{L(g\xi_1,g\wtd I)g\eta_1|L(g\xi_2,g\wtd I)g\eta_2}=\bk{L(g\xi_1,g\wtd I)R(g\eta_1,g\wtd I')\Omega|L(g\xi_2,g\wtd I)R(g\eta_2,g\wtd I')\Omega}\nonumber\\
=&\bk{L(g\xi_2,g\wtd I)^*L(g\xi_1,g\wtd I)\Omega|R(g\eta_1,g\wtd I')^*R(g\eta_2,g\wtd I')\Omega}\nonumber\\
=&\bk{gL(\xi_2,\wtd I)^*L(\xi_1,\wtd I)\Omega|gR(\eta_1,\wtd I')^*R(\eta_2,\wtd I')\Omega}=\bk{L(\xi_2,\wtd I)^*L(\xi_1,\wtd I)\Omega|R(\eta_1,\wtd I')^*R(\eta_2,\wtd I')\Omega}\nonumber\\
=&\bk{L(\xi_1,\wtd I)\eta_1|L(\xi_2,\wtd I)\eta_2}.\label{eq56}
\end{align}
This proves the well-definedness and the unitarity of $V_{\wtd I}(g)$.

Notice that $V_{\wtd I}$ is independent of $\wtd I$, namely, $V_{\wtd I}=V_{\wtd I_0}$ when $\wtd I,\wtd I_0\in\Jtd$. Indeed, it suffices to check this when $\wtd I_0\subset \wtd I$. In that case, the actions of $V_{\wtd I}(g)$ and $V_{\wtd I_0}(g)$ on $L(\mc H_i(I_0),\wtd I)\mc H_j(I')$ are clearly the same. So they must be equal. We write $V_{\wtd I}$ as $V$ for short. From our definition \eqref{eq54}, it is clear that $V(gh)=V(g)V(h)$ for any $g,h\in\UPSU$. Thus $V$ is a representation of $\UPSU$.

We choose $\wtd I_0\subset \wtd I$ such that $I_0\subset\joinrel\subset I$.  To check the continuity of the representation $V$, we need to show that for any sequence of elements $g_n$ in $\UPSU$ converging to $1$, $L(g_n\xi,g_n\wtd I_0)g_n\eta$ converges to $L(\xi,\wtd I_0)\eta$ for any $\xi\in\mc H_i(I_0)$ and $\eta\in\mc H_j$. Assume without loss of generality that $g_nI_0\subset I$ for any $n$. Since $g_n\eta$ converges to $\eta$, it suffices to show that $L(g_n\xi,g_n\wtd I_0)|_{\mc H_j}=L(g_n\xi,\wtd I)|_{\mc H_j}$ converge strongly to $L(\xi,\wtd I)$ and are uniformly bounded over $n$. By the locality and the state-field correspondence of $\Ef$,  $L(g_n\xi,\wtd I)\chi=R(\chi,\wtd I')g_n\xi$ converges to $R(\chi,\wtd I')\xi=L(\xi,\wtd I)\chi$ for any $\chi\in\mc H_j(I')$. Set $x=L(\xi,\wtd I)^*L(\xi, \wtd I)|_{\mc H_0}\in\mc A(I)$. Then $x_n:=g_nxg_n^*$ equals $L(g_n\xi,g_n\wtd I)^*L(g_n\xi, g_n\wtd I)|_{\mc H_0}$. So $x_n\Omega=L(g_n\xi,g_n\wtd I)^*g_n\xi$. Hence, by proposition \ref{lb5},
\begin{align*}
&\lVert L(g_n\xi,\wtd I)|_{\mc H_j} \lVert^2=\lVert L(g_n\xi,\wtd I)^*L(g_n\xi,\wtd I)|_{\mc H_j} \lVert=\lVert L(L(g_n\xi,\wtd I)^*g_n\xi,\wtd I)|_{\mc H_j} \lVert\\
=&\lVert L(x_n\Omega,g_n\wtd I)|_{\mc H_j} \lVert=\lVert \pi_{j,g_nI}(x_n) \lVert\leq \lVert x_n \lVert=\lVert x\lVert.
\end{align*}
This shows that $\lVert L(g_n\xi,\wtd I)|_{\mc H_j} \lVert$ is uniformly bounded over all $n$. Thus $L(g_n\xi,\wtd I)$ converges strongly to $L(\xi,\wtd I)$.

To show that $V$ makes $\mc H_i\boxtimes\mc H_j$ M\"obius covariant, we need to check that
\begin{align*}
V(g)\pi_{i\boxtimes j,J}(x)=\pi_{i\boxtimes j,gJ}(gxg^{-1})V(g)
\end{align*}
for any $g\in\UPSU,J\in\mc J,x\in\mc A(J)$. It suffices to verify this equation when both sides act on $L(\mc H_i(I),\wtd I)\mc H_j(J)$ where $I$ is disjoint from $J$. This is easy. Therefore $\mc H_i\boxtimes\mc H_j$ is M\"obius covariant, and the unique representation of $\UPSU$ is described by \eqref{eq54}. From \eqref{eq54} it is clear that $L(g\xi,g\wtd I)=gL(\xi,\wtd I)g^{-1}$ is always true on any $\mc H_j\in\Obj(\RepfA)$. By the braiding axiom of $\Ef$ and lemma \ref{lb21}, we also have $R(g\xi,g\wtd I)=gR(\xi,\wtd I)g^{-1}$. This proves the M\"obius covariance of $\Ef$.

Finally,  we explain how the above arguments can be adapted to show the conformal covariance of $\scr E$ when $\mc A$ is conformal covariant. Let $\mc A$ be conformal covariant.   Recall that for any $\mc H_i\in\Obj(\RepA)$, $\xi\in\mc H_i(I)$, and $g\in\GA$, we have $g\xi g^{-1}\in\mc H_i(gI)$ where $g\xi g^{-1}:=gL(\xi,\wtd I)g^{-1}\Omega=gR(\xi,\wtd I)g^{-1}\Omega$. Thus, for any $\mc H_i,\mc H_j\in\Obj(\RepA)$, one can define an action of $\GA$ on $\mc H_i\boxtimes\mc H_j$ by setting
\begin{align}
V_{\wtd I}(g)L(\xi,\wtd I)\eta=L(g\xi g^{-1},g\wtd I)g\eta\label{eq55}
\end{align}  
for any $\xi\in\mc H_i(I),\eta\in\mc H_j(I'),g\in\GA$. Note that \eqref{eq55} also equals
\begin{align*}
L(g\xi g^{-1},g\wtd I)gR(\eta,\wtd I')g^{-1}g\Omega=L(g\xi g^{-1},g\wtd I)R(g\eta g^{-1},g\wtd I')g\Omega
\end{align*}
since we have $gR(\eta,\wtd I')g^{-1}=R(g\eta g^{-1},g\wtd I')$ when acting on $\mc H_0$. Using this relation and the calculations as in \eqref{eq56}, one checks that $V_{\wtd I}(g)$ is well-defined and unitary.  Similar arguments as in the above paragraphs show that $V_{\wtd I}$ is independent of $\wtd I$, that  $V$ respects the group multiplication of $\GA$ (which follows clearly from the definition of $V_{\wtd I}(g)$, that the representation $V:\GA\curvearrowright\mc H_i\boxtimes\mc H_j$ is continuous, and that \eqref{eq10} holds. Thus $V$ is the unique representation of $\GA$ making $\mc H_i\boxtimes\mc H_j$ conformal covariant. The relations \eqref{eq1} follow easily from the definition \eqref{eq55} of $V$.

\printindex

\noindent {\small \sc Department of Mathematics, Rutgers University, USA.}

\noindent {\em E-mail}: bin.gui@rutgers.edu\qquad binguimath@gmail.com

\newpage

\begin{thebibliography}{99999}
\footnotesize

\bibitem[AFK04]{AFK04}
D'antoni, C., Fredenhagen, K. and K\"oster, S., 2004. Implementation of conformal covariance by diffeomorphism symmetry. Letters in Mathematical Physics, 67(3), pp.239-247.

\bibitem[ALR01]{ALR01}
D'Antoni, C., Longo, R. and Radulescu, F., 2001. Conformal nets, maximal temperature and models from free probability. Journal of Operator Theory, pp.195-208.


\bibitem[BCL98]{BCL98}
Bertozzini, P., Conti, R. and Longo, R., 1998. Covariant Sectors with Infinite Dimension and Positivity of the Energy. Communications in mathematical physics, 193(2), pp.471-492.

\bibitem[BDH14]{BDH14}
Bartels, A., Douglas, C.L. and Henriques, A., 2014. Dualizability and index of subfactors. Quantum topology, 5(3), pp.289-345.

\bibitem[BDH15]{BDH15}
Bartels, A., Douglas, C.L. and Henriques, A., 2015. Conformal nets I: Coordinate-free nets. International Mathematics Research Notices, 2015(13), pp.4975-5052.

\bibitem[BDH17]{BDH17}
Bartels, A., Douglas, C.L. and Henriques, A., 2017. Conformal nets II: Conformal blocks. Communications in Mathematical Physics, 354(1), pp.393-458.

\bibitem[BGL93]{BGL93}
Brunetti, R., Guido, D. and Longo, R., 1993. Modular structure and duality in conformal quantum field theory. Communications in Mathematical Physics, 156(1), pp.201-219.

\bibitem[BKLR15]{BKLR15}
Bischoff, M., Kawahigashi, Y., Longo, R.,  Rehren, K. H. (2015). Tensor categories and endomorphisms of von neumann algebras: with applications to quantum field theory, Springer Briefs in Mathematical Physics, vol. 3.

\bibitem[BW75]{BW75}
Bisognano, J.J. and Wichmann, E.H., 1975. On the duality condition for a Hermitian scalar field. Journal of Mathematical Physics, 16(4), pp.985-1007.


\bibitem[Bar54]{Bar54}
Bargmann, V., 1954. On unitary ray representations of continuous groups. Annals of Mathematics, pp.1-46.

\bibitem[Bor92]{Bor92}
Borchers, H.J., 1992. The CPT-theorem in two-dimensional theories of local observables. Communications in Mathematical Physics, 143(2), pp.315-332.

\bibitem[CKL08]{CKL08}
Carpi, S., Kawahigashi, Y. and Longo, R., 2008, October. Structure and classification of superconformal nets. In Annales Henri Poinca\'e (Vol. 9, No. 6, pp. 1069-1121). SP Birkh\"auser Verlag Basel.

\bibitem[CKLW18]{CKLW18}
Carpi, S., Kawahigashi, Y., Longo, R. and Weiner, M., 2018. From vertex operator algebras to conformal nets and back (Vol. 254, No. 1213). Memoirs of the American Mathematical Society


\bibitem[Con80]{Con80}
Connes, A., 1980. On the spatial theory of von Neumann algebras. Journal of Functional Analysis, 35(2), pp.153-164.




\bibitem[DSW86]{DSW86}
Driessler, W., Summers, S.J. and Wichmann, E.H., 1986. On the connection between quantum fields and von Neumann algebras of local operators. Communications in mathematical physics, 105(1), pp.49-84.



\bibitem[FJ96]{FJ96}
Fredenhagen, K. and J\"or\ss, M., 1996. Conformal Haag-Kastler nets, pointlike localized fields and the existence of operator product expansions. Communications in mathematical physics, 176(3), pp.541-554.

\bibitem[FRS89]{FRS89}
Fredenhagen, K., Rehren, K.H. and Schroer, B., 1989. Superselection sectors with braid group statistics and exchange algebras. Communications in Mathematical Physics, 125(2), pp.201-226.

\bibitem[FRS92]{FRS92}
Fredenhagen, K., Rehren, K.H. and Schroer, B., 1992. Superselection sectors with braid group statistics and exchange algebras II: Geometric aspects and conformal covariance. Reviews in Mathematical Physics, 4(spec01), pp.113-157.

\bibitem[Fal00]{Fal00}
Falcone, T., 2000. $L^2$-von Neumann modules, their relative tensor products and the spatial derivative. Illinois Journal of Mathematics, 44(2), pp.407-437.

\bibitem[GF93]{GF93}
Gabbiani, F. and Fr\"ohlich, J., 1993. Operator algebras and conformal field theory. Communications in mathematical physics, 155(3), pp.569-640.

\bibitem[GL96]{GL96}
Guido, D. and Longo, R., 1996. The conformal spin and statistics theorem. Communications in Mathematical Physics, 181(1), pp.11-35.

\bibitem[Gui19]{Gui19}
Gui, B., 2019. Unitarity of the modular tensor categories associated to unitary vertex operator algebras, I,  Comm. Math. Phys., 366(1), pp.333-396. 

\bibitem[Gui21a]{Gui21a}
Gui, B., 2018. Categorical extensions of conformal nets, Comm. Math. Phys., 383, 763-839 (2021).

\bibitem[Gui21b]{Gui21b}
Gui, B., 2019. Q-systems and extensions of completely unitary vertex operator algebras. arXiv preprint arXiv:1908.03484. To appear in IMRN.

\bibitem[Hen19]{Hen19}
Henriques, A., 2019. H. Loop groups and diffeomorphism groups of the circle as colimits. Communications in Mathematical Physics,  Volume 366, Issue 2, pp 537-565

\bibitem[J\"or\ss96]{Joer96}
J\"or\ss, M., 1996. The construction of pointlike localized charged fields from conformal Haag—Kastler nets. Letters in Mathematical Physics, 38(3), pp.257-274.

\bibitem[KL04]{KL04}
Kawahigashi, Y. and Longo, R., 2004. Classification of local conformal nets. Case $c<1$. Annals of mathematics, pp.493-522.


\bibitem[LR95]{LR95}
Longo, R. and Rehren, K.H., 1995. Nets of subfactors. Reviews in Mathematical Physics, 7(04), pp.567-597.


\bibitem[LR97]{LR97}
Longo, R. and Roberts, J.E., 1997. A theory of dimension. K-theory, 11(2), pp.103-159.

\bibitem[LR04]{LR04}
Longo, R. and Rehren, K.H., 2004. Local fields in boundary conformal QFT. Reviews in Mathematical Physics, 16(07), pp.909-960.


\bibitem[Lon94]{Lon94}
Longo, R., 1994. A duality for Hopf algebras and for subfactors. I. Communications in mathematical physics, 159(1), pp.133-150.

\bibitem[Lon97]{Lon97}
Longo, R., 1997. An Analogue of the Kac-Wakimoto Formula and Black Hole Conditional Entropy. Communications in mathematical physics, 186(2), pp.451-479.

\bibitem[Lon08]{Lon08}
Longo, R.: Lectures on Conformal Nets. Preliminary lecture notes that are available at http://www.mat.uniroma2.it/~longo. The first part is published as follows: Longo, R.: Real Hilbert subspaces, modular theory, $SL(2, R)$ and CFT. In: ``Von Neumann algebras in Sibiu", Theta Series Adv. Math. 10, Bucharest: Theta, 2008, pp. 33–91

\bibitem[Mas97]{Mas97}
Masuda, T., 1997. An analogue of Longo's canonical endomorphism for bimodule theory and its application to asymptotic inclusions. International Journal of Mathematics, 8(02), pp.249-265.


\bibitem[M\"ug00]{Mueg00}
M\"uger, M., 2000. Galois theory for braided tensor categories and the modular closure. Advances in Mathematics, 150(2), pp.151-201.



\bibitem[Nel59]{Nel59}
Nelson, E., 1959. Analytic vectors. Annals of Mathematics, pp.572-615.


\bibitem[Tak70]{Tak70}
Takesaki, M., 1970. Tomita's theory of modular Hilbert algebras and its applications (Vol. 128). Springer.

\bibitem[Tak02]{Tak02}
Takesaki, M., 2002. Theory of operator algebras II (Vol. 125). Springer Science \& Business Media.

\bibitem[Wei06]{Wei06}
Weiner, M., 2006. Conformal covariance and positivity of energy in charged sectors. Communications in mathematical physics, 265(2), pp.493-506.



\bibitem[Yam04]{Yam04}
Yamagami, S., 2004. Frobenius duality in C*-tensor categories. Journal of Operator Theory, pp.3-20.





	

	
\end{thebibliography}
\end{document}